\documentclass[10pt]{article}
\usepackage{amsmath, amsthm, amssymb}
\usepackage{mathrsfs}
\usepackage[table]{xcolor}
\usepackage{geometry}
\usepackage{multicol}
\usepackage{tabularx}
\usepackage{graphicx}
\usepackage{tikz}
\usetikzlibrary{shapes.geometric, arrows.meta, positioning, calc, backgrounds, shadows, fit,shapes.multipart}
\usepackage{fontawesome5}

\title{  Achieving Generational Peace in Mali through Intergenerational Mean-Field-Type Game-based Incentives}
 \author{Hamidou Tembine\thanks{ H. Tembine is with the Department of Electrical Engineering and Computer  Science, School of Engineering, UQTR, Quebec, Canada.
 He is also  with Learning and Game Theory Laboratory, TIMADIE.
Email to : tembine(at)ieee.org }
 }
\date{12 - 12 - 2024 }

\theoremstyle{plain}
\newtheorem{theorem}{Theorem}
\newtheorem{result}{Our Main Result }
\newtheorem{lemma}{Lemma}
\newtheorem{definition}{Definition}
\newtheorem{example}{Example}

\newtheorem{remark}{Remark}

\newtheorem{corollary}{Corollary}
\theoremstyle{hyperref}

\usepackage{enumitem}
\usepackage{mathtools}
\usepackage{algorithm}
\usepackage{algpseudocode}

\definecolor{StableGreen}{HTML}{375623}
\definecolor{CriticalRed}{HTML}{C00000}
\definecolor{ActionBlue}{HTML}{1F4E79}
\definecolor{StressGray}{HTML}{767171}

\begin{document}
\maketitle

\begin{abstract} 

This article  develops an intergenerational mean-field-type game (MFTG) to  model Mali's and neighbouring countries  multi-actor conflict ecosystem, which includes formal state forces, traditional hunters, nonstate militias, jihadists, criminal networks, civil societies, and international proxies. Each decision-maker (agent, a group of agents or representative agent) is defined by a type,  state, information structure, and action, with payoffs dependent not only on individual decisions but also on the evolving distribution of all agents' profiles. The model reveals that cycles of violence can persist across multiple generations due to the embedded presence of retaliatory types such as revenger child-soldiers whose trauma-conditioned best-responses favor conflict, and whose behavior reinforces intergenerational transmission of violence. The model also captures the strategic exploitation of institutional fragility by {\it war entrepreneurs} who profit from sustained instability through arms sales, militia contracting, and unregistered market mediation. These actors inject minimal resources to trigger profitable escalations, turning latent tensions into self-reinforcing violence economies. We show that in the absence of structural counterincentives, peaceful strategies are non-absorbing, and violence remains dynamically rewarding for war entrepreneurs. However, by embedding incentive-compatible, information-adaptive transfers directly into instantaneous payoffs, rewarding verifiable peacebuilding and penalizing aggression, it is possible to shift the mean-field-type equilibrium distribution intergenerationally toward more peaceful types and drive systemic de-escalation. We also discuss about the funding and the real implementation of such mechanisms in the field. 
\end{abstract} 

\newpage 
\tableofcontents
\newpage 

\section{Introduction} 
Conventional models of armed conflict often focus on immediate military or political actors, one-shot or short-term equilibria, or bounded tactical rationality. While valuable for capturing localized confrontations and short-term decision dynamics, these models often fail to account for the recursive, multi-layered, and temporally extended mechanisms by which violence becomes entrenched, particularly in fragile or postcolonial states. In the case of Mali, persistent insecurity is not solely driven by strategic competition between state and non-state actors, but by deeply embedded and intergenerational forces: economic situations, inherited grievances, trauma-induced behavioral priors, revenge-driven memory, and social narratives of injustice \cite{turner2004political}. These factors are reinforced by structural fragilities, such as the erosion of trust in institutions, localized governance vacuums, and overlapping authority networks \cite{mortimer1972federalism,boyd1979african,bell1972age,cisse1980sedentarization}. To address this situation, we propose an intergenerational mean-field-type game (MFTG). MFTGs provide a  way to model strategic interaction among heterogeneous agents whose decisions depend not only on their individual state and private information but also on the  distribution of all agents’ states, actions, and informational signals. Each agent is described by a type (encoding identity, role, or strategic inclination), a dynamic state (trauma level, exposure to violence, social trust), a set of feasible actions (reconciliation, retaliation, disengagement), and private information (localized experiences, memory of betrayal, ideological commitment). The types evolve across generations through a probabilistic transition kernel influenced by past behavior and observed population patterns, allowing the model to capture intergenerational transmission of norms, fears, and behavioral strategies, such as those passed down from children who witnessed their parents killed by rival groups or state forces. 
Mali's and in sahelian areas abroad current conflict ecosystem is shaped by interactions among state forces, jihadist insurgencies, ethnic militias, self-defense groups, criminal networks, and foreign interventions \cite{welz2022african,toaldo2012origins,gazeley2022strong}. Actors operate under overlapping jurisdictions and ambiguous loyalties, often with asymmetric access to information and power. At the local level, violence is frequently triggered not by military necessity but by revenge, perceived injustice, and inherited trauma, all of which persist across generational boundaries \cite{benjaminsen2008does}. Repeated local peace accords, local ceasefires, and military campaigns have not led to durable stabilization, pointing to the need for a framework that captures these deeper dynamics of conflict reproduction.

\begin{figure}[htb]
\includegraphics[scale=0.4]{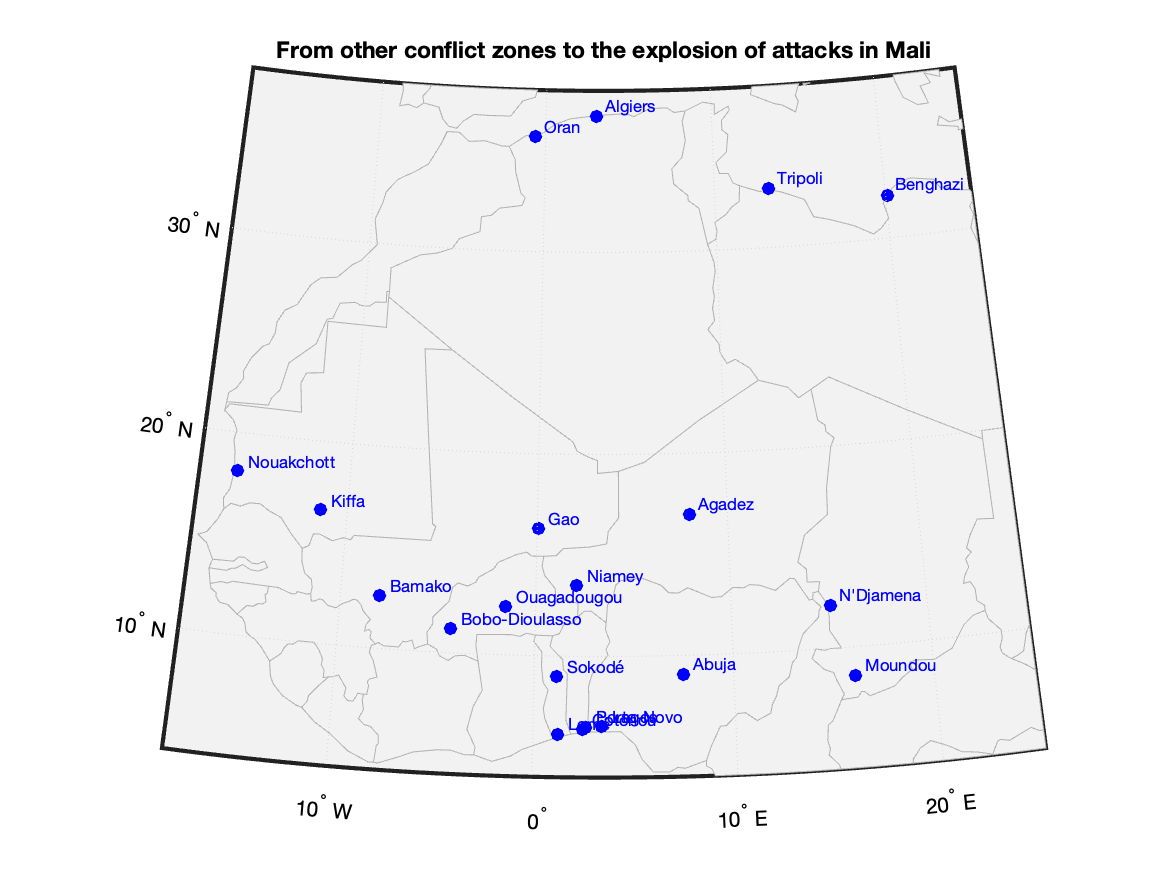}
\includegraphics[scale=0.4]{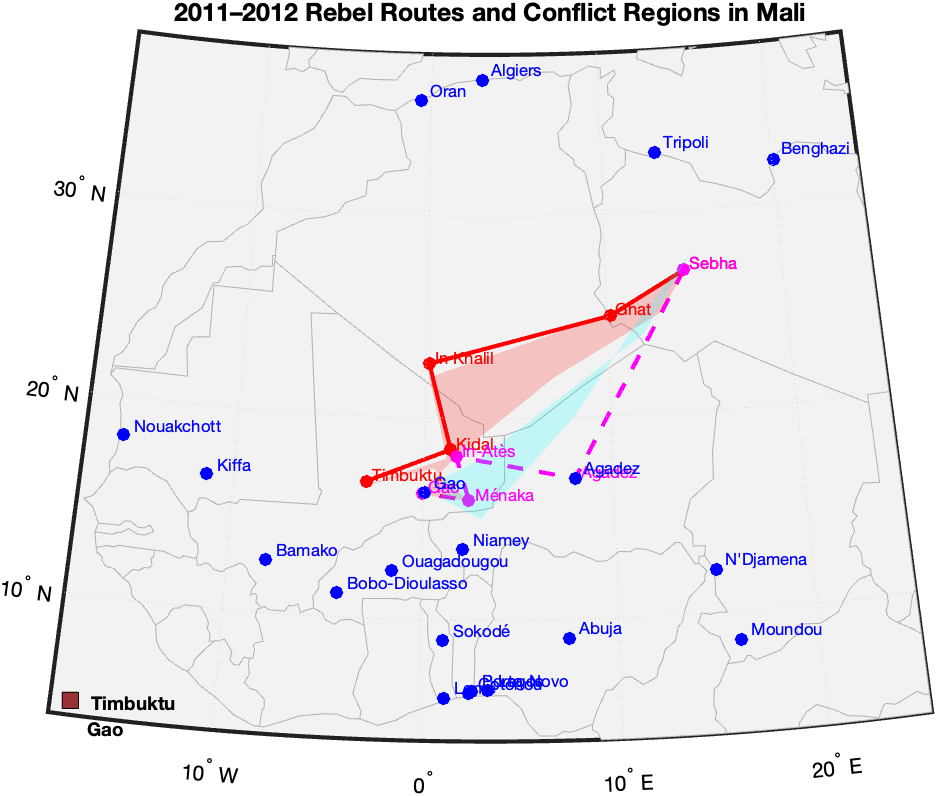} 
\caption{How conflicts in other territories have affected an already fragile situation in Mali.}
\label{initfig0}
\end{figure}

Before examining about the roots of the conflict propagation in Mali \cite{besenyHo2013war, keys2013mali, lecocq2013one, marchal2013mali, gowan2013mali, nijenhuis2013farmers, marchal2013military, cline2013nomads, logan2013roots, themner2013armed, luengo2012symptoms}, we may  first zoom on the roots  of the conflict in Libya   which can be seen as  a wave of recent roots (or accelerators) for the Malian conflict \cite{stewart2013next}.  We refer the readers to \cite{paoletti2011libya, romanet2019deep, musa2024historical, daniel2021understanding, pargeter2012libya, mckinney2012illegal, toaldo2012origins}  for the roots of the conflict  in Libya.  The roots of Mali's conflict trace back more than two centuries, beginning with precolonial empires such as the Bambara, Senufo, Toucouleur, and Songhay confederacies, which established enduring hierarchies of domination, enslavement, and regional exclusion \cite{keys2013mali,lecocq2013one}. French colonial rule (1890s-1960) institutionalized these inequalities through indirect governance, reinforcing ethnic stratification, privileging southern elites, and entrenching territorial fragmentation \cite{besenyHo2013war,marchal2013mali,boyd1979african}. Post-independence state formation further centralized power in Bamako. Most investment in the northern, western and central part of the country have not implemented mainly to due corruption, systematically marginalizing central, western, northern, farmers and nomadic populations and failing to integrate them into national co-advancement or security frameworks. This structural neglect triggered recurring rebellions led by a small minority of Tuareg, in 1963, 1990, 2006, and 2012 \cite{keita1998conflict}, fueled by broken peace accords, underdevelopment, and deep mistrust. Successive regimes responded with co-optation or militarized suppression without addressing underlying grievances.   Although Mali shares {\bf no border with Libya} (see Figure  \ref{initfig0}), the 2011 collapse of the Libyan state had a direct and destabilizing effect on Mali due to long-standing cross-Saharan military, ethnic, and logistical linkages \cite{luengo2012symptoms}. Thousands of Tuareg and other ethnic fighters from Mali and Niger, previously integrated into Gaddafi’s security brigades, returned home after Libya's fall, heavily armed, economically dislocated, and ideologically disillusioned \cite{besenyHo2013war,keys2013mali,benjaminsen2008does}. Many formed the backbone of the 2012 rebellion under the National Movement for the Liberation of Azawad (MNLA). 
Northern Mali, especially the Kidal, Gao and Timbuktu regions, offered natural rear bases, limited state presence, and highly porous borders, making it a low-cost destination for regrouping \cite{lecocq2013one,marchal2013mali}. In contrast, Algeria's post-civil war security apparatus and Niger’s early alignment with Western counterterrorism operations including U.S. drone deployments in Agadez via ``Trans-Sahara Counterterrorism Partnership" and France’s Operation Serval then Barkhane, enabled those countries to contain most spillovers. Mali, meanwhile, was destabilized by a military coup in March 2012, which paralyzed state response and triggered the suspension of Western security assistance due to constitutional constraints. The influx of fighters and weapons into Mali occurred through long-established Tuareg trade and smuggling corridors running from southern Libya into northern/western Niger, then across the Tamesna desert into Mali through under-monitored passages such as In-Atès and I-n-Khalil \cite{pargeter2012libya,mckinney2012illegal,gowan2013mali,gowan2013mali}. Niger’s logistical challenges and vast desert terrain made it difficult to fully intercept these flows, despite receiving substantial intelligence support. The conflict propagated not via direct geographic proximity, but through political and institutional vulnerability \cite{romanet2019deep}. At the time, Mali was uniquely exposed: facing a resurgent rebellion, fragmented governance, inflamed ethnic tensions between herders and farmers in central, eastern, western and southern part of the country, and vast ungoverned west and northern zones. These features made Mali a strategic focal point for jihadist expansion, arms trafficking, and proxy conflict, as it was both more penetrable and more profitable to destabilize than neighboring states with stronger institutions or support. 
This is not just about geography: it diffuses opportunistically toward the weakest institutional and societal links \cite{anderson2025rethinking}. In early 2012, Mali presented a uniquely vulnerable convergence of such factors: {\it a long-simmering  rebellion seeking renewed momentum; a weak, fragmented state undergoing political crisis and coup; localized ethnic tensions, particularly between pastoralist and agrarian groups (West, South, Central, East, North); and vast ungoverned spaces in the north with minimal border surveillance}. These conditions made Mali a low-cost, high-impact target for weaponized disorder. For arms traffickers and militant entrepreneurs, it was far easier and more profitable to ignite and sustain conflict in Mali than in neighboring Algeria with its hardened borders and intelligence apparatus, or even Niger, which had already become a hub of Western-supported counterterrorism operations \cite{marchal2013military,logan2013roots}. Thus, Mali became the conflict’s primary receiver not because it bordered Libya, but because it was strategically and structurally the most exploitable terrain in the aftermath of the Libyan collapse.
The convergence of these dynamics created a security vacuum quickly filled by rebels, jihadists, and transnational criminal networks. Over  the past 13 years, external military interventions \cite{uvere2025role,spearin2024russia,trunov2023modeling,pokalova2023wagner,spearin2025russia,turhan2021turkey,lanteigne2019china,benabdallah2020development,welz2022african}, from French-led operations to United Nations peacekeeping, from Wagner and to Afrika Korps private security, often reinforced factional divides rather than resolving root causes. Today, the conflict in Mali reflects a self-reinforcing system shaped by historical exclusion, inherited trauma, identity-based grievances, and opportunistic violence, an ecosystem of instability continually regenerated by the interplay of local weaknesses and transnational pressures \cite{shahid2025conflict,ahmed2025irrational,de2025aren}.

It is worth mentioning that the conflict between small agropastoralists and sedentary small farmers is not confined to the Dogon country; it is a widespread and structurally embedded phenomenon across much of West and Central Africa \cite{nijenhuis2013farmers,cline2013nomads,kafando2025effects,brottem2025violent,vinke2025conflicts,bontogho2025natural,eberle2025heat,ayeb2025we,mustafa2025analysing,ngalim2025managing,khorsandi2025prospects,yeboah2025farmer,moritz2006changing,moritz2006changing,turner2004political,brottem2014hosts}. The analysis of data collected by Timadie and its platforms like Guinaga, WETE, Grabal, SK1 Sogoloton, MoutonChain, CI4SI, Labinco  for 4 years in Mali indicates that similar tensions are present in the western and southern regions, where transhumant herders (these are entrepreneurs herders and not necessarily limited to one ethnic group) move through farming zones, especially during the dry season in the Nara, Nioro,  Niono, San and the Wassoulo area, leading to disputes over grazing rights, crop damage, and water access. These friction points are exacerbated by land pressure, mining unfairness, climate variability, and weak conflict resolution mechanisms. Beyond Mali, such agro-pastoral conflicts are prevalent in Burkina Faso, Niger, Nigeria, Benin, Guinea, Chad, and other neighbouring countries. In these regions, seasonal migration routes intersect with expanding agricultural frontiers, while state withdrawal, ethnic polarization, and the proliferation of small arms have transformed local disputes into violent clashes. In some contexts, these tensions have been co-opted by extremist  groups, politicized by elites, or intensified by competition over customary land tenure systems.

A war entrepreneur is a strategic actor (state, private military company, militia network, or hybrid organization) whose objective function is explicitly structured so that  armed conflict, insecurity, or regime fragility are not merely constraints but revenue-generating state variables. In MFTG, this means the agent $i$ maximizes a payoff that depends on its controls and on the joint distribution  of violence, protection demand, and political stability, with the key property that the marginal payoff with respect to certain moments, notably the level or variance of insecurity is non-negative over relevant regions so that persistent, managed instability is economically optimal. Unlike conventional security providers whose objective is to minimize threats to zero, the war entrepreneur endogenizes conflict as an asset, extracting rents through mechanisms such as regime  contracts, access to natural resources, coercive taxation, or geopolitical leverage, while calibrating the intensity and distribution of violence to avoid both total collapse (which destroys the market) and full stabilization (which eliminates demand), thereby generating a self-sustaining equilibrium of "profitable insecurity"  in which the agent’s strategic behavior actively shapes and benefits from the mean-field-type dynamics of conflict.

War entrepreneurs and arms traffickers have strategically exploited these weakest institutional and societal links such as unresolved agropastoral tensions, localized grievances, and ungoverned rural zones, to escalate and propagate conflict within Mali and across the Sahel. In contexts where the state is absent in this immensely vast territory  (Mali is a very big country  territory-wise of an area over 1.241.238 square kilometers and with seven bordering countries) or mistrusted, these actors serve as informal suppliers of violence, offering weapons, tactical training, and logistical support to communities in conflict or to armed factions seeking leverage. War entrepreneurs take advantage of the non-regulated breeding and herder movement in the country. 
They have created some disputes across several areas, for example, are transformed into militarized inter-ethnic or inter-tribal  conflicts when external arms suppliers flood the area with low-cost firearms, often on credit or in exchange for control over smuggling routes. War entrepreneurs deliberately target areas of latent tension such as regions marked by revenge killings \cite{kapil1966conflict} or land competition, knowing that minimal inputs (a few rifles or motorbikes) can ignite sustained cycles of retaliation. The war entrepreneurs also take advantage of the diaspora members who are not aware of the real situation and who fund some of the arms in the villages. As each village should its own list of arms and guns, there is an explosion of demand. 
In turn, war entrepreneurs profit from the continued demand for arms, ammunition, and mediating services, creating self-reinforcing violence economies \cite{cisse1980sedentarization}. This dynamic not only deepens polarization within communities but also accelerates the fragmentation of local authority, as armed actors replace traditional mediators. The result is a conflict in which the violence is not just the outcome of structural fragility, but also the product of deliberate market incentives to perpetuate insecurity for profit \cite{ajala2025new,le2010property,mefflassoux1970class,lewis1978small,kurtz1970political}.

In the context of Mali, a war entrepreneur  for a Private Military Company (PMC) is a strategic actor that internalizes the regime's survival probability as a primary revenue-generating asset, treating the nationwide distribution of insecurity as a productive field to be harvested rather than a threat to be eliminated. Within  MFTG, the entrepreneur optimizes a coupled payoff functional (integrating direct security rents, mining concessions, and political leverage) by ensuring that the "field" of the cycle of violence remains high enough to necessitate their presence but low enough to prevent the junta's collapse. This creates a structural incentive to engineer a regime-guarding equilibrium  where the entrepreneur's control actions intentionally maintain a positive variance of instability; by preventing total regime failure while allowing peripheral violence to persist, the war entrepreneur transforms the regime’s existential fear into a permanent curve cycle of dependency, ensuring that full pacification remains sub-optimal for the contractor’s long-term profit. 

From this MFTG perspective, Afrika Korps is also a war entrepreneur in Mali because its strategy is not to minimize the total amount of violence, but to optimize the distribution of that violence  to maximize its own political and economic payoffs. In Mali, they have successfully engineered a state where controlled instability is the most profitable outcome for the agent, leading to a permanent presence. Afrika Korps in Mali optimizes a payoff functional that extends far beyond classical military victory and is explicitly coupled to the joint law insecurity distribution: its rewards are jointly driven by the survival probability of the state authorities in Bamako (which guarantees contract continuity), the density of resource extraction opportunities (notably access to gold zones), and the persistence of a nonzero threat level that rationalizes its ongoing deployment, implying that the objective is not to solve the conflict but to sustain a trajectory of insecurity distribution that maximizes intertemporal rents.

Accordingly, Afrika Korps manages the ``field" of insecurity rather than eliminating it, treating the spatial-temporal distribution of violence as a controllable but not suppressible variable, where operations that displace insurgents from northern zones like Kidal induce diffusion into western and central regions or across borders, and where the contractor's optimization problem selects a nonzero steady-state variance of the distribution of insecurity since a zero-variance (fully pacified) field collapses demand for its services, while excessive variance (state breakdown) destroys the contractual and extractive base thereby targeting an interior optimum of instability; this interacts with a principal-agent distortion in which the principal (the Malian state) minimizes coup or capital-capture risk while the agent (Afrika Korps) controls coercive intensity, training, and frontline allocation under asymmetric information, enabling it to re-engineer the distribution by securing the center while allowing peripheral insecurity to persist, which locks the regime into a dependency equilibrium where the contractor becomes the indispensable guarantor of elite survival; consequently, the observed outcome is not a failed intervention but a stable fixed point (or limit cycle) of the coupled dynamics, in which sovereignty is effectively co-produced and partially internalized by the contractor, yielding a hybrid sovereign-corporate equilibrium that endogenously maintains the field of insecurity within a profitable band that maximizes long-run rents (financial and mineral) for the war entrepreneur.

It is worth mentioning that Afrika Korps  is not the only war entrepreneur in Mali. There are other 
dozens of local, regional, and international actors including Secopex, all seek to leverage the conflict for their own strategic ends.

We now model the consequences in field.
The ``straw that breaks the camel’s back" represents not the root cause but the final, often underappreciated, triggering event, a seemingly small or isolated incident that, when layered upon pre-existing systemic stressors (such as chronic marginalization, weak governance, ethnic/tribal polarization, or environmental strain), pushes the local context into open violence or accelerates the transition from latent tension to active conflict. While arms proliferation and institutional collapse had long destabilized parts of northern, central, western, eastern Mali, it was the post-foreign-intervention exclusion of specific  groups 
from peace negotiations perceived as a slight by key actors, that catalyzed renewed insurgent campaigns in adjacent zones. In rural areas of central and western Mali, longstanding tensions between several communities had been managed through fragile customary arrangements. The targeted attacks some of specific villages  often served as the flashpoint for intercommunal violences. 
The withdrawal of state services  in contested communes had long weakened legitimacy, but it was often the execution of a schoolteacher or the burning of a public building that crystallized fear, prompted administrative retreat, and signaled to residents that the state could no longer offer basic protection in many areas.  In cross-border regions, jihadists and bandits infiltration had been gradually intensifying for years, but it was often a single high-casualty ambush on a patrol or the kidnapping of a  humanitarian worker  (local as well as foreign) that drew international attention and triggered either foreign intervention or local reprisals, each of which risked spiraling into broader cycles of conflict. The banning of customary militias or public denunciation of traditional actors, even when well-intentioned, functioned as the final insult to already fragile security coalitions, collapsing local resilience structures and exposing entire populations to unmediated violence. In each case, these events may appear minor or isolated in isolation, yet they accumulate atop a burdened sociopolitical ecosystem like the proverbial last straw, producing irreversible shifts in local dynamics, such as mass displacement, arms mobilization, allegiance switching, or communal polarization and thereby transforming a localized grievance into a regional contagion of instability. 
 
 What defines these moments as ``straws" is not their absolute magnitude but their symbolic and catalytic role  in tipping an already overloaded system beyond its threshold of tolerance, setting into motion a sequence of events that security actors, policymakers, and community leaders can no longer easily reverse.
 
 {\bf In many regions of Mali, the notion of state presence has long been more symbolic than substantive, with government authority often limited to a single overburdened school or administrative post located in the largest commune of a vast rural zone. 
 Beyond these minimal enclaves, entire expanses of territory, farms, forest belts, hill regions, seasonal grazing lands, and peripheral villages, have operated outside formal governance for decades. In these areas, bandits or armed actors have not emerged in opposition to the state so much as in its absence, filling vacuums of justice, protection, and resource regulation through informal or coercive structures. The geography of insecurity is thus deeply tied to the geography of neglect: armed groups embed themselves in unpatrolled woodlands, remote escarpments, and transhumance corridors far beyond the reach of state institutions, which lack the logistical capacity, personnel, and legitimacy to exert consistent control. The result is a dual reality where state visibility is concentrated in isolated pockets, while non-state actors dominate the everyday governance, economy, and security landscape in much of the surrounding territory, rendering efforts at stabilization or reform fundamentally disconnected from the lived experience of many rural populations.}

\subsection*{Why is classical multi-population mean-field games not applicable here?}
In this context,  {\bf there is no   truly ``negligible" actors}: even a single individual’s choice from a villager to arm, inform, shelter, or reconcile can tip local security dynamics, ripple through social networks, and reshape macro‐level conflict trajectories. Boundaries blur between combatants and civilians, formal institutions are weak, and revenge or protection motives run high. A herder choosing to arm his youth, a village elder granting water access to a rival clan, or a teacher organizing a peace circle all exert non‐trivial influence on whether conflict flares or abates. Because arms are light and cheap, intelligence travels fast by word of mouth, mobile phone or online platform, and social capital is both the currency and casualty of violence, each decision‐maker carries “atomic” weight: their singular action shifts local mean‐fields of trust, fear, or retaliation. In this intergenerational mean‐field framework, then, no one is truly negligible: every person is a potential pivot in the security game, capable of instigating or preventing violence, and thus must be modeled as an atomic agent whose strategic choices can dominate equilibrium outcomes. 

In classical mean-field games, one assumes an atomless continuum of identical decision-makers per type/class whose individual actions infinitesimally shift the population distribution, so that each decision‐maker is ``price‐taker" in the population mean‐field and has no measurable impact on the aggregate outcome. By contrast, our intergenerational MFTG explicitly allows every agent: farmer, herder, elder, youth, outsider or insider, to be “atomic”  or mixed in key respects: a single armed raid, a peace pact brokered by one mediator, or a diaspora‐funded arms transfer by one individual materially alters local threat levels, resource flows, and subsequent best‐responses by others. Thus each person's choice enters directly into the state transitions, the joint distributions, and hence the future payoffs of all actors. There is no vanishing‐influence limit: every decision‐maker can trigger or prevent cycles of violence and shapes the intergenerational evolution of conflict. This departure from the negligibility assumption of classical MFGs is what makes our framework both richer, able to model ``veto‐power" actors and revenge catalysts, and more realistic for contexts where institutional fragility and asymmetries of power grant outsized leverage to individual choices. Here the entire population mass trend can start from a single action of an individual.

\subsection*{ Concrete examples where multi-population mean-field games are not applicable }

Gavrilo Princip's assassination of Archduke Franz Ferdinand in Sarajevo in June 1914, a lone gunman’s shot, precipitated the diplomatic crises and alliance obligations that ignited World War I, proving that one violent choice can redraft international fault lines; When Rosa Parks refused in December 1955 to give up her seat on a Montgomery bus, her individual act of nonviolent resistance crystallized widespread discontent into the Montgomery Bus Boycott and accelerated the U.S. Civil Rights Movement, demonstrating that a single refusal can mobilize a mass coalition and alter the course of state legitimacy; Thich Quang  Duc's
 public self-immolation in June 1963 in Saigon exposed to global audiences the brutal repression of Buddhists under South Vietnam's regime and catalyzed both domestic uprising and U.S. policy reversals;  The death of reformist leader Hu Yaobang in April 1989 in Beijing inspired students to gather in Tiananmen Square, transforming private grief into a nationwide pro-democracy movement before it was violently suppressed. Mohamed Bouazizi's self-immolation in December 2010 in Sidi Bouzid, Tunisia, a solitary act of protest by a street vendor, set off a wave of demonstrations that toppled long-standing regimes across the Arab world, revealing how one person’s desperate assertion of dignity can invert entrenched power relations and reshape regional security architectures; In February 2011 a single protest by Mohammad al-Durbani at the Pearl Roundabout in Manama ignited Bahrain’s own Arab Spring-style demonstrations, underscoring that even minor acts of dissent can erupt into sustained unrest. 
 
 In each instance, a solitary, ``atomic" decision reshaped the distribution of public sentiment, trust, and collective efficacy, effects that purely non-atomic mean-field models, which assume each individual’s influence vanishes in the aggregate, would miss entirely. Capturing the security dynamics of such turning points therefore demands an intergenerational mean-field-type  in which every actor, no matter how seemingly small, can wield decisive, system-level impact.

\subsection*{Contribution}
We formalize the geospatio-temporal  environment from Libya to Niger to Algeria to Mali to Mauritania to Burkina Faso to Togo and Benin, using an intergenerational MFTG model. We show that under the presence of retaliatory or aggressive types whose best-response strategies lie in conflict-inducing action sets, and in the absence of structured counterincentives, the system admits dynamically equilibrium trajectories in which cycles of violence persist over generations. {\bf Even when temporary peace is achieved, the underlying distribution of types and inherited trauma makes a return to conflict inevitable.} To address this, we propose a mean-field-type policy mechanism based on incentive-compatible transfers that are embedded directly into the instantaneous utility of agents and conditioned on verifiable peaceful behaviors and contributions to local co-development. These incentives, implemented through observable and adaptive mechanisms, reshape the utility, making peaceful behavior a dominant strategy for a broader class of types. We  prove that under such incentive structures, the dynamic equilibrium distribution shifts over time toward peaceful types. As this shift propagates intergenerationally through the type-transition kernel, the population composition moves toward a stable, peaceful mean-field-type configuration. Our results are formalized through the analysis of Bellman and master adjoint systems (MASS \cite{tembine2024mean,tembine2023machine,basar2024foundations,basar2024applications, tapo2024machine}). We provide explicit conditions for emergence of equilibria and both a theoretical foundation and practical roadmap for breaking cycles of violence in fragile conflict ecosystems like Mali.

\subsection*{Structure}
  
The rest of this article  is structured as follows.  Section \ref{secmftgv1} formulates the problem using MFTG and present the key results.
 Section \ref{seccoi6} focuses on co-Intelligence between and MI agents.  Section \ref{seccoi7}  focuses on the limitations of all-military approach. 
 Parrondo-type Cooperation  for some Sahel- Southern Maghreb Region are examined in Section \ref{secmag}.
 Section  \ref{seccoi8} concludes the article. Table \ref{tab:actor_abbreviations} summarizes some abbreviations used in the article.

\newpage 
\begin{table}[htb]
\centering
\renewcommand{\arraystretch}{1.2}
\begin{tabularx}{\textwidth}{|l|X|}
\hline
\textbf{Abbreviation / Label} & \textbf{Actor or Entity Description} \\
\hline
UN & United Nations \\
MINUSMA & United Nations Multidimensional Integrated Stabilization Mission in Mali \\
AU & African Union \\
ECOWAS & Economic Community of West African States \\
FAMa & Malian Armed Forces \\
SE & State Intelligence Agency of Mali \\
MNLA &  National  Movement for the Liberation Azawad \\
FLA & Azawad  Liberation Front\\
CMA & Coordination des Mouvements de l'Azawad \\
HCUA & High Council for the Unity of Azawad \\
MAA & Arab Movement of Azawad \\
CSP &  Permanent Strategic Framework \\
JNIM & Jama'at Nusrat al-Islam wal-Muslimin \\
ISGS & Islamic State in the Greater Sahara \\
AQIM & Al-Qaeda in the Islamic Maghreb \\
GATIA & Imghad Tuareg Self-Defense Group and Allies\\
MSA & Movement for Salvation of Azawad \\
Antolna & Traditional hunter groups in Mali \\
PMC & Private Military Company \\
Wagner & Russian-linked PMC operating in Mali \\
TG & Turkish Guards (Turkish advisors or contractors) \\
AK & Afrika Korps (Russian-linked) \\
NGO & Non-Governmental Organization \\
\hline
\end{tabularx}
\caption{List of actor labels and their corresponding full meanings.}
\label{tab:actor_abbreviations}
\end{table}

\section{Mean-Field-Type Game} \label{secmftgv1}

The MFTG considered here involves a finite number of agents, each characterized by a type \( \theta_i \), a state \( s_i \), an action  \( a_i \), an information  \( \mathcal{I}_i \), and an order of play. The type \( \theta_i \) represents the agent’s intergenerational identity and strategic disposition such as being a government soldier, bandit, jihadist recruiter, revenger child-soldier, village elder, or peace-focused worker, capturing both inherited traits (trauma, ideology, ethnic affiliation) and persistent preferences over peace or violence. The state \( s_i \) reflects the agent’s evolving condition, including location, security exposure, trauma level, institutional trust, and access to weapons, and evolves stochastically over time through both endogenous actions and exogenous shocks. At each time, agent \( i \) selects an action \( a_i \in A_i \) based on their information \( \mathcal{I}_i \), which may include private, local, or hidden variables such as past victimization, knowledge of weapons caches, or informal ties. The action is the choice the agent makes at each point in time, ranging from engaging in conflict (attacking, extorting, joining a militia), to peaceful behavior (attending school, organizing reconciliation, participating in co-development programs), or remaining passive. The instantaneous payoff \( r_i(\theta, s, a, \mu) \) depends not only on the full type-state-action profile \( (\theta, s, a) \), but also on the distribution \( \mu \) of that profile across all agents, capturing population-level effects like conflict intensity, social fragmentation, or collective trust. This structure allows the MFTG to rigorously model the interplay between individual decisions and systemic dynamics in  conflict environments where secrecy, distrust, and identity play central roles.

\subsection{Key Decision-Makers}

\begin{center}
\rowcolors{1}{}{gray!10}
\begin{tabular}{|p{10cm}|}
\hline
\rowcolor{blue!15} \textbf{State Actors} \\
Government of Mali \\
Malian Armed Forces \\
National Security Services \\
\hline
\rowcolor{red!15} \textbf{Non-State Armed Groups} \\
Rebel Movements (MNLA, HCUA, CMA) \\
Jihadist Groups (AQIM, JNIM, ISGS) \\
Criminal Bandits \\
Ethnic Militias / Self-defense Groups \\
Traditional Hunters \\
Local Militias \\
Foreign Militias \\
\hline
\rowcolor{green!20} \textbf{Coalitions} \\
 Rebels + Jihadists \\
 Rebels + Bandits \\
 Bandits + Jihadists \\
Rebels + Hunters \\
 Regular Army + Hunters \\
 Regular Army + Foreign Militias \\
 Local Population + Hunters \\
 Militias + Rebels \\
Mixed Armed Coalitions \\
Serval+ Government of Mali \\
Barkhane + Government of Mali \\
Barkhane + Rebels \\
 MINUSMA + Government of Mali \\
 EU+ Government of Mali \\
AU + Government of Mali \\
 Ecowas + Government of Mali \\
 Wagner Group + Government of Mali \\
  Turkish Guards + Government of Mali \\
 Afrika Korps +  Government of Mali \\
Afrika Korps + Local Militias \\
 Turkish Guards + Multinational Corporations \\
Turkish Guards + Islamic NGOs \\
 Wagner +  Rebel Groups \\
 Afrika Korps + Jihadist or Rebel Groups \\
Turkish Guards + Foreign Militias \\
\hline
\hline
\end{tabular}
\end{center}

\begin{center}
\rowcolors{1}{}{gray!10}
\begin{tabular}{|p{10cm}|}
\hline

\rowcolor{yellow!20} \textbf{Civil Society} \\
Local Population \\
The Communities \\
Political Parties \\
Local Civil Society Organizations \\
Diaspora Groups \\
\hline
\rowcolor{cyan!15} \textbf{International Actors} \\
International Community \\
United Nations Peacekeeping (MINUSMA) \\
ECOWAS \\
AU Mission and Observers \\
International NGOs and Humanitarian Organizations \\
Foreign Military Forces \\
Multinational Corporations \\
\hline
\rowcolor{magenta!10} \textbf{Transnational / Proxy Forces} \\
Transnational Criminal Networks \\
Turkish Guards \\
Afrika Korps \\
Multi-country militia \\
\hline
\rowcolor{orange!15} \textbf{Hybrid / Irregular Actors} \\
Opportunistic double agents \\
Mixed agents \\
Revenger combatants \\
Defectors and dissatisfied fighters \\
Arms dealers and war-sustaining intermediaries \\
Propaganda funders \\
Conflict entrepreneurs \\
Wearing the insignia or attire of opposing groups \\
Specialized in online psychological operations \\ Intelligence service suppliers \\ 
\hline
\end{tabular}
\end{center}

The Government of Mali is the internationally recognized central authority, responsible for national sovereignty, governance, and coordination of defense efforts across ministries and agencies. The Malian Armed Forces (FAMa) serve as the primary military instrument of the state, tasked with counterinsurgency, territorial defense, and restoration of public order. National Security Services, comprising the police,  paramilitary police and  intelligence office, are charged with maintaining internal security, surveillance, and criminal investigations. Rebel Movements such as the MNLA, HCUA, MAA, FLA, CMA  and CSP emerge from ethno-regional grievances, demanding autonomy, federalism, or political inclusion through both armed struggle and negotiation. Jihadist Groups, including AQIM, JNIM, and ISGS, operate transnationally and pursue Sharia-based governance through asymmetric warfare and strategic terror. Criminal Bandits engage in opportunistic violence, kidnappings, theft, extortion, motivated by profit and often intersecting with other armed groups. Ethnic Militias and Self-defense Groups, like GATIA, Dan Na Ambassagou, and MSA, are community-anchored formations mobilized around ethnic identity and territorial defense. Traditional Hunters such as the Antolna  possess indigenous authority and spiritual legitimacy, providing localized protection based on customary norms. Local Militias are non-state armed actors formed around village defense or political patronage, operating with variable autonomy. Foreign Militias, notably Wagner Group affiliates, function as state-aligned or semi-autonomous actors with unclear legal status and significant combat capabilities.
Coalition: Rebels + Jihadists involves strategic alignment between nationalist insurgents and Salafi-jihadist actors, often driven by mutual enemies or short-term territorial control. Rebels + Bandits reflects fluid collaborations where criminal profit-seeking supports political destabilization. Bandits + Jihadists merges ideological violence with illicit economies, sustaining both actors' operations through shared logistics or mutual protection. Rebels + Hunters sees traditional forces reinforcing rebel legitimacy or territorial control. Regular Army + Hunters is a state-sanctioned collaboration leveraging local legitimacy for counterinsurgency. Regular Army + Foreign Militias incorporates contracted or proxy foreign units into national military operations, often for strategic depth. Local Population + Hunters formalizes community-backed local defense through traditional authority. Militias + Rebels combines grassroots mobilization with organized political insurgency. Mixed Armed Coalitions refer to transient military alliances across otherwise divergent armed actors, militias, rebels, and foreign proxies, built around tactical convergence.
The Local Population includes unarmed civilians, farmers, herders, and urban dwellers affected by and responding to violence through flight, cooperation, or resistance. The Communities comprise ethnic, religious, and linguistic groups whose historical grievances and identities shape conflict dynamics. Political Parties contest for power through electoral and informal means, sometimes engaging armed groups for leverage. Local Civil Society Organizations include chiefs, women’s associations, and youth groups, mediating community needs and peacebuilding. Diaspora Groups are expatriate Malians engaged in remittances, lobbying, and transnational mobilization with material or ideological impact.
The International Community includes diplomatic missions, regional organizations, and multilateral agencies seeking stabilization, governance reform, and peace negotiations. United Nations Peacekeeping (MINUSMA) is mandated to protect civilians and facilitate political dialogue but faces operational constraints. ECOWAS plays a mediating role in political transitions and sanctions enforcement. The AU Mission and Observers offer pan-African legitimacy to regional governance efforts. International NGOs and Humanitarian Organizations deliver aid, education, and healthcare, often filling state gaps in conflict zones. Foreign Military Forces, including remnants of Operation Barkhane and EU missions, support counterterrorism and state capacity building. Multinational Corporations, especially in mining and infrastructure, operate under security risk and require armed protection. Transnational Criminal Networks engage in trafficking (drugs, arms, humans), exploiting porous borders and weak institutions.
Turkish Guards which include private security firms and military advisors, are tied to Ankara's strategic and commercial interests in Mali and the Sahel. Afrika Korps is to foreign expeditionary actors backed Russia and by Middle Eastern or North African regimes, engaged in regional influence through security provision. Multi-country Militias consist of fighters trained and supported by foreign states, potentially deployed across borders with political or economic motives. Turkish Guards + Government of Mali embodies bilateral military cooperation with geostrategic undertones. Afrika Korps + Local Militias links regional power projection with community-based defense forces. Turkish Guards + Multinational Corporations provides armed protection to foreign commercial interests. Turkish Guards + Islamic NGOs enables religious soft power backed by security infrastructure. Afrika Korps + Rebel or Jihadist Groups represents covert or semi-official support to insurgents for regional leverage. Turkish Guards + Foreign Militias may reflect geopolitical counterbalancing against rival forces such as Wagner, creating layered international contestation within Mali’s borders.

Mali's security ecosystem includes also numerous informal, fluid, and often invisible participants whose roles are pivotal yet underrecognized. Among them are opportunistic double agents who navigate and exploit multiple allegiances, sometimes simultaneously reporting to rebels, government forces, or foreign actors for personal survival or gain. Mixed agents are individuals embedded across institutional boundaries, such as local officials who covertly finance militias, or fighters with one foot in formal security and another in illicit trade. Revenger combatants, many of whom were children when their families were massacred, represent a psychologically motivated group driven by trauma, vengeance, and cycles of inherited grievance. Defectors and dissatisfied fighters shift between groups, often citing unmet promises, poor leadership, or ideological disillusionment, thereby altering allegiance dynamics. Arms dealers and war-sustaining intermediaries profit from prolonged instability, covertly supplying weapons and ammunition to opposing factions to maintain demand. Propaganda funders are actors, domestic or foreign, who pay influencers, disinformation agents, and agenda-driven communicators to shape narratives, discredit rivals, or incite ethno-political divisions on social media platforms. Conflict entrepreneurs, including urban elites and transnational brokers, contract militias on a mission-specific basis for protection, extortion, or territorial enforcement. Some actors wear the insignia or attire of opposing groups, such as jihadist or military uniforms, to carry out attacks, sow confusion, and delegitimize enemies in the eyes of the population. Others specialize in online psychological operations, including false-flag campaigns, meme warfare, and digital impersonation, targeting both local and international audiences. 

\subsection{Information Structure}

The information structure across the multitude of actors in Mali's conflict  is deeply {\it asymmetric, fragmented, and conditioned by geography, alliances, digital access, and psychological factors}. The Government of Mali possesses access to formal administrative channels, diplomatic cables, and high-level intelligence via national institutions and foreign allies, yet suffers from gaps in real-time local knowledge and often lacks trustable field-level informants in contested zones. The Malian Armed Forces (FAMa) receive operational intelligence through hierarchical command and regional reports, but face structural bottlenecks, limited aerial surveillance, and infiltration risks. The National Security Services possess more localized networks of informants and technological monitoring tools, but are constrained by jurisdictional overlaps and internal rivalries. Rebel Movements like the MNLA or CMA depend on dense ethno-linguistic networks, territorial familiarity, and informal cross-border coordination, often benefiting from long-standing community allegiances. Jihadist Groups, such as JNIM or ISGS, use ideological cohesion and clandestine cells to maintain tight, compartmentalized communication structures, leveraging religious schools, couriers, and encrypted digital channels. Criminal Bandits operate with situational intelligence drawn from bribes, threats, and informal markets, exploiting corruption and the weakness of local surveillance. Ethnic Militias and Self-defense Groups like GATIA or Dan Na Ambassagou rely on oral transmission, elders’ networks, and hyper-localized rumor systems, but often lack broader strategic oversight. Traditional Hunters (Antolna, Dozo, Koglweogo) utilize animist belief systems, ancestral legitimacy, and face-to-face assemblies, privileging symbolic knowledge and territorial intuition over formal reporting. Local Militias and Foreign Militias (Wagner) maintain distinct structures: locals depend on kinship-based mobilization and informal patrols, while foreigners often rely on technical intelligence, satellite reconnaissance, or instructions from non-transparent chains of command.
In coalitions, information structures are even more convoluted. Rebel+Jihadist alliances often manage parallel information flows to avoid betrayal. Bandit+Rebel and Bandit+Jihadist coalitions share tactical data selectively, centered on logistics and targets. Rebel+Hunter and Army+Hunter coalitions transmit knowledge through intermediaries, often village elders or dual-loyalty agents, with asymmetric access to state or local sources. Army-Foreign Militia collaborations benefit from advanced technological systems (drones, signals intelligence), though trust deficits hinder integration. Local Population+Hunter coalitions thrive on mutual surveillance and trust-based oral exchanges, with a focus on immediate threat detection. Militia+Rebel and mixed armed coalitions are fluid, often lacking a stable information core, relying on shared ideology or perceived threats for coordination. The Local Population, meanwhile, accesses information primarily through oral networks, local radio, WhatsApp and Telegram groups, and mosque-based communication, often navigating misinformation, fear, and survival bias. Communities draw on historical memory, kinship networks, and religious discourse, shaping interpretations of violence and allegiance. Political Parties use partisan media, electoral polling, and elite networks, yet are vulnerable to disinformation and digital manipulation. Civil Society Organizations gather information through grassroots outreach, community forums, and humanitarian reporting mechanisms. Diaspora Groups are linked to satellite media, encrypted messaging apps, and international news cycles, but may suffer from outdated or distorted ground-level feedback.
The International Community, including diplomats and multilateral envoys, relies on embassy reports, peacekeeping briefings, and international media, often filtered through geopolitical interests. MINUSMA uses structured mission reporting, human rights monitoring, and civilian liaison teams, though access constraints and safety risks hinder full coverage. ECOWAS, AU observers, and International non-governmental organisations (NGOs) possess variable information depth, mostly depending on host government cooperation and regional offices. Foreign Military Forces (Serval, Barkhane, EU missions) have advanced intelligence capabilities including signals intelligence (SIGINT), Unmanned Aerial Vehicles (UAV), and satellite imagery, yet struggle with cultural decoding and rapid human terrain changes. Multinational Corporations operate with private security assessments, risk consultancies, and direct liaisons with local power brokers, yet are susceptible to opaque and self-interested information. Transnational Criminal Networks use closed, trusted networks with robust anti-surveillance countermeasures, often ahead of formal actors in tracking supply-demand shifts and territorial flux.
Among the informal actors, opportunistic double agents possess overlapping knowledge bases across factions, leveraging role ambiguity, selective truth-telling, and incentives from multiple sides. Mixed agents utilize hybrid access to state and non-state actors, creating opaque and self-curated information environments tailored for influence or protection. Revenger combatants filter information through trauma-induced cognitive lenses, often emotionally driven and ideologically radicalized by narratives of martyrdom or historical injustice. Defectors and dissatisfied fighters possess insider knowledge of former factions and actively disseminate intelligence when defecting, while also seeking to manipulate narratives to justify realignment. Arms dealers and war sustainers maintain privileged transactional intelligence via black markets, logistics chains, and overlapping militancy hubs. Social media propagandists and their funders shape information architecture through bot networks, fake accounts, meme production, and the strategic use of hashtags and misinformation loops. Contract militia-hirers operate via private negotiations, ephemeral contracts, and encrypted digital transactions, often obscuring information flows from third parties. False-flag actors wearing others’ insignia weaponize misidentification to generate strategic confusion and sow distrust among civilian observers. Online psy-ops actors integrate metadata scraping, linguistic mimicry, and MI-generated content to distort reality perception, targeting both Malian and international 
audiences. {\bf As we can see, in the ground, several actors have stayed in Mali only a limited time period, leaving the stage to other new actors, making the game intergenerational or block-multistage.}

\begin{table}
\rowcolors{2}{gray!10}{white}
\begin{tabular}{|p{7.3cm}|p{7.3cm}|}
\hline
\rowcolor{gray!30}
\textbf{State-Based Actions} & \textbf{Others' Actions} \\
\hline
Legislative reform and peace accord negotiation (Government of Mali) &
Armed insurgency and territorial occupation (Rebels) \\
\hline
Deployment of state resources and shifting between conciliation and repression &
Targeted assassinations, IED attacks, and forced taxation (Jihadist Groups) \\
\hline
Patrols, search-and-destroy operations, joint actions with proxies (FAMa) &
Kidnapping-for-ransom, highway robbery, extortion (Criminal Bandits) \\
\hline
Urban and rural security operations &
Communal retribution and ethnic cleansing (Ethnic Militias / Self-defense Groups) \\
\hline 
Intelligence gathering, arrests, border surveillance, counterterrorism  &
Armed village defense, mystical rituals, symbolic patrols (Traditional Hunters) \\
\hline
Ceasefire monitoring, early warning, civilian escort (MINUSMA) &
Checkpoint enforcement, property protection, and raids (Local Militias) \\
\hline
Sanctions, diplomatic mediation, and funding of stabilization (International Community) &
Kinetic operations, assassinations, and intelligence work (Foreign Militias) \\
\hline
Drone strikes, intelligence-sharing, force training (Foreign Military Forces) &
Joint attacks, propaganda coordination (Rebels + Jihadists Coalition) \\
\hline
Humanitarian corridor negotiation, peacebuilding reports (Civil Society Orgs, NGOs) &
Shared logistics, financial exchanges (Rebels + Bandits, Bandits + Jihadists) \\
\hline
Base protection, force projection, regional mediation (AU, ECOWAS, MINUSMA) &
Hybrid commands and rotational attacks (Militias + Rebels, Mixed Coalitions) \\
\hline
Food, medical aid, education in conflict zones (International NGOs) &
Survival strategies: displacement, silence, parallel cooperation (Local Population) \\
\hline
Lobbying, diaspora funding, political mobilization (Diaspora Groups) &
Communal petitions, inter-communal dialogues, traditional justice (Communities) \\
\hline
Rallies, alliances with armed groups, electoral propaganda (Political Parties) &
Targeted retribution, radicalization (Revenger Combatants) \\
\hline
Institutional advocacy, public pressure (Local Civil Society) &
Splinter formation, insider leaks (Defectors, Dissatisfied Fighters) \\
\hline
Hiring of private security, infrastructure lobbying (Multinational Corporations) &
Arms trade, Unregistered market auctions (Arms Dealers) \\
\hline
Support of information warfare, propaganda campaigns (None-state propagandists) &
False-flag operations, psychological manipulation (False-flag operatives, Psy-ops agents) \\
\hline
\end{tabular}
\end{table}

\subsection{Actions}
The Government of Mali undertakes legislative action, peace accord negotiation, and deployment of state resources in conflict zones, yet often oscillates between conciliation and repression depending on internal political pressures. The Malian Armed Forces  conduct patrols, search-and-destroy missions, urban and rural security operations, and are increasingly engaged in joint actions with allied or proxy forces; however, they are also accused of excessive use of force, extrajudicial killings, and ethnically biased operations. The National Security Services, are responsible for arrests, border surveillance, intelligence operations, and counterterror investigations, yet often face allegations of arbitrary detention or complicity with certain militias.  Here each actors can hire local, regional, continental, and international intelligence gathering and intelligence services for drones,  satellite, night cameras, underwater, ground, air as part of strategic investment to be  exploited. The nation can select several intelligence services to extract useful strategic information. 
Rebel Movements  engage in armed insurgency, negotiations, strategic territorial occupation, and international diplomacy, frequently alternating between formal peace processes and covert mobilization. Jihadist Groups execute targeted assassinations, ambushes, vehicle-borne IEDs, community infiltrations, forced taxation, religious proselytism, and the strategic use of terror to undermine state legitimacy. Criminal Bandits pursue kidnapping-for-ransom, highway robbery, extortion, cattle rustling, and collusion with both insurgents and state elements for tactical shelter. Ethnic Militias and Self-defense Groups act as localized enforcers, mobilizing patrols, protecting ethnic homelands, and engaging in communal retribution, while often partaking in or resisting ethnic cleansing. Traditional Hunters conduct armed village defense, mystical rituals of protection, targeted retributive justice, and symbolic authority assertion through visible patrolling. Local Militias perform checkpoint enforcement, property protection, and occasional offensive raids depending on sponsor alignment. Foreign Militias, such as Wagner or Chadian auxiliaries, undertake kinetic operations, assassinations, intelligence gathering, and political enforcement in favor of their sponsor states. Coalitions display actions shaped by tactical convergence. The Rebels+Jihadists coalition launches joint attacks, shares territory, and coordinates propaganda; the Rebels+Bandits axis leverages rebel protection for bandit operations in exchange for financial or logistical support. The Bandits+Jihadists alliance supports black market taxation and civilian intimidation. Rebels+Hunters engage in localized campaigns and community-based legitimacy consolidation. The Army+Hunters coalition conducts joint rural patrols and village sweeps. Army+Foreign Militias are active in high-impact strikes, site protection, and targeted extra-judicial violence. Local Population+Hunters engage in local surveillance, conflict mediation, and group defense formations. Militias+Rebels often create hybrid command structures and rotational engagements in ethnic flashpoints. Mixed Armed Coalitions pursue adaptive, temporary control over roads, mines, or towns, often followed by fragmentation. The Local Population adopts a wide array of survival actions: flight, strategic silence, compliance with multiple authorities, clandestine documentation of abuses, and displacement. Communities mobilize petitions, inter-communal dialogue, traditional justice mechanisms, or in some cases, communal militancy. Political Parties organize rallies, broker alliances with armed groups, disseminate electoral propaganda, and engage in legislative advocacy or obstruction. Civil Society Organizations run peacebuilding workshops, publish human rights reports, lobby international bodies, and serve as intermediaries between civilians and combatants. Diaspora Groups fund militias, lobby foreign governments, and run transnational campaigns of identity solidarity or political opposition.
The International Community, including diplomats and envoys, facilitates negotiations, issues sanctions or statements, and funds stabilization programs. MINUSMA undertakes base protection, civilian escort, ceasefire monitoring, and early warning dissemination, often under threat. ECOWAS imposes political pressure, offers mediation, and coordinates regional sanctions. The AU and its observers conduct visits, write assessments, and issue normative declarations. International NGOs provide food aid, medical assistance, and conflict-sensitive education, while negotiating humanitarian corridors with armed actors. Foreign Military Forces (France’s Serval, Barkhane, EU trainers, Afrika Korps trainers) engage in drone strikes, intelligence-sharing, local force training, and direct tactical support. Multinational Corporations hire private security, extract minerals, build infrastructure, and lobby for regulatory concessions. Transnational Criminal Networks traffic weapons, narcotics, and humans, bribe border officials, and coordinate logistics across Mali’s porous frontiers. Among the covert actors, opportunistic double agents pass false intelligence, infiltrate opposing groups, and manipulate information to serve multiple patrons. Mixed agents coordinate sabotage, distribute patronage, and perform political brokering while maintaining deniability. Revenger combatants seek personal retribution through enlistment in violent factions, selective assassination, or radicalization. Defectors and dissatisfied fighters reveal secrets, form splinter groups, or become mercenaries. Arms dealers buy, store, and sell weapons covertly to multiple factions, often orchestrating unregistered market  auctions. Propaganda funders pay influencers, launch disinformation campaigns, and coordinate cyber operations through anonymous media. Mission-based militia contractors engage in targeted raids, property destruction, or temporary area control before dissolving or shifting allegiance. False-flag operatives commit atrocities in the uniforms of rivals to erode public trust. Psy-ops specialists generate viral content, deploy deepfakes, and simulate mass sentiment online to influence both domestic and global narratives. These actors animate a conflict system defined not by fixed frontlines but by dynamic, strategic actions layered with ambiguity, contestation, and psychological warfare.

\subsection{Resources}
Each agent type in Mali’s conflict ecosystem operates with a distinct resource endowment shaped by institutional access, strategic alliances, geographic control, and informal networks. The Government of Mali, FAMa, and National Security Services possess formal legal authority, conventional weapons, military bases, international funding, intelligence infrastructure, and limited air capabilities. Rebel movements control territory, light arms, local legitimacy in specific regions, and exploit cross-border ethnic ties. Jihadist groups have access to ideological networks, asymmetric warfare tactics, extortion revenue, foreign fighters, and covert supply chains. Criminal bandits exploit terrain knowledge, motorbikes, firearms, and kidnap-for-ransom economies. Ethnic militias and traditional hunters rely on local mobilization, artisanal arms, cultural legitimacy, and knowledge of rural terrain. Local and foreign militias, including Wagner or Turkish-trained units, operate with paramilitary gear, external sponsors, logistical support, and intelligence sharing. Coalitions draw resources from their constituent groups, often merging arms, fighters, intelligence, and propaganda tools. Local populations and communities hold land, information, and legitimacy but limited mobility and protection. Political parties and civil society groups possess social capital, narrative influence, and access to diplomatic channels. Diaspora groups contribute remittances, lobbying influence, and sometimes ideological narratives. International actors, MINUSMA, ECOWAS, AU, foreign militaries, NGOs, hold peacekeeping mandates, humanitarian resources, surveillance tools, and diplomatic leverage. Multinational corporations possess financial capital, security contracts, and logistical assets. Transnational criminal networks leverage global trafficking routes, illicit funds, and corrupt intermediaries. Hybrid agents (double agents, mercenaries, propagandists, arms dealers) draw on multi-source funding, false identities, disinformation channels, and covert logistics. Intergenerational revengers and allegiance-shifting individuals wield narrative legitimacy, emotional drive, insider knowledge, and social media access. False-flag actors and mission-based militias use uniforms, misdirection, and tactical deception to manipulate perception and control local dynamics. Collectively, this fragmented but interlinked resource structure defines a complex, asymmetric battlefield shaped as much by information and narrative as by guns and territory.
\subsection{The Sources of the Resources}
The sources of resources for each actor type in Mali’s conflict are diverse, often opaque, and shaped by both formal and illicit channels. The Government of Mali, FAMa, and National Security Services receive state budget allocations, international military aid (notably from France, EU, and the U.S.), and training programs, but their effectiveness fluctuates due to corruption, sanctions, and shifting alliances. Rebel movements derive resources from cross-border smuggling, diaspora donations, kidnapping ransoms, and political patronage, with funding often surging after peace negotiations or ceasefire breakdowns. Jihadist groups tap into regional trafficking (drugs, arms, humans), foreign jihadist funding (especially via Gulf or Maghreb networks), and local taxation in controlled areas, with increases tied to state withdrawal or foreign troop exits. Bandits depend on ransoms, robbery, and weapon access through porous borders and black markets. Ethnic militias, hunters, and local militias are locally equipped, sometimes by village contributions or sympathetic elites, but also receive covert arms from state actors or foreign sponsors during escalations. Foreign militias, including the Wagner Group or Turkish-aligned units, are funded by geopolitical patrons through mineral concessions, security contracts, or opaque transfers. Coalitions pool their resources, financial, human, and informational, through pragmatic or opportunistic alliances. Communities, civil society, and diaspora groups rely on remittances, NGO support, and international lobbying, which can increase or decline based on global attention and crisis fatigue. International actors are funded by multilateral institutions and donor states, with changing mandates and budgets influenced by regional security trends and diplomatic shifts. Multinational corporations self-fund through extractive profits, but depend on hiring militias or private security to sustain operations. Transnational criminal networks are financed through global trafficking revenues and operate with impunity when border controls weaken. Hybrid actors, including arms dealers, propagandists, and double agents, are opportunistic and funded indirectly via conflict economies, mission payments, or covert intelligence budgets.

\begin{table}
\rowcolors{2}{gray!10}{white}
\begin{tabular}{|p{7.3cm}|p{7.3cm}|}
\hline
\rowcolor{gray!30}
\textbf{State-Based Actors} & \textbf{Other Actors (Non-State, Hybrid, Covert)} \\
\hline
\textbf{Payoffs:} Political legitimacy, territorial control, restored authority, international recognition, donor funding, and strategic alliances &
\textbf{Payoffs:} Resource extraction, local dominance, revenue from illicit trade, ideological propagation, revenge satisfaction, social prestige \\
\hline
\textbf{Costs:} Operational expenses, loss of personnel, political backlash, reputational damage (from abuses), corruption-driven inefficiencies &
\textbf{Costs:} Exposure to counterstrikes, intra-group rivalry, betrayal risk, fatigue, dependence on unstable funding (extortion, smuggling) \\
\hline
\textbf{Risks:} Coup risk, mission failure, international sanction, civilian casualties eroding support, infiltration by hostile elements &
\textbf{Risks:} Targeted elimination, loss of territorial base, internal splintering, exposure by defectors, cyber or drone surveillance \\
\hline
\textbf{Outcomes:} Partial reassertion of control in urban centers, fragile ceasefires, donor-driven stabilization, but recurrent legitimacy crises &
\textbf{Outcomes:} Localized control through fear or patronage, ongoing recruitment, informal governance structures, but frequent fragmentation \\
\hline
\textbf{Payoffs (specific):} Legislative gains, increased control over borders and extractive zones, reinforcement of elite networks &
\textbf{Payoffs (specific):} Ransom profits, taxation of black markets, myth-building via martyrdom narratives, transnational jihadist funding \\
\hline
\textbf{Costs (specific):} Budgetary drain from prolonged deployments, donor conditionality, strained civil-military relations &
\textbf{Costs (specific):} Arms procurement, bribes to border guards, loss of local support due to brutality, logistical overreach \\
\hline
\textbf{Risks (specific):} International legal scrutiny (ICC), fragmentation of national command, dependency on foreign troops &
\textbf{Risks (specific):} Intelligence penetration, loss of ideological cohesion, civilian resistance, drone-based targeting \\
\hline
\textbf{Outcomes (specific):} Tactical gains without strategic consolidation, regional containment without national reconciliation &
\textbf{Outcomes (specific):} Expanded shadow economies, deeper community polarization, recruitment of revenge-driven youth, but unstable gains \\
\hline
\end{tabular}
\end{table}

\subsection{Risk-Aware Payoffs}
The risk-aware payoff structures of Mali’s conflict actors are deeply heterogeneous and shaped by distinct utility functions reflecting survival, power, legitimacy, wealth, ideology, reputation, and community protection, each modulated by perceived risk, uncertainty, and adversarial behavior. The Government of Mali derives payoff from preserving sovereign control, ensuring political continuity, and gaining international legitimacy; however, its risk-sensitive payoff is discounted by the threat of coups, internal dissent, and loss of donor support if abuses or corruption are revealed. The Malian Armed Forces receive payoff through territorial dominance, mission success, resource control, and institutional promotions, yet these are penalized by the risk of ambush, reputational damage from human rights violations, and battlefield attrition. The National Security Services maximize payoffs via intelligence dominance, urban control, and political access, but suffer loss functions from counterintelligence failure, infiltration, or public exposure of abuse. Rebel Movements seek secessionist or autonomist gains, local legitimacy, and international recognition; their payoff is severely risk-sensitive to betrayal, airstrikes, sanctions, and internal fragmentation. Jihadist Groups maximize ideological penetration, Sharia governance, martyrdom prestige, and tax extraction, but face high volatility from drone attacks, civilian resistance, and intra-jihadist rivalries. Criminal Bandits target immediate monetary payoffs through ransoms and smuggling, but factor in risk penalties from capture, loss of safe routes, or betrayal by allies. Ethnic Militias derive utility from community defense, ethnopolitical leverage, and localized dominance; their risk-aware payoff is discounted by retaliation cycles, resource depletion, and being labeled as terrorists. Traditional Hunters  maximize spiritual legitimacy and social cohesion within their community, with risk penalties emerging from modern weapons asymmetry and erosion of traditional authority. Local Militias obtain payoff from payment, power projection, and property protection, but lose utility under betrayal, lack of cohesion, or mission failure. Foreign Militias (Wagner) maximize geopolitical influence, income, and control over resource zones, while accounting for risk via casualty sensitivity, international backlash, and operational overextension.
In hybrid formations, coalitions compute joint and individual payoffs using asymmetric access to information and trust. Rebel+Jihadist coalitions optimize territory and tactical shock impact but suffer disutility from ideological misalignment and coordination breakdowns. Rebel+Bandit payoffs are transactionally defined and risk being undermined by informant leakage or greed. Bandit–Jihadist alliances balance income maximization with ideological scrutiny and reputational volatility. Rebel+Hunter coalitions gain local legitimacy but risk alienating broader support bases. Army+Hunter and Army+Foreign Militia coalitions aim to boost operational efficiency, but incur risk from moral hazard, civilian casualties, and divergent rules of engagement. Population+Hunter coalitions derive utility through community protection and cultural continuity, yet face disutility from collective punishment or being mislabelled as insurgents. Militia+Rebel and Mixed Armed Coalitions often engage in payoff-sharing under short-term gain logic, heavily discounted by mistrust, side-switching, and power asymmetry.
The Local Population evaluates payoff through safety, access to services, and preservation of livelihoods, highly penalized by displacement, violence, and forced recruitment. Communities, religious or ethnic, maximize cohesion, intergenerational survival, and cultural autonomy, while suffering disutility from division, manipulation, and mass violence. Political Parties derive payoff from electoral success, power brokering, and narrative control, but face discounting through destabilization and loss of grassroots support. Civil Society Organizations receive utility from funding, mediation success, and local empowerment, but are penalized when perceived as partial, elite-driven, or externally co-opted. Diaspora Groups maximize symbolic capital, remittance-driven influence, and ideological imprint on the homeland, but risk marginalization, backlash, or misinformation loops.
The International Community, including diplomats and peace envoys, derives payoff from mediation credibility, geopolitical balance, and policy implementation, discounted by failures of peacekeeping, deadlocked dialogues, and hostage situations. MINUSMA gains through civilian protection metrics, ceasefire verification, and field presence, but pays a heavy cost under attack, obstruction, or political scapegoating. ECOWAS and the AU pursue regional stability and normative leadership, with disutility arising from perceived inaction, failed sanctions, or fragmented consensus. International NGOs earn utility from access, donor satisfaction, and successful delivery, while encountering risk-adjusted losses through security incidents, funding gaps, or co-optation. Foreign Military Forces accrue payoff through tactical success, alliance maintenance, and regional influence, highly penalized by political backlash at home, civilian casualties, and operational fatigue. Multinational Corporations aim to maximize extraction profits, regulatory advantage, and security stability, but suffer sharp payoff reductions from sabotage, labor unrest, and reputational loss. Transnational Criminal Networks maximize profit through stealth logistics, corruption, and market control, but discount payoff based on interdiction, partner betrayal, and violence volatility.
The informal and shadow actors have distinct risk-aware incentives. Opportunistic double agents seek personal survival, income, and influence; their payoff is fragile and discounted by high betrayal probability and limited long-term trust. Mixed agents derive benefits from cross-cutting authority and rent-seeking, but suffer from exposure risks and inter-factional reprisals. Revenger combatants operate under trauma-driven utility functions favoring punitive actions; their perceived gains are existential but often lead to high-risk suicidal missions or retribution spirals. Defectors and dissatisfied fighters seek redemption, protection, or greater personal benefit, incurring risk when trust is not restored or their insider knowledge degrades. Arms dealers optimize logistics throughput and transactional dominance, but discount for theft, sting operations, or war fatigue. Propaganda funders and operatives maximize influence through narrative steering, misinformation loops, and morale shifts, with disutility from counter-narratives, cyber exposure, and saturation. Mission-based militia contractors earn short-term gains based on task success but are penalized for failure, betrayal, or overspecialization. False-flag actors aim to disrupt coherence and create scapegoats, with payoff tied to confusion duration and media uptake, but risk high if exposed. Online psy-ops agents derive utility from shaping public opinion, destabilizing enemy morale, and driving factionalism; their downside is linked to traceability, algorithmic suppression, and cross-border prosecution.  This extensive network forms a multi-level, multi-agent asymmetric dynamic game, with:
Hidden types (informants, double agents), asymmetric resources (weapons, terrain knowledge, foreign aid), different time horizons (short-term mercenaries vs long-term political actors), non-aligned payoffs (ideology vs profit vs security vs survival)

%
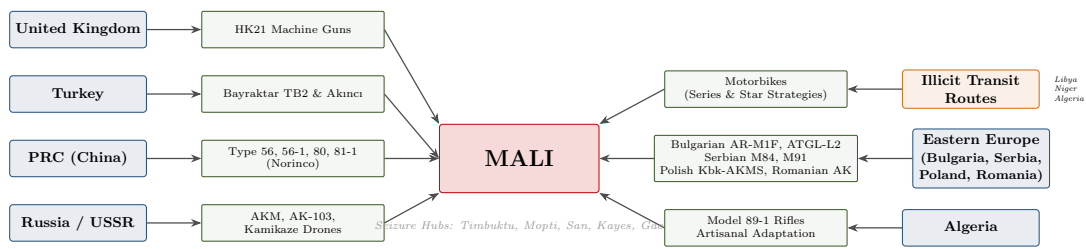
\begin{figure}[htbp]
    \centering
    \definecolor{OriginBlue}{HTML}{1F4E79}
    \definecolor{ResourceGreen}{HTML}{375623}
    \definecolor{MaliRed}{HTML}{C00000}
    \definecolor{TransitOrange}{HTML}{D16A16}
    
    \resizebox{0.95\textwidth}{!}{%
    \begin{tikzpicture}[
        font=\sffamily,
        >=Stealth,
        node distance=0.6cm and 1.2cm,
        country/.style={
            rectangle, draw=OriginBlue, fill=OriginBlue!10,
            rounded corners=2pt, minimum width=3cm,
            minimum height=0.8cm, align=center, font=\bfseries\small
        },
        resource/.style={
            rectangle, draw=ResourceGreen, fill=ResourceGreen!5,
            rounded corners=1pt, minimum width=4cm,
            minimum height=0.8cm, align=center, font=\scriptsize
        },
        transit/.style={
            rectangle, draw=TransitOrange, fill=TransitOrange!10,
            rounded corners=2pt, minimum width=3cm,
            minimum height=0.8cm, align=center, font=\bfseries\small
        },
        destination/.style={
            rectangle, draw=MaliRed, fill=MaliRed!15,
            rounded corners=2pt, minimum width=3.5cm,
            minimum height=1.5cm, align=center, font=\bfseries\Large
        },
        flow/.style={->, draw=black!70, line width=0.8pt}
    ]

    \node[destination] (mali) {MALI};

    \node[resource, left=of mali] (china_w) {Type 56, 56-1, 80, 81-1 \\ (Norinco)};
    \node[country, left=of china_w] (china) {PRC (China)};

    \node[resource, above=of china_w] (turk_w) {Bayraktar TB2 \& Akıncı};
    \node[country, left=of turk_w] (turkey) {Turkey};

    \node[resource, below=of china_w] (rus_w) {AKM, AK-103, \\ Kamikaze Drones};
    \node[country, left=of rus_w] (russia) {Russia / USSR};

    \node[resource, above=of turk_w] (uk_w) {HK21 Machine Guns};
    \node[country, left=of uk_w] (uk) {United Kingdom};

    \node[resource, right=of mali] (ee_w) {Bulgarian AR-M1F, ATGL-L2 \\ Serbian M84, M91 \\ Polish Kbk-AKMS, Romanian AK};
    \node[country, right=of ee_w] (eeur) {Eastern Europe \\ \footnotesize(Bulgaria, Serbia, \\ \footnotesize Poland, Romania)};

    \node[resource, below=of ee_w] (alg_w) {Model 89-1 Rifles \\ Artisanal Adaptation};
    \node[country, right=of alg_w] (algeria) {Algeria};

    \node[resource, above=of ee_w] (motor) {Motorbikes \\ (Series \& Star Strategies)};
    
    \node[transit, right=of motor] (routes) {Illicit Transit \\ Routes};
    \node[align=left, font=\tiny\itshape, right=0.2cm of routes] (transit_list) {Libya \\ Niger \\ Algeria};

    \draw[flow] (uk) -- (uk_w);
    \draw[flow] (uk_w.east) -- (mali.north west);
    \draw[flow] (turkey) -- (turk_w);
    \draw[flow] (turk_w.east) -- (mali.west);
    \draw[flow] (china) -- (china_w);
    \draw[flow] (china_w.east) -- (mali.west);
    \draw[flow] (russia) -- (rus_w);
    \draw[flow] (rus_w.east) -- (mali.south west);

    \draw[flow] (routes) -- (motor);
    \draw[flow] (motor.west) -- (mali.north east);
    \draw[flow] (eeur) -- (ee_w);
    \draw[flow] (ee_w.west) -- (mali.east);
    \draw[flow] (algeria) -- (alg_w);
    \draw[flow] (alg_w.west) -- (mali.south east);

    \node[below=0.5cm of mali, font=\scriptsize\itshape, color=gray!80] 
    {Seizure Hubs: Timbuktu, Mopti, San, Kayes, Gao};

    \end{tikzpicture}
    }
    \caption{Example of Supply Chain Mapping: Origins and Diversification of Tactical Resources into the Malian Conflict Environment.}
    \label{fig:mali_resource_mapping}
\end{figure}
\subsection{How are these Resources Landing in Mali?}
Tracking of the data from Timadie, Guinaga, WETE, Grabal, SK1 Sogoloton shows clearly that Mali does not manufacture drones, motorbikes, or firearms domestically; instead, these items are sourced through a combination of official imports, illicit trafficking, and local artisanal modification/adaptation. The Malian Armed Forces have acquired drones such as the Turkish-made \textit{Bayraktar TB2} and   Bayraktar \textit{Akıncı} models, delivered between late 2022 and March 2023, as well as Russian-made kamikaze drones. Firearms used by various actors in Mali originate from multiple sources. The Malian military has received substantial stockpiles of weapons produced by China's Norinco, including Type 80 machine guns and Type 81-1 assault rifles. Several weapons from East of Europe and older Soviet Union models are present. Illicit arms trafficking is rampant, with weapons smuggled along routes from Libya, Algeria, and Niger into Mali, and a significant number of illegal weapons have been seized in regions like Timbuktu, Mopti, San, Kayes, and Gao.  Different series are available. Examples (see Figure  \ref{fig:mali_resource_mapping})  include 
Chinese Type 56 assault rifles (series 520XXXX, 521XXXX, 522XXXX, 523XXXX, 524XXXX; 373XXXX; 377XXXX; 571XXXX; 290XXXX 7-digit from Factory 9696; 290XXXXX 8-digit from Factory 9336) and Type 56-1 assault rifles (series 460XXXX, 171XXXXX, 560XXXXX, 361XXXX); Russian AKM assault rifles (series 26XXXX from 1974 and THXXXX); Algerian Model 89-1 assault rifles (series 0101XXXX, 2001); Polish Kbk-AKMS assault rifles (1970s–1980s); Bulgarian ATGL-L2 RPG launcher (KO 54 1080, 2014); Serbian M91 marksman rifle (2242, 2013); British HK21 machine gun (EN 50937); Serbian M84 machine guns (42884, 42909, 56669); Bulgarian AR-M1F assault rifles (1N 50 8616, 1N 50 8771, 2010); Russian AKM assault rifle (UR-2293-06, 2006) and AK-103 assault rifles (171153211, 171167071, 2017); Bulgarian AR-M1F41 assault rifle (KO 57 6251, 2017); and a Romanian AK-pattern assault rifle (1980 KR 2192, 1980). Motorbikes, commonly used by armed groups for mobility, are typically imported through both legal and illegal channels, with very little domestic production in Mali.  Motorbikes are used in two different attack strategies: series or star.

\subsection{Intergenerational MFTG }
The formulation of Mali's multi-actor security crisis as an intergenerational MFTG captures the full scope of systemic dynamics observed on the ground. MFTGs are a generalization of classical  games in which the instantaneous payoff and state dynamics of each agent are functions not only of their own type, state, action, and information but also of the full distribution over these variables in the population of interacting agents. This structure is essential in the Malian context, where over 50 heterogeneous actor types, including formal entities such as the Malian Armed Forces and international coalitions, informal actors such as traditional hunters and diaspora networks, and covert agents such as double-dealers, ideological revengers, arms traffickers, and psychological warfare operators, interact in a decentralized, volatile, and partially observed environment. The consequences of any one actor’s decision (a rebel group defecting to a jihadist alliance, or a local militia shifting support toward a foreign sponsor) are rarely isolated; rather, they generate cascading effects that alter the collective state-action-information field of the system. These distributional externalities where an agent's risk-aware payoff and belief evolution depend on the emergent density of violence, displacement, radicalization, misinformation, and trauma, can only be rigorously modeled within an MFTG framework. Moreover, the intergenerational nature of the conflict, marked by historical grievances, trauma transmission, and dynamic allegiance shifts, necessitates a temporal layering in the model where agents' current behaviors are shaped by inherited state variables (exposure to conflict, loss of kin, indoctrination histories) and where their actions alter the empirical measure that defines the strategic shapes for future cohorts. Standard multi-agent or aggregate models collapse these long-term, path-dependent effects, whereas intergenerational MFTGs explicitly encode them in the evolution of type distributions and forward-looking objective functions. In Mali, agents’ beliefs and decisions are shaped not only by direct violence or resource scarcity but also by the perceived distributions of allegiances, risks, and narratives in the wider population, including both physical and virtual domains (social media propaganda and misinformation). MFTGs offer a principled way to integrate these belief-mediated dynamics by treating the information structure itself as a state variable subject to distributional influence. In doing so, they also facilitate the modeling of  hybrid equilibria involving probabilistic misperceptions, stochastic switching of group identities, and endogenous network formation.

\begin{figure}
\begin{tikzpicture}[
    node distance=1.6cm and 3cm,
    every node/.style={align=center},
    block/.style={rectangle, draw=black, fill=blue!10, rounded corners, minimum width=3.5cm, minimum height=1cm, text centered},
    arrow/.style={draw, thick, ->, >=Stealth},
    ]

\node[block, fill=red!15] (G0) {Generation 0\\\textbf{Libya}\\Initial destabilization\\Conflict start};
\node[block, right=of G0] (G1) {Generation 1\\\textbf{Libya $\rightarrow$ Niger corridor}\\\textbf{Libya $\rightarrow$ Algeria corridor}\\Militia movement};
\node[block, below=of G1] (G2) {Generation 2\\\textbf{Algeria $\rightarrow$ Northern Mali}\\\textbf{Niger $\rightarrow$ Northern Mali}\\Infiltration and arms flow};
\node[block, left=of G2] (G3) {Generation 3\\\textbf{Northern Mali}\\ Conflict consolidation};
\node[block, below=of G3] (G4) {Generation 4\\\textbf{Northern \& Central Mali}\\ Conflict expansion};
\node[block, below=of G4] (G5) {Generation 5\\\textbf{Northern, Central, West, East Mali}\\\textbf{+ Burkina, Niger}\\Multinational spread};
\node[block, right=of G5] (G6) {Generation 6\\\textbf{Mali (all regions), Burkina, Niger}\\\textbf{+ Togo, Benin, Nigeria}\\Expanded destabilization};

\draw[arrow] (G0) -- (G1);
\draw[arrow] (G1) -- (G2);
\draw[arrow] (G2) -- (G3);
\draw[arrow] (G3) -- (G4);
\draw[arrow] (G4) -- (G5);
\draw[arrow] (G5) -- (G6);

\node[above=0.3cm of G0] {\textbf{\Large   Intergenerational MFTG:} \\  A new wave to   existing weakest links};

\end{tikzpicture}
\end{figure}

\subsection*{Intergenerational MFTG : formulation}

We consider $G\ge1$ discrete generations.  The total time horizon is
\[
T \;=\;\sum_{g=1}^G T_g,
\]
where generation \(g\) is active over
\[
t \in \{\,T_{g-1}+1,\dots,T_{g-1}+T_g\},\qquad T_0=0, \  T_g\geq 0.
\]

{\bf The intergenerational setup is inspired from the field: revenger child-soldiers whose trauma-conditioned perceived best-responses is favouring the conflict. The ones who were child in 2012 in the Northern part of the country are now men and women. Some of them are entering as active combatants in the conflict. The leads to a non-peaceful area over several generations. The proposed model  captures it.}

\subsubsection*{Common Types}
At each time \(t\), the \emph{common type}  is 
$\theta_{c,t}\;\in\;\Theta_c.$

\subsubsection*{Common (Macro)  States, resource states or nature state}
At each time \(t\), the \emph{common resource state, macro-state} and \emph{ nature state} are
\[
s_{c,t}\;=\;\bigl(s_{c,1,t},\dots,s_{c,n_c,t}\bigr)
\;\in\;S_c
=\prod_{k=1}^{n_c}S_{c,k},
\]
where, for example,
\[
s_{c,1,t}=\text{resource availability},\quad
s_{c,2,t}=\text{conflict intensity},\;\dots
\]

\subsubsection*{ Dominating (Atomic) Agents}
There is a finite set \(\mathcal D\) of \emph{atomic agents} \(i\), each with type, individual state, and private information
\[
(\theta^i_{\rm do},\,s^i_{\rm do,t},\,I^i_{\rm do,t})
\;\in\;
\Theta_{\rm do}\times S_{\rm do}\times\mathcal I_{\rm do}.
\]
Their joint state and action vectors are
\[
s_{\rm do,t }=(s^i_{\rm do,t})_{i\in\mathcal D},\qquad
a_{\rm do,t}=(a^i_{\rm do,t})_{i\in\mathcal D}.
\]
We denote by
\[
\mu_{\rm do,t}\;\in\;\mathcal P\bigl(\Theta_{\rm do}\times S_{\rm do}\bigr),
\quad
\tilde\mu_{\rm do,t}\;\in\;\mathcal P\bigl(A_{\rm do}\bigr)
\]
the probability distribution of types‒states and of actions of all dominating agents, respectively.
Each atomic agent \(i\) receives payoff
\[
u_{\rm do,t}^i\bigl(t,
\theta_c,s_c,\,
\theta_{\rm do},s_{\rm do},\,I^i_{\rm do},\,a_{\rm do},\,
\mu_{\rm do},\tilde\mu_{\rm do},\,
m,\tilde m,\,
\{(\theta_{p,do}^{j},s_{p,do}^{j},a_{p,do}^{j},\mu_{p,do}^{j},\tilde\mu_{p,do}^{j})\}_j
\bigr),
\]
where \((m,\tilde m)\) are the population mean‐fields of the mixed agents.

\subsubsection*{ Mixed (Non‐Atomic) Agents}
There is a discrete set \(\mathcal P\) of \emph{mixed agents} \(j\), each with
\[
(\theta^j_{p},\,s^j_{p,t},\,I^j_{p,t})
\;\in\;
\Theta_{p}\times S_{p}\times\mathcal I_{p}.
\]
Their type, state, and action decompose into atomic vs.\ non‐atomic parts:
\[
\theta_p^j=(\theta_{p,do}^j,\theta_{p,-do}^j),\quad
s_{p,t}^j=(s_{p,do,y}^j,s_{p,-do,y}^j),\quad
a_{p,t}^j=(a_{p,do,t}^j,a_{p,-do,t}^j).
\]
The joint mean‐fields over mixed‐agent states and actions are
\[
m_t\;\in\;\mathcal P(S_p),
\quad
\tilde m_t\;\in\;\mathcal P(A_p),
\]
while each \(j\) also tracks personal distributions
\[
\mu^j_{p,t}\;\in\;\mathcal P\bigl(S_{p,do}\times S_{p,-do}\bigr),
\quad
\tilde\mu^j_{p,t}\;\in\;\mathcal P\bigl(A_{p,do}\times A_{p,-do}\bigr).
\]
Their payoff is
$
u_{p,t}^j\bigl(t,\,
\theta_c,s_c,\,
\theta_{\rm do},s_{\rm do},\,a_{\rm do},\,\mu_{\rm do},\tilde\mu_{\rm do},\,
m,\tilde m,\, \{(\theta_{p,do}^{j'},s_{p,do}^{j'},a_{p,do}^{j'},\mu_{p,do}^{j'}, \tilde\mu_{p,do}^{j'})\}_{j'},\, $\\ $
\theta_{p,-do}^j,s_{p,-do}^j,I_{p}^j,a_{p,-do}^j,\,
\mu^j_{p,-do},\tilde\mu^j_{p,-do}
\bigr).
$

Over each generation interval, the mean-field terms
$\mu_{\rm do}, m, \, \{(\mu_{p,do}^{j},\mu^j_{p,-do})\}_j, $  evolves
according to Kolmogorov forward equations and 
the mean-field terms $\tilde\mu_{\rm do},\, \tilde m,\, \{( \tilde\mu_{p,do}^{j}, \tilde\mu^j_{p,-do})\}_j, $ evolves depending on the chosen actions.
The evolution of the payoff of the dominating agent  $i$ is determined by the probability law  $\rho_{cdo}(t,.)$ of 
$$(\theta_c,s_c,\,
\theta_{\rm do},s_{\rm do}, I_{do},\,
\{(\theta_{p,do}^{j'},s_{p,do}^{j'})\}_{j'}).$$
The evolution of the payoff of the mixed agent $j$ is determined by the probability law  $\rho_{cdop}(t,.)$ of 
$(\theta_c,s_c,\,
\theta_{\rm do},s_{\rm do},\,
\{(\theta_{p,do}^{j'},s_{p,do}^{j'})\}_{j'},\, \theta_{p,-do}^j,s_{p,-do}^j,I_{p}^j).$

Each atomic agent \(i\)  of generation $g$ receives  a  payoff 
\[ h_{\rm do,T_g}^i\bigl(T_g,\,
\theta_c,s_c,\,
\theta_{\rm do},s_{\rm do},
\mu_{\rm do},
m,
\{(\theta_{p,do}^{j},s_{p,do}^{j},\mu_{p,do}^{j})\}_j
\bigr)
+ \sum_{t=T_{g-1}+1}^{T_g-1}
u_{\rm do,t}^i,
\]
and a mixed agent $j$ in generation $g$ receives  $
h_{p,T_g}^j\bigl(T_g,\,
\theta_c,s_c,\,
\theta_{\rm do},s_{\rm do},\,\mu_{\rm do},
m,\,\{(\theta_{p,do}^{j'},s_{p,do}^{j'},\mu_{p,do}^{j'})\}_{j'},$ $\theta_{p,-do}^j,s_{p,-do}^j,\mu^j_{p,-do},
\bigr) +  \sum_{t=T_{g-1}+1}^{T_g-1}
u_{p,t}^j,
$
where $ h_{\rm do,T_g}^i\bigl(T_g, .)$ is the terminal payoff at generation $g$ (see Figure \ref{figgg10})

\begin{figure}[htbp]
    \centering
    \resizebox{0.9\textwidth}{!}{%
        \begin{tikzpicture}[
            font=\sffamily,
            >=Stealth,
            box/.style={
                rectangle, draw, rounded corners=6pt, 
                minimum width=4.8cm, minimum height=1.5cm, 
                align=center, thick, drop shadow={opacity=0.1},
                inner sep=12pt
            },
            common/.style={box, fill=blue!10, draw=blue!70!black},
            atomic/.style={box, fill=red!10, draw=red!70!black},
            mixed/.style={box, fill=green!12, draw=green!70!black},
            mathnode/.style={box, fill=white, draw=gray!40, dashed, minimum width=4.8cm},
            gen_tag/.style={fill=gray!80, text=white, font=\bfseries, rounded corners=2pt, inner sep=4pt},
            side_label/.style={font=\scshape\bfseries, color=gray!60, align=right}
        ]

            
            \node[gen_tag] (G1_Tag) at (-6.5, 6.5) {GENERATION $g=1$};

            \node[common] (macro1) at (0, 6.5) {
                \textbf{Common Macro-State $s_{c,t}$} \\ 
                \small Resource Availability \& Conflict Intensity
            };

            \node[atomic] (at1) at (-3.5, 4) {
                \textbf{Dominating Agent $i$} \\
                \small State $s_{do,t}^i$ \textbullet{} Action $a_{do,t}^i$
            };
            \node[mathnode] (mu1) at (3.5, 4) {
                \textbf{Mean-Field-Type $\mu_{do,t}$} \\
                \footnotesize $\mathcal{P}(\Theta_{do} \times S_{do})$
            };

            \node[mixed] (mx1) at (-3.5, 1) {
                \textbf{Mixed Agent $j$} \\
                \small State $s_{p,t}^j$ \textbullet{} Action $a_{p,t}^j$
            };
            \node[mathnode] (m1) at (3.5, 1) {
                \textbf{Population Mean Fields} \\
                \footnotesize $m_t \in \mathcal{P}(S_p),\  \tilde{m}_t\in  \mathcal{P}(A_p)$
            };

            \node[box, fill=orange!10, draw=orange!70!black, minimum width=11.8cm] (term1) at (0, -1.5) {
                \textbf{Terminal Payoff $h_{T_1}$ \& Generational Trauma Transfer} \\
                \small Captures evolution of perceived best-responses for the next generation
            };

            \draw[->, thick, blue!70!black] (macro1.south) -- ++(0,-0.5) -| (at1.north);
            \draw[->, thick, blue!70!black] (macro1.south) -- ++(0,-0.5) -| (mu1.north);
            
            \draw[<->, ultra thick, red!60!black] (at1) -- (mu1) node[midway, fill=white, font=\tiny\bfseries] {KOLMOGOROV FORWARD};
            \draw[<->, ultra thick, green!60!black] (mx1) -- (m1) node[midway, fill=white, font=\tiny\bfseries] {POPULATION POLICY LAW};

            \draw[->, thick, gray!60, dashed] (at1) -- (mx1) node[midway, right, font=\tiny] {DOMINANCE};
            \draw[->, thick, gray!60, dashed] (mu1) -- (m1) node[midway, right, font=\tiny] {MEAN-FIELD COUPLING};

            \draw[->, thick, gray] (mx1.south) |- ($(term1.north)+( -3.5, 0)$);
            \draw[->, thick, gray] (m1.south)  |- ($(term1.north)+(  3.5, 0)$);

            
            \draw[->, line width=4pt, gray!20] (term1.south) -- ++(0,-1.5);
            \draw[->, line width=1.5pt, orange!80, postaction={draw, orange!80, line width=0.5pt, dash pattern=on 4pt off 4pt}] 
                (term1.south) -- ++(0,-1.5) 
                node[midway, right=0.2cm, text=black, font=\footnotesize\itshape] {Initial conditions for $g+1$ (Conditioned Trauma)};

            
            \begin{scope}[yshift=-10.5cm]
                \node[gen_tag] (G2_Tag) at (-6.5, 6.5) {GENERATION $g=2$};

                \node[common] (macro2) at (0, 6.5) {
                    \textbf{Common Macro-State $s_{c,t}$} \\ 
                    \small Evolved Resource \& Nature States
                };

                \node[atomic] (at2) at (-3.5, 4) {
                    \textbf{Dominating Agent $i'$} \\
                    \small (Active Combatants)
                };
                \node[mathnode] (mu2) at (3.5, 4) {
                    \textbf{Mean-Field-Type $\mu'_{do,t}$} \\
                    \footnotesize Updated Probability Law
                };

                \node[mixed] (mx2) at (-3.5, 1) {
                    \textbf{Mixed Agent $j'$}
                };
                \node[mathnode] (m2) at (3.5, 1) {
                    \textbf{Population MF $m'_t$}
                };
                
                \draw[->, thick, blue!70!black] (macro2.south) -- ++(0,-0.5) -| (at2.north);
                \draw[->, thick, blue!70!black] (macro2.south) -- ++(0,-0.5) -| (mu2.north);
                \draw[<->, ultra thick, red!60!black] (at2) -- (mu2);
                \draw[<->, ultra thick, green!60!black] (mx2) -- (m2);
            \end{scope}

            \begin{scope}[on background layer]
                \node[fit=(G1_Tag)(macro1)(term1), fill=black!2, rounded corners=12pt, draw=gray!15] {};
                \node[fit=(G2_Tag)(macro2)(m2), fill=black!2, rounded corners=12pt, draw=gray!15] {};
                
                \node[side_label] at (-7.5, 7.5) {Macro-Level\\(Nature)};
                \node[side_label] at (-7.5, 4.0) {Meso-Level\\(Dominating)};
                \node[side_label] at (-7.5, 1.0) {Micro-Level\\(Mixed)};
            \end{scope}

        \end{tikzpicture}%
    }
    \caption{Architecture of the Intergenerational MFTG. The vertical flow highlights the sequential inheritance of state distributions across discrete generations.} \label{figgg10}
\end{figure}
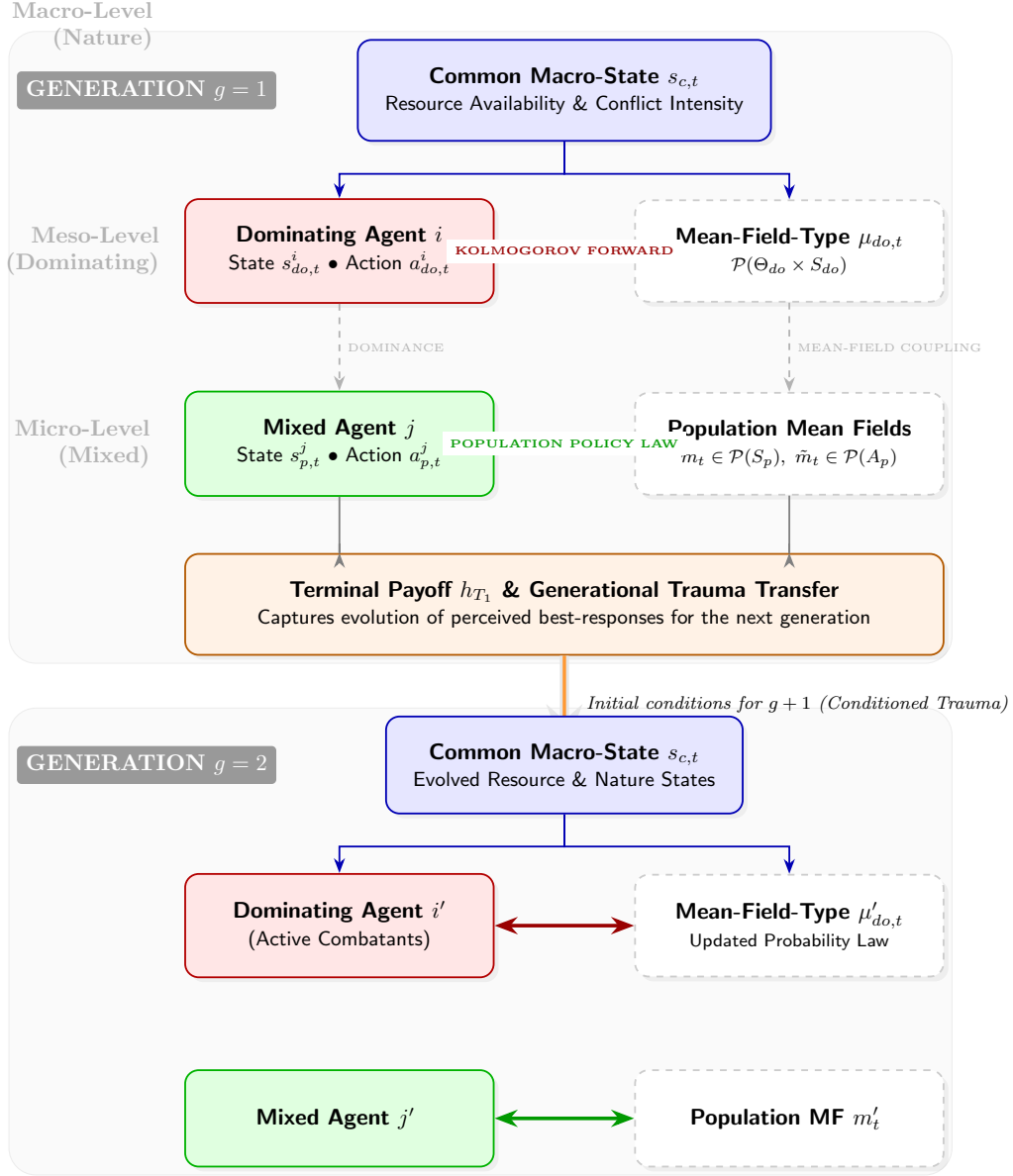

\begin{remark} Here the standard dynamic programming fails because the state $s$ alone is not Markovian. One needs to augment it at the distributional level with  $\rho_{cdop}(t,.)$ the distribution of $$(\theta_c,s_c,\,
\theta_{\rm do},s_{\rm do},\,
\{(\theta_{p,do}^{j'},s_{p,do}^{j'})\}_{j'},\, \theta_{p,-do}^j,s_{p,-do}^j,I_{p}^j).$$ The marginal of  $\rho_{cdop,t}(.)$ with respect to $(\theta_c,s_c,\,
\theta_{\rm do},s_{\rm do}, I_{do},\,
\{(\theta_{p,do}^{j'},s_{p,do}^{j'})\}_{j'})$ yields  $\rho_{cdo,t}(.).$ One can therefore work with $\rho_{cdop}$ as the new augmented state of the problem. Within this framework, $\rho_{cdop}$ evolves deterministically via Kolmogorov equation.  Therefore, the Bellman equation will be applied in the space of measures for each class of agent.
\end{remark} 

\begin{remark}As we know in state-based dynamic games, playing an  equilibrium at each step, does not necessarily lead to an equilibrium for the finite horizon game. Similarly, playing a mean-field-type equilibrium at given generation does not lead to an equilibrium of the intergenerational MFTG.
\end{remark}

The transition kernel of $\rho_{cdop}$ is denoted  by $$\rho_{cdop,t+1}=  K_{p,t}\bigl(t, \rho_{cdop}, 
a_{\rm do},\, \tilde\mu_{\rm do}, m,\tilde m,\,\{(a_{p,do}^{j'},\tilde\mu_{p,do}^{j'})\}_{j'},a_{p,-do}^j,\, \tilde\mu^j_{p,-do}\bigr).$$ Similarly for  $\rho_{cdo}$ which is a marginal of $\rho_{cdop}.$

If there is a solution to the following Bellman system 

\begin{equation}
\begin{array}{ll}
V^i_{do,T_{G}}( \rho_{cdo})= \int h_{\rm do}^i\bigl(T_g, .)  d\rho_{cdo},  \ i\in  \mathcal{D}\\
V^j_{p,T_{G}}( \rho_{cdop})= \int h_{p}^j\bigl(T_g, .)  d\rho_{cdop}, \ j\in  \mathcal{P} \\
V^i_{do,t}( \rho_{cdo,t})= \sup_{a^i_{do}}  \{\int  u_{do,t}^j d\rho_{cdo,t}+  V^i_{do,t+1}( \rho_{cdo,t+1})\}, \ i\in  \mathcal{D} \\

V^j_{p,t}( \rho_{cdop,t})= \sup_{a^j_{p}}  \{\int  u_{p,t}^j d\rho_{cdop,t}+  V^j_{p,t+1}( \rho_{cdop,t+1})\},  \  \ j\in  \mathcal{P} 
\\ \ t\in \{1,\ldots, T_{G}-1\}
\end{array}
\end{equation}
where $a^i_{do}, a^j_{p}$ belong to the class of randomized behavioural actions. Then, there resulting $(V^i_{do}, V^j_{p})$ is a mean-field-type equilibrium payoff and  the behavioral mixed actions $(a^i_{do}, a^j_{p})_{i,j}$ is a  mean-field-type equilibrium.

\begin{remark} Here, the equilibrium structure is in randomized behavioral  state-and-mean-field-type feedback strategies  i.e., a measurable function of $\rho$ (or its marginals)  and $$(\theta_c,s_c,\,
\theta_{\rm do},s_{\rm do},\,
\{(\theta_{p,do}^{j'},s_{p,do}^{j'})\}_{j'},\, \theta_{p,-do}^j,s_{p,-do}^j,I_{p}^j).$$ 
\end{remark}

\subsection{Intergenerational Cycles of Violence}
This result formalizes, within a game-theoretic and measure-theoretic framework, why structural grievances, asymmetric information, trauma inheritance, and revenge-based preferences can sustain violent conflict despite intermittent periods of peace or foreign intervention

\begin{result}[Persistence of War via Intergenerational Cycles of Violence]
 Consider an intergenerational MFTG defined over a dynamic population of heterogeneous agents indexed by type, state, action, and informational history. Let each generation inherit a violence-conditioned distribution over states and priors, including trauma, revenge propensity, mistrust, and perceived threat. Suppose
there exists a non-zero subset of agents whose payoff is strictly increasing in retaliatory action under perceived harm;
the state transition kernel is path-dependent and includes memory of past trauma or group loss;
the type distribution evolves with a non-vanishing measure on revenge-seeking types or opportunistic actors exploiting war economies;
temporary breaks in violence do not erase past states, and no actor has perfect information or global enforcement power.

Then, even under temporary ceasefires, peace agreements, or generational shifts, the MFTG admits equilibria in which cycles of violence persist over multiple generations with positive probability. Moreover, such equilibria are dynamically stable under weak continuity conditions of the state-action-measure kernel and admit no absorbing peaceful state unless external shocks or structural transformations alter the type-action distribution.
\end{result}

To prove this main result we examine both a generic MFTGs  involving all 50 types of actors with relatively general payoffs.

\subsection*{Retaliation, Trauma, and Peace in the Intergenerational MFTG}
We introduce the following definitions.

\begin{definition}[Retaliatory Type]
An \emph{atomic} agent \(i\in\mathcal D\) of type \(\theta^i_{\rm do}\in\Theta_{\rm do}\) is called \textbf{retaliatory} if there exists a subset of its state space 
\(\;S^i_{\rm trauma}\subset S_{\rm do}\)
and a unique feedback action 
\(\;a^i_{r}\in A_{\rm do}\)
such that, for all 
\(\;s^i_{\rm do,t}\in S^i_{\rm trauma},\;
I^i_{\rm do,t}\in\mathcal I_{\rm do},\)
and any joint mean‐fields \(\mu_{\rm do,t},\tilde\mu_{\rm do,t},m_t,\tilde m_t\),
\[
u^i_{\rm do,t}\Bigl(t,\ldots,, a^i_r,\,a^{-i}_{\rm do,t},\, \dots\Bigr)
\;>\;
u^i_{\rm do}\bigl(\dots,a^i,a^{-i}_{\rm do,t},\dots\bigr)
\]
for every alternative atomic action \(a^i\neq a^i_r\).  In other words, whenever \(i\)'s state lies in the \emph{trauma manifold} \(S^i_{\rm trauma}\), its uniquely optimal response is the \emph{retaliatory action} \(a^i_r\).
\end{definition}

\begin{definition}[Trauma State Persistence]
Define the \emph{trauma region} of the atomic state by 
\(\mathcal S_{\rm trauma}\subset S_{\rm do}\).  We say trauma is \emph{persistent} if there exists \(\delta>0\) such that whenever some atomic agent \(i\) chooses its retaliatory action \(a^i_r\),
\[
\Pr\Bigl(s^i_{\rm do,t+1}\in\mathcal S_{\rm trauma}\;\Big|\;s^i_{\rm do,t}\in\mathcal S_{\rm trauma},\;a^i_t=a^i_r\Bigr)
\;\ge\;\delta.
\]
Thus retaliation both arises from and reinforces trauma across time.
\end{definition}

\begin{definition}[Peaceful Equilibrium Distribution]
A distribution \(\mu^*\in\mathcal P(\Theta_{\rm do}\times S_{\rm do}\times\mathcal I_{\rm do}\times A_{\rm do})\times\mathcal P(S_p\times A_p)\) is a \emph{peaceful state} if
\[
\mu^*\bigl\{\,s_{\rm do}\in\mathcal S_{\rm trauma}\bigr\}=0,\quad
\mu^*\bigl\{\theta^i_{\rm do}\text{ retaliatory}\bigr\}=0,
\quad
\mathrm{supp}\bigl(\tilde\mu_{\rm do}^*\bigr)\subset A_{\rm do}^{\rm peace},\quad
\mathrm{supp}\bigl(\tilde m^*\bigr)\subset A_p^{\rm peace},
\]
where \(A_{\rm do}^{\rm peace}\subset A_{\rm do}\) and \(A_p^{\rm peace}\subset A_p\) denote non‐violent (de‐escalatory) action sets.  In a peaceful equilibrium, no agent, atomic or mixed, retaliates or remains in trauma.
\end{definition}

\begin{definition}[Peaceful vs.\ Aggressive Types]
Partition the atomic and mixed action spaces as
\[
A_{\rm do}=A_{\rm do}^{\rm peace}\cup A_{\rm do}^{\rm conflict},\quad
A_p=A_p^{\rm peace}\cup A_p^{\rm conflict}.
\]
\begin{itemize}
  \item An atomic type \(\theta^i_{\rm do}\) is \emph{peaceful} if for every possible state \(s^i_{\rm do,t}\), information \(I^i_{\rm do,t}\), and mean‐fields, its unique best‐response action lies in \(A_{\rm do}^{\rm peace}\).
  \item It is \emph{aggressive} if its unique best‐response action always lies in \(A_{\rm do}^{\rm conflict}\).
  \item Similarly, a mixed non‐atomic type \(\theta^j_p\) is \emph{peaceful} if its best‐response action belongs to \(A_p^{\rm peace}\) under all states, information, and mean‐fields, and \emph{aggressive} if always in \(A_p^{\rm conflict}\).
\end{itemize}
\end{definition}

\begin{theorem}[Persistence of War via Intergenerational Cycles of Violence]
Within the above intergenerational MFTG, suppose:
\begin{enumerate}
  \item[(i)] {\bf Retaliatory Payoff.}  There exists a non–zero mass of atomic types
  \(\theta^i_{\rm do}\in\Theta_{\rm do}\) and mixed types \(\theta^j_p\in\Theta_p\)
  for which the instantaneous utilities
  \[
    u^i_{\rm do}\bigl(\dots,a^i_{\rm do}=a_r,\dots\bigr)
    \quad\text{and}\quad
    u^j_p\bigl(\dots,a^j_p=a_r,\dots\bigr)
  \]
  are strictly larger than for any non‐retaliatory action \(a\notin A^{\rm conflict}\), whenever their private or collective states lie in the trauma region \(S_{\rm trauma}\subset S_{\rm do}\) or \(S_{p,{\rm trauma}}\subset S_p\).
  \item[(ii)] {\bf Path‐Dependent Memory.}  The joint state‐transition maps \(\Phi^{\rm do}\) and \(\Phi^p\) are path‐dependent and satisfy
  \[
    \Pr\bigl(s^{i}_{\rm do,t+1}\in S_{\rm trauma}\mid s^i_{\rm do,t}\in S_{\rm trauma},\,a^i_{\rm do,t}=a_r\bigr)\ge\delta>0,
  \]
  and similarly for mixed states, so that past retaliation reinforces future trauma.
  \item[(iii)] {\bf Enduring Revenge Types.}  The population distributions \(\mu_{\rm do,t}\) and \(m_t\) evolve under kernels \(\Gamma_{\rm do},\Gamma_p\) that preserve a strictly positive measure on retaliatory (revenge‐seeking or opportunistic) types in \(\Theta_{\rm do}\) and \(\Theta_p\) across all \(t\).
  \item[(iv)] {\bf Imperfect Enforcement.}  No single agent or coalition wields perfect information \(I\) or enforcement power to drive the system to a fully peaceful action support \(A^{\rm peace}\).
\end{enumerate}
Then, for any generation \(g=1,\dots,G\) and any \(t\in\{T_{g-1}+1,\dots,T_{g-1}+T_g\}\), there exist mean‐field equilibria \(\bigl(\mu_{\rm do,t},\tilde\mu_{\rm do,t},m_t,\tilde m_t\bigr)\) in which
\[
\mu_{\rm do,t}\bigl(\Theta_{\rm do}^{\rm rev}\bigr)>0,
\quad
m_t\bigl(\Theta_{p}^{\rm rev}\bigr)>0,
\]
and the support of \(\tilde\mu_{\rm do,t}\) and \(\tilde m_t\) intersects \(A^{\rm conflict}\) with positive probability.  These equilibria are dynamically stable under mild continuity of \(\Gamma_{\rm do},\Gamma_p,\Phi^{\rm do},\Phi^p\), and no absorbing peaceful distribution (i.e.\ one with zero mass on \(S_{\rm trauma}\), \(\Theta^{\rm rev}\), or \(A^{\rm conflict}\)) exists unless exogenous shocks or structural transformations reallocate measure off the retaliatory types.
\end{theorem}
\begin{proof} 
We prove that under assumptions (i)–(iv), for each generation \(g\) and each time \(t\in\{T_{g-1}+1,\dots,T_{g-1}+T_g\}\), there exists an equilibrium with positive mass on retaliatory types and conflict actions, and no absorbing peaceful distribution.

\paragraph{Non‐vanishing retaliation mass.}  By (i) there is a subset \(\Theta_{\rm do}^{\rm rev}\subset\Theta_{\rm do}\) of atomic types of strictly positive measure under \(\mu_{\rm do, T_{g-1}+1}\) whose utility is maximized by the unique retaliatory action \(a_r\) whenever \(s^i_{\rm do}\in S_{\rm trauma}\).  Similarly, a subset \(\Theta_p^{\rm rev}\subset\Theta_p\) of mixed types exists under \(m_{T_{g-1}+1}\).  Hence initially
\[
\mu_{\rm do, T_{g-1}+1}\bigl(\Theta_{\rm do}^{\rm rev}\bigr)>0,\quad m_{T_{g-1}+1}\bigl(\Theta_p^{\rm rev}\bigr)>0.
\]

\paragraph{Trauma persistence and feedback.}  By (ii), once an atomic agent \(i\) with \(s^i_{\rm do}\in S_{\rm trauma}\) plays \(a_r\), its next‐period state remains in \(S_{\rm trauma}\) with probability at least \(\delta>0\).  Thus, for any time \(t\), the conditional probability that such an agent remains in the trauma manifold decays at most geometrically: after \(k\) successive retaliatory actions,
\[
\Pr\bigl(s^i_{\rm do}(t+k)\in S_{\rm trauma}\mid s^i_{\rm do,t}\in S_{\rm trauma}\bigr)\;\ge\;\delta^k.
\]
An identical argument holds for mixed agents in \(S_{p,\rm trauma}\).  Consequently, once an agent enters trauma and given its best‐response is retaliation, its state and thus its incentive to retaliate are self‐reinforcing over time.

\paragraph{ Intergenerational type reinforcement.}  By (iii), the transition kernels \(\Gamma_{\rm do}\) and \(\Gamma_p\) preserve a strictly positive measure on retaliatory types.  Concretely, let
\[
\alpha:=\inf_{t\in[T_{g-1}+1,\,T_{g-1}+T_g]}\mu_{\rm do,t}\bigl(\Theta_{\rm do}^{\rm rev}\bigr)>0.
\]
Then by the path‐dependence of states and policies, each period’s retaliation by a fraction \(\alpha\) of the population yields a next‐period distribution with at least \(\alpha\) mass on \(\Theta_{\rm do}^{\rm rev}\).  This bound holds uniformly over generations since the same kernels apply in each generation interval.

\paragraph{ Existence of conflict‐supporting equilibrium.}  
Under mild continuity and compactness assumptions, the best‐response correspondences for atomic and mixed agents are upper‐hemicontinuous, convex‐valued, and non‐empty on the compact metric space of distributions \((\mu_{\rm do},\tilde\mu_{\rm do},m,\tilde m)\).  By Kakutani’s fixed‐point theorem, for each \(t\) there is an equilibrium mean‐field
\[
\bigl(\mu_{\rm do,t}^*,\tilde\mu_{\rm do,t}^*,m^*_t,\tilde m^*_t\bigr)
\]
consistent with these best responses.  Since \(\mu_{\rm do,t}\bigl(\Theta_{\rm do}^{\rm rev}\bigr)\ge\alpha>0\), the support of \(\tilde\mu_{\rm do,t}^*\) must include conflict actions \(A^{\rm conflict}\) with positive probability; likewise for \(\tilde m^*_t\).

\paragraph{ Non‐existence of absorbing peaceful state.}  
A peaceful absorbing distribution would require
\(\mu_{\rm do,t}(S_{\rm trauma})=0\), \(\mu_{\rm do,t}(\Theta_{\rm do}^{\rm rev})=0\), and
\(\mathrm{supp}(\tilde\mu_{\rm do,t})\subset A^{\rm peace}\).  But by the above,
\(\mu_{\rm do,t}(\Theta_{\rm do}^{\rm rev})\ge\alpha>0\) for all \(t\), and
trauma states persist under retaliation.  Moreover, (iv) precludes any single agent or coalition from enforcing a jump to \(\tilde\mu_{\rm do,t}\subset A^{\rm peace}\).  Hence no such absorbing peaceful equilibrium can arise absent an exogenous shock or structural change that reallocates mass off \(\Theta^{\rm rev}\).

\medskip
Thus, even in the presence of temporary ceasefires or generational breaks, one obtains a family of dynamical paths in which a positive fraction of the population remains in trauma and continues to choose conflict actions across all \(G\) generations.  This completes the proof.
\end{proof}

\subsection{Intergenerational Dynamic Equilibria } \label{seccoi4}

The perpetuation of intergenerational violence emerges from the interaction of distinct actor “types” and affective drivers whose strategic choices reinforce one another over time. Aggressive types whether jihadist cadres enforcing territorial tax regimes, armed bandits exploiting state absence, or rebel commanders seeking tactical advantage, derive higher utility from coercive actions (ambushes, forced levies) and therefore systematically select conflict‐oriented options. By contrast, peaceful types, including community elders, civil‐society activists, and human‐rights-focused security personnel, maximize gains through dialogue, public‐goods provision, and negotiated settlements, yet remain vulnerable to reputational damage and resource loss when exposed to violence. Interwoven with these strategic dispositions are affective motivations: a revenge attitude (personal or collective vendetta following witnessed atrocities), a provocative attitude (deliberate signaling of strength to deter rivals), and a pervasive feeling of injustice (rooted in ethnic/community marginalization or perceived state abandonment), each of which lowers the threshold for violent retaliation and elevates the salience of aggressive payoffs. Because children and bystanders internalize both the normative belief that ``violence pays" and the emotional trauma of past ruptures, the distribution of actor types and attitudes in one generation conditions the payoffs and information structures of the next. Consequently, retaliatory spirals and strategic provocations reproduce themselves, generating a self‐reinforcing cycle of violence that only a coordinated reshaping of incentives, and a restoration of procedural justice, can break

\subsubsection{Intergenerational Revengers Game}

We partition each action space into peaceful vs.\ conflict:
\[
A_{\rm do}=A_{\rm do}^{\rm peace}\cup A_{\rm do}^{\rm conflict},\qquad
A_p=A_p^{\rm peace}\cup A_p^{\rm conflict}.
\]

\begin{definition}[Revenger Agent]
An atomic type $\theta^i_{\rm do}$ (resp.\ mixed type $\theta^j_p$) is a \emph{revenger} if there exists a nonempty trauma subset $S^i_{\rm trauma}\subset S_{\rm do}$ (resp.\ $S^j_{p,\rm trauma}\subset S_p$) and a unique conflict action $a_r\in A_{\rm do}^{\rm conflict}$ (resp.\ $A_p^{\rm conflict}$) such that for all 
\(\,s^i_{\rm do}\in S^i_{\rm trauma},\,I^i_{\rm do},\,\mu_{\rm do},\tilde\mu_{\rm do},m,\tilde m\),
\[
u^i_{\rm do}\bigl(\dots,a_r,\dots\bigr)
>
u^i_{\rm do}\bigl(\dots,a,\dots\bigr)
\quad
\forall\,a\neq a_r,
\]
and similarly for $u^j_p$ when $s^j_p\in S^j_{p,\rm trauma}$.  
\end{definition}

\begin{theorem}[Persistence of Intergenerational Retaliation]

Assume that there exist $\beta_{\rm do},\beta_p>0$ and trauma‐memory kernels $\Phi^{\rm do},\Phi^p$ such that:
\begin{enumerate}
  \item If $s^i_{\rm do,t}\in S^i_{\rm trauma}$ and $a^i_r$ is played, then 
    \[
      \Pr\bigl(s^i_{\rm do,t+1}\in S^i_{\rm trauma}\mid s^i_{\rm do,t},a^i_r\bigr)
      \;\ge\;\beta_{\rm do}>0.
    \]
    Similarly for $s^j_{p,t}\in S^j_{p,\rm trauma}$.
  \item The type‐transition maps $\Gamma_{\rm do}$, $\Gamma_p$ satisfy
    \[
      \Pr\bigl(\theta^i_{\rm do,t+1}\text{ is revenger}\mid \theta^i_{\rm do,t}\text{ revenger}\bigr)
      \;\ge\;\beta_{\rm do}>0,
    \]
    and likewise for mixed types with $\beta_p>0$.
  \item No choice of peaceful incentives $\pi$ can overturn the unique conflict best‐response of a revenger type in trauma.
\end{enumerate}

If  at the start of generation \(g\) the distributions satisfy
\(\mu_{\rm do, T_{g-1}+1}(\Theta_{\rm do}^{\rm rev})>0\)
and
\(m_{T_{g-1}+1}(\Theta_p^{\rm rev})>0\),
then for every $t\in\{T_{g-1}+1,\dots,T_{g-1}+T_g\}$,
\[
\mu_{\rm do,t}\bigl(\Theta_{\rm do}^{\rm rev}\bigr)
\;>\;0,
\quad
m_t\bigl(\Theta_p^{\rm rev}\bigr)
\;>\;0.
\]
Hence, each generation retains a strictly positive mass of revenger agents, and conflict actions persist with positive probability across all $G$ generations.
\end{theorem}

\begin{proof}
We proceed by induction on generation index \(g\).

\textbf{Base case ($g=1$).}  By hypothesis, 
\(\mu_{\rm do, 1}(\Theta_{\rm do}^{\rm rev})>0\)
and
\(m_1(\Theta_p^{\rm rev})>0.\)

\textbf{Inductive step.}  Fix generation \(g\) and suppose for some \(\delta>0\),
\[
\mu_{\rm do,T_{g-1}+1}(\Theta_{\rm do}^{\rm rev})\ge\delta,
\quad
m_{T_{g-1}+1}(\Theta_p^{\rm rev})\ge\delta.
\]
During times \(t\in \{T_{g-1}+1,\dots,T_{g-1}+T_g\}\), any revenger in trauma will choose action \(a_r\) by definition, and by Assumption the state remains in the trauma set with probability \(\ge\beta_{\rm do}\) each period.  Hence at time \(t+1\),
\[
\mu_{\rm do,t+1}(\Theta_{\rm do}^{\rm rev})
\;\ge\;
\delta\cdot\beta_{\rm do}
\;>\;0,
\]
since a fraction \(\delta\) of agents are revengers and at least \(\beta_{\rm do}\) of those remain in trauma and hence remain revengers by Assumption.  An identical argument applies to mixed agents with \(\beta_p\).  By induction, the positivity of revenger mass holds throughout generation \(g\) and carries into generation \(g+1\).

By induction over \(g=1,\dots,G\), the mass of revenger types remains strictly positive in every generation and thus conflict actions persist indefinitely across the intergenerational horizon.
\end{proof}

By appropriately implementing a well-calibrated incentive function $\pi$ that raises the relative utility of peaceful actions beyond that of retaliatory ones, the best-response policy of agents originally labeled as “revenger types” can be shifted toward nonviolent strategies. As a result, The Assumption above fails to hold, and the intergenerational persistence of retaliation collapses. The transition kernel responds accordingly, leading to a progressive reduction in the population mass of revenger types demonstrating the tractability and effectiveness of incentive design in disrupting inherited cycles of violence.
This means that  the mass of revenger types $\mu^g(\Theta^{\mathrm{rev}} \times S \times A \times \mathcal I)$ can decrease over generations.
Let $\theta \in \Theta^{\mathrm{rev}}$ be a revenger type, defined by the property that for all $(s, I, \mu)$, its optimal action without incentive lies in $A^{\mathrm{retaliation}}$:
\[
a^\star \in  \arg\max_{a \in A} \left\{ R(\theta, s, a) - C(\theta, a, \mu) \right\} \in A^{\mathrm{retaliation}}.
\]

We now define an incentive function $\pi$ as follows:
\[
\pi(\theta, s, a, I, \mu) =
\begin{cases}
\gamma > 0 & \text{if } a \in A^{\mathrm{peace}}, \\
0 & \text{if } a \in A^{\mathrm{retaliation}}.
\end{cases}
\]

Let $a' \in A^{\mathrm{peace}}$. We choose $\gamma$ such that
\[
R(\theta, s, a') - C(\theta, a', \mu) + \gamma \geq R(\theta, s, a^\star) - C(\theta, a^\star, \mu).
\]
Then under the modified utility
\[
u(\theta, s, a, I, \mu) = R(\theta, s, a) - C(\theta, a, \mu) + \pi(\theta, s, a, I, \mu),
\]
the new optimal action becomes
\[
a^{\text{new}} \in  \arg\max_{a \in A} u(\theta, s, a, I, \mu) \in A^{\mathrm{peace}}.
\]
Thus, the agent $\theta$ no longer selects a retaliatory action and therefore no longer satisfies the definition of a revenger type.

Moreover, the type transition kernel $\mathcal K$ satisfies:
\[
\Pr[\theta^{g+1} \in \Theta^{\mathrm{rev}} \mid \theta^g = \theta, s, a^{\text{new}}, I, \mu] < \beta,
\]
since the action is now peaceful and the offspring is more likely to adopt peaceful traits under intergenerational transmission.

By adjusting $\gamma$ appropriately across generations, the measure $\mu^g(\Theta^{\mathrm{rev}} \times S \times A \times \mathcal I)$ decreases over $g$, leading to a convergence toward peaceful types.

The incentive function can be implemented and tracked using a blockchain-based system integrated with VoiceMI tools tailored to local languages \cite{audioiamali}. It enables transparent, decentralized, and verifiable distribution of peace rewards in conflict-affected regions. Community members or trusted local mediators submit voice-verified reports on successful mediation events, absence of retaliatory action, or participation in co-development initiatives via mobile devices using VoiceMI models trained to recognize speech in Senufo, Bozo,  Tommo-So Dogon, and other regional tongues. These voice submissions are automatically transcribed, interpreted, and validated through a hybrid model combining local adjudication nodes and MI moderation. Once verified, the corresponding peaceful actions are encoded as transactions on a permissioned blockchain ledger (using Hyperledger or Algorand), ensuring tamper-proof recording and traceability of  rewards. Tokenized incentives (mobile credit, food vouchers, or cooperative investment tokens) are then disbursed to verified individuals or communities via mobile wallets. This integration of local-language VoiceMI and blockchain not only minimizes fraud and strengthens trust, but also ensures that incentive delivery is adaptive, transparent, and locally grounded effectively operationalizing the incentive $\pi$ in fragile, low-literacy, low-infrastructure environments. The funding of these incentives is sourced from locally generated revenue streams: community cooperatives, small-scale agri-processing, microenterprise hubs, and renewable energy projects seeded through catalytic capital. By linking peace behavior to tangible economic opportunity grounded in local entrepreneurship, job creation, and reinvestment into village-scale businesses, the system ensures that incentives are both contextually legitimate and financially self-sustaining, aligning behavioral transformation with endogenous economic growth.

\subsubsection{ War Entrepreneur Subgame }
{\it A war entrepreneur is a conflict capitalist, whose livelihoods depend on a strategic perpetuation of war.}  War entrepreneurs are  diverse and strategically embedded class of actors who profit, directly or indirectly, from the continuation, expansion, and mutation of violence. At the center are arms traffickers and transnational dealers who supply small arms, light weapons, and increasingly, improvised drones and explosives to all factions ranging from jihadist groups  to ethnic militias, self-defense groups, and even segments of the national armed forces. Operating through porous Sahelian corridors, such as from southern Libya via northern Niger or across the Mauritania-Mali border, these actors exploit institutional weaknesses to maintain a steady flow of arms, often profiting more in times of fragmentation and peak insecurity. Closely linked are unregistered market intermediaries and logistical smugglers who traffic in fuel, food, medicine, and dual-use goods, often extracting “conflict premiums” and facilitating barter exchanges of weapons for livestock, gold, or narcotics. Militia leaders and local warlords, particularly those in western, central and northern Mali, transform their territorial control into monetized protection economies by taxing traders, charging fees at checkpoints, mining site hostages, or leasing access to land, often initiating or sustaining violence to reinforce their legitimacy or retain external patronage. Political entrepreneurs, including some elected officials and aspirants, manipulate ethnic grievances and fund armed youth to secure electoral advantages or influence negotiations, sometimes switching sides or maintaining relations with multiple armed factions. Within the state apparatus itself, corrupt security officials engage in war entrepreneurship by diverting weapons, leaking intelligence, or offering safe passage to traffickers and armed actors in exchange for bribes, while overstating threats to maintain inflated budgets and foreign military support. Foreign-linked private military contractors and security consultants, often hired by multinational corporations, further embed profit motives into the violence by lobbying for long-term security contracts, sometimes collaborating with local militias or promoting alarmist narratives to justify their presence. Proxy militias backed by foreign powers, such as Wagner-affiliated units or irregular forces aligned with external geopolitical agendas, extract rents from mining concessions, guard illicit trade, or destabilize regions to serve their sponsors’ interests. 
In the digital sphere, informational war entrepreneurs engage in sophisticated propaganda, paying influencers and deploying disinformation campaigns on social media to inflame ethnic divisions, provoke retaliation, or undermine peace actors, thereby generating conditions that justify arms purchases or paramilitary mobilization. Kidnap-for-ransom networks, often overlapping with insurgent factions or enabled by corrupt intermediaries, orchestrate abductions of foreigners, NGO workers,  political figures or even local civilian, fueling cycles of hostage negotiation that sustain cash flows and perpetuate insecurity. Transnational criminal syndicates trafficking drugs, people, and weapons depend on Mali’s zones of limited state control to maintain impunity, financing armed escorts and bribing local authorities to ensure unimpeded flows across borders. These war entrepreneurs thrive in conditions of weak governance, economic desperation, and institutional fragmentation, transforming insecurity into a profitable market system that continually regenerates demand for violence, arms, and protection, making peace not merely difficult to achieve, but systemically unprofitable for many of the actors embedded within the conflict economy.

Ammunition circulating in Mali  originates from a wide array of producing states, often funneled through complex and opaque supply chains that span continents and exploit weak regulatory enforcement. Among the most prominent sources are former Eastern Bloc countries, particularly Russia, which remains a major producer of AK-pattern ammunition such as 7.62×39mm and 7.62×54mm rounds, RPG munitions, and belt-fed machine gun calibers. These Russian-origin munitions are frequently found in the hands of both state and non-state actors across Mali. Serbia and Bulgaria also contribute significantly through legacy arms industries, supplying 7.62mm and 9mm calibers either directly or via intermediaries; their exports have often been traced in conflict zones, including through markings on casings found in seized caches. Romania and Ukraine similarly export Soviet-standard calibers, with Ukrainian rounds increasingly appearing since 2014, either as surplus or leaked through destabilized logistics chains. China is another major supplier, manufacturing vast quantities of 7.62×39mm, 12.7mm, and 5.8mm ammunition that often enters the Sahel via indirect state transfers or unauthorized third-party channels. Chinese munitions have been frequently documented in Malian territory, especially among insurgent and militia groups. Sudan, via its Military Industry Corporation, is a regional producer of Chinese and Russian-caliber ammunition; Sudanese rounds have been discovered in several Sahelian conflicts, often used by non-state actors or trafficked across porous borders. Iran also contributes, particularly through its export of 7.62mm, 12.7mm, mortar rounds, and rocket projectiles; seizures in West Africa have included Iranian-marked items, some of which arrived through militant-affiliated smuggling networks. Pakistan, through Pakistan Ordnance Factories (POF), manufactures both NATO and Soviet calibers, which have appeared in battlefield recoveries in Mali and its neighboring states. Turkey, an increasingly assertive regional arms exporter, has shipped 9mm, 7.62mm, and shotgun ammunition to clients and proxies, including actors operating in Mali and the wider Sahel. Ammunition of U.S. and NATO origin (from France, Belgium, Czech Republic, Germany, and Italy) also circulates in the region, not always via direct transfer to Mali, but often through third countries, battlefield captures, or leaks from peacekeeping or private military stockpiles. Some of these NATO-caliber rounds, such as 5.56mm and 9mm, have been diverted after being issued to partner forces or private contractors. Egypt manufactures compatible Soviet and Chinese-standard ammunition and has historically played a role in regional arms flows, while the United Arab Emirates, especially post-2015, has developed defense firms capable of producing and exporting small arms and munitions that have reached conflict zones via Libya or irregular militias. Perhaps most destabilizing is Libya, not as a producer but as a massive post-2011 source of looted ammunition; after Gaddafi's fall, thousands of tons of Soviet and Chinese rounds, RPGs, and technical-mounted heavy weapons were funneled southward through traffickers into Mali, Niger, and Burkina Faso. This post-collapse stockpile redistribution has served as a key enabler of prolonged conflict in Mali.

Many of the countries where ammunition originates are not merely passive exporters of arms but can be considered strategic war entrepreneurs in the Malian conflict, insofar as they derive economic, geopolitical, or proxy influence benefits from the continuation of instability. These countries may not directly incite violence, but their industries, intelligence services, or political networks play instrumental roles in sustaining the material conditions for war by supplying arms, supporting aligned factions, or shaping security narratives to justify long-term involvement. We isolate a subgame played by a representative \emph{war entrepreneur} who manipulates conflict, embedded in the full intergenerational MFTG. From a game-theoretic perspective, war entrepreneurs are agents whose utility functions are strictly increasing in the prevalence of violence, fragmentation, and mistrust within the population, which they seek to influence through targeted actions, asymmetric incentives, and psychological operations. They may not always appear on the battlefield, but their presence is deeply embedded in the logistical, financial, and informational infrastructure of war. As such, they represent a structural force that opposes stabilization, peacebuilding, and reconciliation, making any post-conflict transition precarious unless their incentives are neutralized or transformed through systemic reform. At time \(t\) the full population law is 
\(\mu_t=(\mu_{\rm do,t},m_t)\),
and the common‐resource state is 
\(\chi_t\), conflict intensity 
\(\omega_t\), and geopolitical index \(\rho_t\).

\paragraph{War‐Entrepreneur Control.}  
A single global actor (or cartel) chooses 
\[
a_t\in\{0,1\},
\]
with \(a_t=1\) denoting active \emph{conflict manipulation} (arms shipments, false‐flag operations).

\paragraph{Aggregate State Variables.}  
\[
W_t\in\{0,1\}
\quad\text{(war indicator),}
\qquad
s_t^{\rm frag}\in[0,1]
\quad\text{(macro fragmentation index).}
\]

\paragraph{Dynamics.}  
\begin{equation}
\begin{array}{l}
\Pr\bigl(W_{t+1}=1\mid W_t,a_t,\mu_t\bigr)
  =\psi\bigl(W_t,a_t,\mu_t\bigr),\\
s_{t+1}^{\rm frag}
  =s_t^{\rm frag}
    +\alpha\,a_t
    +\eta(\mu_t)
    -\delta\,s_t^{\rm frag}, \\ s_0^{\rm frag}
\end{array}
\end{equation}
where \(\psi\) is increasing in \(a_t\), and \(\eta(\mu_t)>0\) when \(\mu_t\) assigns mass to conflict‐prone types or high‐trauma states.

\paragraph{Demand and Revenue.}  
Weapon demand is
\[
D_t
  = f\bigl(W_t,s_t^{\rm frag},\mu_t\bigr),
\]
strictly increasing in each argument.  With fixed unit price \(P>0\), revenue is \(P\,D_t\), and manipulation cost is
\[
C(a_t)=
\begin{cases}
c,&a_t=1,\\
0,&a_t=0.
\end{cases}
\]

\paragraph{Objective.}  
The war entrepreneur chooses \(\{u_t\}\) to
\[
\max_{a_{1:T}}
\sum_{t=1}^T\Bigl[P\,f(W_t,s_t^{\rm frag},\mu_t)\;-\;C(a_t)\Bigr],
\]
subject to the above dynamics and to the mean‐field update
\(\mu_{t+1}=\Phi(\mu_t,a_t)\), capturing the aggregate response of all atomic and mixed agents.

\subsection*{Manipulation Dominates Under Weak‐State }

\begin{theorem}[Profit Maximization via Conflict Manipulation]
Assume:
 \(f(W,s,\mu)\) is strictly increasing in \(W,s\), and in the mass of conflict‐prone types in \(\mu\), 
 \(\eta(\mu)>0\) whenever \(\mu\) places positive mass on conflict actions,
 \(\psi(W,u,\mu)\) is strictly increasing in \(u\).
  \(W_1=1\), \(s_1^{\rm frag}>0\),
   \(c>0\) is the per‐period cost of manipulation.
Then there exists \(c^\star>0\) such that for all \(0<c<c^\star\), the constant policy \(a_t=1\) for all \(t\) strictly outperforms any policy with \(a_{t^*}=0\) for some \(t^*\).
\end{theorem}

\begin{proof}
Compare two policies:

\textbf{Full manipulation (Policy A):} \(a_t\equiv1\).  By induction,
\(s_t^{\rm frag}\) increases each period, and \(W_t\) stays \(1\) with high probability.  Hence
\[
\sum_{t=1}^T\Bigl[P\,f(1,s_t^{\rm frag},\mu_t)\;-\;c\Bigr]
\]
grows rapidly in \(t\).

\textbf{Any restraint (Policy B):} Suppose \(a_{t^*}=0\).  At \(t^*\),
\(s_{t^*+1}^{\rm frag}\) drops by \(\alpha\), and future \(\eta(\mu_t)\)
may vanish if conflict‐prone mass falls.  Also,
\(\Pr(W_{t^*+1}=1)\) decreases since \(\psi(1,0,\mu_{t^*})<\psi(1,1,\mu_{t^*})\).
Thus expected revenue beyond \(t^*\) is strictly lower, and the cost saving \(c\)
at \(t^*\) cannot offset the larger aggregate loss if \(c<c^\star\).

Formally, the profit difference \(\Delta\) between A and B satisfies
\[
\Delta 
\ge
\sum_{t=1}^T \bigl[f(1,s_t^{A},\mu_t^{A}) - f(W_t^{B},s_t^{B},\mu_t^{B})\bigr]\,P
\;-\;T\,c,
\]
which is positive for small enough \(c\).  Hence \(u_t\equiv1\) is strictly optimal.
\end{proof}

For the war entrepreneur subgame to operate and persist in reality, several structural preconditions must hold: chief among them being the institutional weakness of the state, the presence of multiple socio-political fault lines, and widespread economic deprivation among the population. A weak country, characterized by limited territorial control, fragile governance institutions, and fragmented security forces, creates an unmonitored arena in which external and local war entrepreneurs can maneuver with minimal legal or military constraint. The existence of multiple weakest links such as ethnic  polarization, historical grievances, intercommunal tensions, or poorly integrated regions, provides entry points for manipulation, allowing arms dealers and conflict instigators to ignite or escalate localized violence with relatively low cost and high expected return. When the local population is economically impoverished and lacks alternative livelihoods, the opportunity cost of participation in violence is drastically reduced. Under such conditions, even small arms or financial incentives can mobilize youth into militias, generate demand for self-defense weapons, or encourage communities to arm themselves out of fear and distrust. Together, these features constitute a high-reward environment for war entrepreneurs: fragmented governance fails to prevent proliferation; societal fault lines enable the creation of demand; and economic desperation sustains the market. Without structural reforms that address state capacity, reduce social fragmentation, and expand economic alternatives, such a game remains not only possible but systemically self-reinforcing.

\subsubsection{Diaspora/Expats War Entrepreneur Subgame}  
Diaspora/Expats groups who mobilize resources to purchase guns and ammunition for their home villages or ethnic communities, though often motivated by a sense of protection, justice, or solidarity, function as war entrepreneurs when their actions materially contribute to the perpetuation, escalation, or institutionalization of violence. Unlike purely symbolic or humanitarian support, diaspora-driven arms procurement introduces a strategic and financial logic into local conflict, effectively turning these groups into actors who shape battlefield dynamics, fuel militarization, and alter the balance of power on the ground.
By channeling funds into weapons acquisition, diaspora actors directly contribute to the expansion of local armament, enabling communities to bypass or resist national disarmament efforts, peace accords, or state authority. Whether the intention is defense against jihadists, retaliation for past atrocities, or preemption of rival communities, the result is the same: violence becomes more accessible, less controllable, and more likely to be repeated. Second, such actions reinforce ethnic or communal militarization, deepening identity-based fault lines. When external funds are used to arm a specific ethnic group or village, it sends a message of zero-sum security, prompting other communities to rearm in turn. This accelerates arms races, feeds cycles of mistrust, and makes peaceful coexistence harder to achieve.
Diasporas often operate outside regulatory oversight, procuring weapons through black markets or middlemen, and thereby stimulating demand for trafficked arms. This sustains the broader war economy by rewarding arms dealers and intermediaries who profit from instability. Diaspora involvement introduces a feedback loop of symbolic escalation: narratives of victimhood or historical injustice are amplified in exile, sometimes simplified or radicalized, and re-imported through arms transfers and militarized rhetoric. This can prolong conflicts that local actors might otherwise resolve through negotiation, and it crowds out grassroots peacebuilding.
By transforming remittances into instruments of armed resistance or community defense, diasporas take on an entrepreneurial role in the conflict, not just funding violence but shaping its direction, legitimacy, and duration. Even when motivated by trauma or genuine concern, their actions increase the stakes and persistence of violence, positioning them as informal but consequential players in Mali’s conflict ecosystem. In this sense, they must be understood as diaspora war entrepreneurs, agents who invest in, sustain, and sometimes escalate local warfare from afar, embedding violence into the political economy of transnational solidarity.

A diaspora war entrepreneur is a member of an expatriate or transnational community who strategically funds, fuels, or sustains armed conflict in their country of origin, not primarily out of ideological conviction or humanitarian concern, but to gain influence, protect local allegiances, settle inter-group scores, or assert control over territorial, ethnic, or political narratives. These actors often send remittances not merely for economic support, but for the explicit purpose of arming militias, financing retaliatory operations, or maintaining networks of control within their home communities. They may operate under the guise of community solidarity, religious duty, or ancestral loyalty, yet in practice, their financial and informational interventions perpetuate local cycles of violence by enabling armed responses, revenge campaigns, or resistance to peace accords. In many cases, diaspora war entrepreneurs are embedded in social media ecosystems and informal financial circuits that amplify trauma-based narratives, call for retribution, or justify continued mobilization. They often fund arms procurement, propaganda, and logistical support, while remaining physically distant from the direct consequences of the violence they help finance. Their influence is amplified when the local population is traumatized, economically vulnerable, and socially fragmented, creating a context in which externally sponsored violence appears legitimate or necessary. 

From a strategic perspective, diaspora war entrepreneurs exploit their informational asymmetry, emotional capital, and transnational networks to shape conflict dynamics on the ground, often aligning with local actors who serve as enforcers of their agendas. Because they are outside the reach of domestic regulation and often reside in liberal democracies that uphold remittance rights and freedom of expression, their interventions are difficult to trace or constrain, making them a powerful yet frequently overlooked, force in the political economy of protracted conflict.
We isolate a subgame between \emph{diaspora war entrepreneurs} and their \emph{local recipients}, embedded in the full population and macro‐state dynamics.

\paragraph{Time Horizon.}  
Each subgame unfolds over the global horizon \(t\in \{1,\dots,T\}\), itself partitioned into generation intervals \(\{T_{g-1}+1,\dots,T_{g-1}+T_g\}\).

\paragraph{Agents.}  
\begin{itemize}
  \item \textbf{Diaspora Entrepreneurs} \(i\in\mathcal D\):  
    \(\bigl(\theta^i_{\rm do},\,s^i_{\rm do,t},\,I^i_{\rm do,t}\bigr)\in\Theta_{\rm do}\times S_{\rm do}\times\mathcal I_{\rm do}\).  
    They choose \(a^i_{\rm do,t}\in A_{\rm do}=\{\mathrm{symbolic},\,\mathrm{aid},\,\mathrm{arms}\}\).
  \item \textbf{Local Recipients} \(j\in\mathcal P\):  
    \(\bigl(\theta^j_p,\,s^j_{p,t},\,I^j_{p,t}\bigr)\in\Theta_p\times S_p\times\mathcal I_p\).  
    They choose \(a^j_{p,t}\in A_p=\{\mathrm{peace},\,\mathrm{defense},\,\mathrm{retaliation},\,\mathrm{mobilize}\}\).
  \item \(\theta^{-j},s^{-j},a^{-j},I^{-j}\): the types, states, actions, information of all other agents.
\end{itemize}

\paragraph{Macro‐States and Distributions.}  
$ s_{c,t}\in S_c,\quad \chi_t\in Ci,\  \omega_t\in\Omega,\quad \rho_t\in R,$ 
$ \mu_{\rm do,t}\in\mathcal P(\Theta_{\rm do}\times S_{\rm do}),\; $
$  \tilde\mu_{\rm do,t}\in\mathcal P(A_{\rm do}),$
$ m_t\in\mathcal P(S_p),\; $
 $ \tilde m_t\in\mathcal P(A_p).$
We write \(\mu_t\equiv(\mu_{\rm do,t},m_t)\) and similarly for \(\tilde\mu_t\).

\paragraph{State Transitions.}  For idiosyncratic shocks \(\xi^i_t,\xi^j_t\),
\[
\begin{aligned}
s^i_{\rm do,t+1}
&=\Phi^{\rm do}\bigl(s^i_{\rm do,t},\,a^i_{\rm do,t},\,\mu_t,\,s_{c,t},\,\chi_t,\,\omega_t,\,\rho_t,\,\xi^i_t\bigr),\\
s^j_{p,t+1}
&=\Phi^p\bigl(s^j_{p,t},\,a^j_{p,t},\,a^i_{\rm do,t},\,\theta^{-j},\,s^{-j}_t,\,a^{-j}_t,\,I^{-j}_t,\,\mu_t,\,\xi^j_t\bigr).
\end{aligned}
\]

\paragraph{Payoff Functions.}  
\begin{itemize}
  \item \textbf{Diaspora Entrepreneur \(i\)} at \(t\) obtains
  \[
    u^i_{\rm do,t}
    =U_{\rm impact}\bigl(\theta^i_{\rm do},a^i_{\rm do,t},\mu_t,\omega_t\bigr)
    \;-\;C_{\rm transfer}\bigl(a^i_{\rm do,t},\rho_t\bigr)
    \;-\;R_{\rm guilt}\bigl(s^i_{\rm do,t},a^i_{\rm do,t},\mu_t,\omega_t\bigr).
  \]
  \item \textbf{Local Recipient \(j\)} at \(t\) obtains
  \[
  \begin{aligned}
    u^j_{p,t}
    &=
    S_{\rm security}\bigl(\theta^j_p,s^j_{p,t},a^j_{p,t},a^i_{\rm do,t},\chi_t,\omega_t\bigr)\\
    &\quad
    -\,R_{\rm risk}\bigl(\theta^j_p,a^j_{p,t},\mu_{\rm do,t}^{-j},a_{\rm do,t}^{-j},\omega_t\bigr)\\
    &\quad
    +\,B_{\rm revenge}\bigl(\theta^j_p,s^j_{p,t},a^j_{p,t},\theta^{-j},s^{-j}_t,I^{-j}_t,\omega_t\bigr).
  \end{aligned}
  \]
\end{itemize}

\paragraph{Objectives.}  
Diaspora \(i\) and local \(j\) solve, respectively,
\[
\max_{\{a^i_{\rm do,t}\}_{t=1}^T}\;
\sum_{t=1}^T\mathbb{E}\bigl[u^i_{\rm do,t}\bigr],
\quad
\max_{\{a^j_{p,t}\}_{t=1}^T}\;
\sum_{t=1}^T\mathbb{E}\bigl[u^j_{p,t}\bigr].
\]

\paragraph{Evolution of Distributions and Macro‐States.}  
\[
\begin{aligned}
\mu_{\rm do,t+1}&=\Gamma_{\rm do}\bigl(\mu_{\rm do,t},a_{\rm do,t}\bigr),&
\tilde\mu_{\rm do,t+1}&=\tilde\Gamma_{\rm do}\bigl(\tilde\mu_{\rm do,t},a_{\rm do,t}\bigr),\\
m_{t+1}&=\Gamma_p\bigl(m_t,a_{p,t}\bigr),&
\tilde m_{t+1}&=\tilde\Gamma_p\bigl(\tilde m_t,a_{p,t}\bigr),\\
s_c(t+1)&=\Xi_c\bigl(s_{c,t},\mu_t\bigr),\quad
\chi_{t+1}=\Xi_\chi\bigl(\chi_t,\mu_t\bigr),\\
\omega_{t+1}&=\Xi_\omega\bigl(\omega_t,\mu_t,\rho_t\bigr),\quad
\rho_{t+1}=\Xi_\rho\bigl(\rho_t,\mu_t\bigr).
\end{aligned}
\]

This subgame highlights how diaspora funding and local responses co‐evolve within the full mean‐field ecosystem, generating the feedback loops that sustain or break cycles of violence.

The diaspora war entrepreneur subgame models the dynamic interaction between diaspora actors who fund local militarization and the local recipients who implement defense, retaliation, or mobilization strategies, all embedded within a broader conflict ecosystem influenced by the behaviors of other agents and macro-level factors. Each diaspora actor is characterized by a type (ideological, traumatized, strategic), a state capturing emotional intensity and remittance capacity, and private information often shaped by partial news, social media, or family testimony. They choose among symbolic support, non-lethal aid, or direct arms funding, with their payoffs depending on perceived impact, transfer cost, and guilt from potential violence escalation. Local recipients, similarly typed, evolve through states of trauma, perceived threat, and armament level, making decisions that range from peaceful restraint to full-scale retaliation or ethnic mobilization. Their utility balances perceived security gains, risks of reprisal, and emotional satisfaction from revenge.  Both diaspora and local agents condition their strategies not only on their own states and beliefs, but also on the joint distribution of other agents’ types, states, actions, and signals, represented by a dynamic mean-field measure $\mu_t$. These micro-interactions are further shaped by macro-states: the regional resource landscape ($\chi_t$), the conflict map and violence intensity ($\omega_t$), and geopolitical indicators ($\rho_t$), all of which evolve endogenously as functions of the agents’ collective behaviors. A diaspora funder’s decision to arm a community affects not only the local group’s military posture, but also alters perceptions of threat among rival groups, potentially provoking symmetric arming or counterattacks. These interactions create feedback loops, where retaliation by one group spurs new funding waves from its diaspora, escalating arms flows and prolonging cycles of violence. The state transitions, payoff functions, and policy optimization problems are therefore tightly coupled through population-level dynamics and cross-agent externalities.

\subsection*{Persistence of Violence under Sustained Diaspora Funding}

Within the diaspora war‐entrepreneur subgame, we formalize the conditions under which local retaliation and trauma persist over the finite horizon \(t\in \{1,\dots,T\}\).

\begin{theorem}[Persistence of Revenge and Violence]
We make the following assumptions. 
Sustained Arms Funding: 
All diaspora agents \(i\in\mathcal D\) choose the atomic action
\[
a^i_{\rm do,t}=\text{arms funding}
\quad\text{for every }t \in \{1,\dots,T\}.
\]
Under this assumption, communities continuously receive external weapon supplies.

Revenge Amplification: 
There exists a subset \(\Theta_p^{\rm rev}\subset\Theta_p\) of “revenge‐seeking” mixed types such that their revenge‐benefit function
\[
B_{\rm rev}\bigl(\theta^j_p,\,s^j_p,\,a^j_p,\,\theta^{-j},\,s^{-j}_p,\,I^{-j}_p,\omega_t\bigr)
\]
satisfies, for all \(j\) with \(\theta^j_p\in\Theta_p^{\rm rev}\),
\[
\frac{\partial B_{\rm rev}}{\partial s^j_{p,\mathrm{do}}}
>0,
\quad
\frac{\partial B_{\rm rev}}{\partial \bigl(a_{p,\mathrm{do}}^{-j}\bigr)}
>0,
\]
so that higher personal trauma \(s^j_{p,\mathrm{do}}\) and greater observed conflict actions \(a_{p,\mathrm{do}}^{-j}\in A_p^{\rm conflict}\) each strictly increase the marginal payoff to retaliation.

Trauma Dynamics: 
Each local recipient’s atomic trauma component \(s^j_{p,\mathrm{do},t}\) evolves by
\[
s^j_{p,\mathrm{do},t+1}
= s^j_{p,\mathrm{do},t}
  + \eta_1\,\mathbf{1}_{\{a^j_{p,\mathrm{do},t}\in A_p^{\rm conflict}\}}
  + \eta_2\,\mathrm{Agg}^{-j}_t
  \;-\;\delta\,s^j_{p,\mathrm{do},t},
\]
where 
\(\mathrm{Agg}^{-j}_t
:= \tilde m_t\bigl(A_p^{\rm conflict}\bigr)\)
is the population share of conflict actions, and constants \(\eta_1,\eta_2>0\), \(\delta\in(0,1)\).

Type Persistence: 
The mixed‐agent type‐transition kernel \(\Gamma_p\) satisfies, for all \(j\),
\[
\Pr\bigl(\theta^j_{p,t+1}\in\Theta_p^{\rm rev}\mid s^j_{p,\mathrm{do},t}\bigr)
\quad\text{is strictly increasing in }s^j_{p,\mathrm{do},t}.
\]
Hence rising trauma raises the chance of remaining (or becoming) a revenge‐seeker.

Under these Assumptions, suppose the initial mixed‐agent distribution \(m_1\) assigns positive mass to \(\Theta_p^{\rm rev}\) and some agents have \(s^j_{p,\mathrm{do},1}>0\).  Then for every \(t\in \{1,\dots,T\}\),
\[
m_t\bigl(\Theta_p^{\rm rev}\bigr)\;>\;0,
\]
and the share of conflict actions
\(\tilde m_t(A_p^{\rm conflict})\)
remains strictly positive.  Consequently, the system is trapped in a finite‐horizon cycle of revenge and violence.
\end{theorem}

\begin{proof}
Let \(\alpha_1=m_1(\Theta_p^{\rm rev})>0\).  Under the first Assumption, diaspora arms funding ensures local recipients have resources to choose conflict actions without external constraint.  By Assumption 2, any \(j\) with \(\theta^j_p\in\Theta_p^{\rm rev}\) and trauma \(s^j_{p,\mathrm{do},t}>0\) finds retaliation strictly optimal, so \(a^j_{p,\mathrm{do},t}\in A_p^{\rm conflict}\).  

Under Assumption 3, such conflict choices raise next‐period trauma:
\[
s^j_{p,\mathrm{do},t+1}
\;\ge\;\eta_1 + \eta_2\,\mathrm{Agg}_t - \delta\,s^j_{p,\mathrm{do},t}.
\]
Since \(\mathrm{Agg}_t\ge \alpha_1\) (the minimum conflict share) and \(\eta_1,\eta_2>0\), trauma remains positive with strictly positive probability each period.

By Assumption 4, higher \(s^j_{p,\mathrm{do},t}\) increases the chance that \(\theta^j_{p,t+1}\) remains in \(\Theta_p^{\rm rev}\).  Hence
\[
m_{t+1}(\Theta_p^{\rm rev})
\;\ge\;
m_t(\Theta_p^{\rm rev})
\;\times\;
\Pr\bigl(\theta^j_{p,t+1}\in\Theta_p^{\rm rev}\mid s^j_{p,\mathrm{do},t}\bigr)
\;>\;0.
\]
By induction on \(t\), \(m_t(\Theta_p^{\rm rev})\ge\alpha_1\prod_{k=1}^{t-1}\beta_k>0\).  Therefore \(\tilde m_t(A_p^{\rm conflict})>0\) for each \(t\), establishing persistent revenge cycles.
\end{proof}

Diaspora arms funding acts as an external amplifier of violence by injecting financial and material resources into already fragile conflict systems, enabling communities to bypass disarmament pressures and escalate local disputes into militarized confrontations. These funds, often motivated by protection or revenge, allow local actors to access weapons, organize retaliatory actions, and sustain defense structures beyond their organic capacities. This external support lowers the cost and raises the perceived legitimacy of violence, particularly for groups with historical grievances or trauma, thereby reinforcing cycles of retribution. Moreover, as rival communities observe such arming, they too mobilize defensively or preemptively, triggering an arms race that is difficult to reverse. The diaspora, operating from abroad and often detached from ground realities, may misinterpret symbolic insecurity as existential threat, leading to overfunding and persistent arming. Over time, this feedback loop hardens social divisions, entrenches militarized behavior, and reproduces revenge-oriented identities across generations, locking entire communities into self-sustaining cycles of fear, trauma, and violence.

\subsubsection{Breaking the Cycle through Incentives.}  
By aligning material and reputational rewards with de‐escalatory behavior, carefully calibrated incentives can overturn the risk–reward calculus that sustains militant payoffs and shift the population composition toward peaceful types.  Performance‐contingent transfers, indexed to observable metrics such as school attendance, cooperative water‐management outcomes, or verified reductions in local violence, augment the marginal benefit of actions in $\mathcal A^{\mathrm{peace}}$ while imposing opportunity costs on $\mathcal A^{\mathrm{conflict}}$.  When these incentives are set so that even individuals harboring a revenge attitude or deep injustice grievances attain higher expected utility by engaging in community co‐development than by resorting to retaliation, the best‐response policy for borderline types moves definitively into the nonviolent region.  As successive cohorts inherit both a growing mass of peaceful types and the tangible returns to collaboration, the mean‐field distribution $m_t$ converges to an equilibrium dominated by peaceful actions.  By reducing informational asymmetries through transparent verification of program outcomes and ensuring intergenerational retention of benefits (improved health and educational prospects for participants' children), this incentive architecture creates a self‐sustaining cycle in which cooperation outcompetes the old equilibrium of cyclical violence.  

\subsubsection*{Incentive‐Design to Drive Decay of Violence Mass}

We embed a \emph{peace‐incentive} $\pi_{\rm peace}$ into the local recipients’ payoff, augmenting $u^{j}_{p,t}$ by rewarding non‐violent actions.  Under this modified utility, we show the mass of revenger types $\Theta_{p}^{\rm rev}$ decays geometrically over the finite horizon.

\begin{theorem}[Geometric Decay of Revenger Mass]

Suppose Assumptions from the persistence result hold and 
Peace Incentive: 
There exists a bounded continuous function 
\[
\pi_{\rm peace}\bigl(\theta^j_p,s^j_p,a^j_p,\mu_t,\omega_t\bigr)
\]
such that for all $\theta^j_p\in\Theta_p$, 
\[
\pi_{\rm peace}(\cdot,a_{p}^{\rm peace},\mu_t,\omega_t)
\;-\;
\pi_{\rm peace}(\cdot,a_{p}^{\rm conflict},\mu_t,\omega_t)
\;\ge\;\gamma>0,
\]
whenever $a_{p}^{\rm peace}\in A_{p}^{\rm peace}$ and $a_{p}^{\rm conflict}\in A_{p}^{\rm conflict}$.

Reversal of Best‐Response: 
For each revenger type $\theta^j_p\in\Theta_p^{\rm rev}$, there exists an incentive strength $\gamma^\star>0$ such that when $\pi_{\rm peace}$ satisfies Assumption with $\gamma\ge\gamma^\star$, the unique best‐response policy of $\theta^j_p$ for all $s^j_{p,\mathrm{do}}$ and mean‐fields shifts into $A_{p}^{\rm peace}$.

Healing Dynamics: 
With peaceful action $a_{p}^{\rm peace}$, the trauma component decays:
\[
s^j_{p,\mathrm{do},t+1}
= (1-\delta)\,s^j_{p,\mathrm{do},t},
\quad
\delta\in(0,1).
\]

Type Reversion:
The type‐transition kernel $\Gamma_p$ satisfies for $\theta^j_p\in\Theta_p^{\rm rev}$:
\[
\Pr\bigl(\theta^j_{p,t+1}\notin\Theta_p^{\rm rev}\mid s^j_{p,\mathrm{do},t}\bigr)
\;\ge\;\kappa\,s^j_{p,\mathrm{do},t},
\]
for some $\kappa>0$, so that lower trauma increases chance of becoming peaceful.

  Then, choosing the peace incentive strength $\gamma\ge\gamma^\star$, the revenger mass
\(\;m_t(\Theta_p^{\rm rev})\)
decays at least geometrically:
\[
m_{t+1}(\Theta_p^{\rm rev})
\;\le\;(1 - \kappa\,\delta)\;m_t(\Theta_p^{\rm rev}),
\]
for every $t\in \{1,\dots,T-1\}$.  Hence
\[
m_t(\Theta_p^{\rm rev})
\;\le\;(1 - \kappa\,\delta)^{\,t-1}\,m_1(\Theta_p^{\rm rev})
\;\longrightarrow\;0
\quad\text{as }t\to T.
\]
\end{theorem}

\begin{proof}
By assumption every revenger type $\theta^j_p\in\Theta_p^{\rm rev}$ chooses $a_{p}^{\rm peace}$ once $\gamma\ge\gamma^\star$.  By Assumption , their trauma decays:
\[
s^j_{p,\mathrm{do},t+1}
= (1-\delta)\,s^j_{p,\mathrm{do},t}.
\]
Then, for each such agent, the probability of \emph{remaining} in $\Theta_p^{\rm rev}$ satisfies
\[
\Pr\bigl(\theta^j_{p,t+1}\in\Theta_p^{\rm rev}\mid s^j_{p,\mathrm{do},t}\bigr)
\;\le\;1 - \kappa\,s^j_{p,\mathrm{do},t}.
\]
Taking expectation over the subpopulation $m_t\vert_{\Theta_p^{\rm rev}}$,
\[
m_{t+1}(\Theta_p^{\rm rev})
\;\le\;
\int_{\Theta_p^{\rm rev}} 
\bigl[1 - \kappa\,s^j_{p,\mathrm{do},t}\bigr]
\,dm_t(\theta^j_p)
\;=\;
m_t(\Theta_p^{\rm rev})
\;-\,\kappa\;
\mathbb E_{m_t}[\,s^j_{p,\mathrm{do},t}\,\bigm|\Theta_p^{\rm rev}\bigr].
\]
But since under peace $s^j_{p,\mathrm{do},t}\) decays from its initial value and remains positive, we have
\(\mathbb E_{m_t}[s^j_{p,\mathrm{do},t}]\ge \delta\,m_t(\Theta_p^{\rm rev})\).
Thus
\[
m_{t+1}(\Theta_p^{\rm rev})
\;\le\;
m_t(\Theta_p^{\rm rev})
\;-\;\kappa\,\delta\;m_t(\Theta_p^{\rm rev})
\;=\;
(1-\kappa\,\delta)\;m_t(\Theta_p^{\rm rev}),
\]
establishing the geometric contraction.  Iterating yields the claimed bound and convergence to zero by finite horizon \(T\).
\end{proof}

\begin{example}[Funding the incentive mechanism]
Funding performance‐contingent incentives  can draw on a diversified portfolio of domestic and international resources. Reallocating a portion of national and local government budgets, specifically from security and emergency relief lines, into a conditional grants program ties peace dividends directly to de‐escalatory behavior. Multilateral donors, bilateral partners, and UN agencies can underwrite matched‐fund schemes that scale with local revenue generation. Using diaspora remittances through ``peace bonds" or matched‐savings accounts incentivizes expatriates to co‐finance schooling, health clinics, or micro‐irrigation projects in their home communities. Public-private partnerships with mining, telecom, and agribusiness firms can direct a share of corporate social responsibility funds into outcome‐based grants, while innovative instruments such as social or co-development impact bonds, where investors front capital and are repaid by donors or governments only upon verified results, align risk and accountability. Carbon‐credit revenues from reforestation or sustainable land‐use projects offer a long‐term financing stream for community incentives, embedding co‐development within global environmental markets. By combining budgetary realignment, donor matching, private co‐investment, and market‐based finance, Mali can ensure robust, transparent, and sustainable funding for the incentive architecture that breaks cycles of violence. Efficient implementation hinges on decentralizing decision‐making, leveraging digital delivery, and embedding robust feedback loops:
Establish local steering committees, comprised of community leaders, civil‐society representatives, and peace‐incentive recipients, to co‐design performance metrics (school enrolment rates, communal labor days) and oversee disbursement. Use mobile‐money platforms  to transfer conditional cash or vouchers directly to households and service providers, minimizing leakage and administrative overhead. Partner with local NGOs and cooperatives to conduct real‐time monitoring via simple SMS or smartphone surveys, feeding data into a centralized dashboard that flags under‐performance or irregularities for rapid correction. Contract independent third‐party verifiers for instance local university research teams or accredited auditors, to validate outcomes at pre‐scheduled intervals, ensuring that payments reflect genuine community gains. Build a lightweight data‐management system (cloud‐hosted, open‐source GIS and mobile forms) so that performance against targets, budget utilization, and grievance submissions are transparent to all stakeholders. Schedule quarterly ``learning fora" where results are reviewed by donors, government, and communities jointly, allowing incentive parameters to be adaptively tuned (raising match‐rates for lagging sectors or introducing new peace‐building indicators). By combining community ownership, digital transfers, rigorous monitoring, and adaptive governance, funded incentives can be scaled across Mali’s regions with both speed and accountability. 
\end{example}

\begin{tikzpicture}[
    font=\sffamily,
    >=Stealth,
    node distance=0.2cm
]

    \draw[->, thick] (0,0) -- (12,0) node[below=5pt, pos=0.5, font=\bfseries] {Time (or Cumulative Actions/Stress)};
    \draw[->, thick] (0,0) -- (0,8) node[rotate=90, left=15pt, pos=0.5, font=\bfseries] {Cumulative System Stress $S_t$};
    
    \coordinate (Scrit) at (0,6);
    \draw[dashed, ultra thick, CriticalRed] (0,6) -- (11.5,6) node[above, pos=0.9, font=\bfseries\Large, color=CriticalRed] {$S_{\rm crit}$};
    \node[anchor=east, CriticalRed, font=\bfseries, align=right] at (Scrit) {Critical Threshold\\(Regime Shift Limit)};

    \draw[->, ultra thick, StressGray!60] (0,0) -- (4,2.5) coordinate (A) node[midway, below=3pt, rotate=28, font=\scriptsize, color=StressGray] {Latent Accumulation (Invisible Build-up)};
    
    \draw[->, ultra thick, ActionBlue] (A) -- (5,3.125) coordinate (B) node[midway, above left, font=\scriptsize\bfseries] {Marginal\\Action 1};
    \draw[dotted, thick, StableGreen] (B) -- (11,3.125);
    \node[anchor=west, font=\scriptsize, StableGreen, align=left] at (11,3.125) {Observable Shift: \textbf{Zero}\\Result: \textbf{Stable Regime}};

    \draw[->, ultra thick, StressGray!60] (B) -- (8.5,5.5) coordinate (C) node[midway, below=3pt, rotate=28, font=\scriptsize, color=StressGray] {Ecological Load / Political Tension};

    \fill[CriticalRed!10] (8.5,5.5) rectangle (11.5,6.5);
    \node[anchor=south west, font=\scriptsize, color=CriticalRed, align=left] at (8.5,6.0) {Narrow Corridor\\of Instability};
    
    \draw[->, ultra thick, ActionBlue] (C) -- (9.3,6.02) coordinate (D) node[midway, above left, font=\scriptsize\bfseries] {Marginal\\Action 2};
    \node[circle, fill=ActionBlue, inner sep=2pt, label={[ActionBlue, font=\bfseries]right:Trigger Point ($S_t \ge S_{\rm crit}$)}] at (D) {};

    \node[draw, rectangle, rounded corners, inner sep=5pt, align=left, fill=StressGray!10] (fragility) at (7, 1.5) {
        \textbf{Cumulative Stress $S_t$ encodes Fragility:}\\
        - Ecological load (Mali 2006, Burkina 2017)\\
        - Financial leverage (Mali 2012 Coup)\\
        - Political tension (Mali 2020/21 Coup)\\
        - Social frustration (Mali 2017 Killings)
    };
    \draw[->, thick, StressGray] (4,1.5) to[bend left] (fragility);
    
    \draw[->, thick, CriticalRed, line width=1.5pt] (D) -- (10.5, 0.2) node[midway, right=3pt, font=\bfseries, color=CriticalRed, align=left] {Sudden, Discontinuous transition\\Catalyzed by otherwise\\insignificant action};
    
    \draw[dotted, thick, CriticalRed] (9.3,0.2) -- (11.5,0.2);
    \node[anchor=east, font=\scriptsize, CriticalRed, align=right] at (9.3,0.2) {Result: \textbf{Crisis Regime}};

    \begin{scope}[font=\scriptsize]
        \node[above=0.3cm of A] (feat1) {\textbf{1. Individual actions are context-dependent}};
        \node[below=0.1cm of C, StressGray] (feat2) {\textbf{2. There exists a cumulative stress variable $S_t$}};
        \node[left=0.2cm of Scrit, CriticalRed, anchor=south west] (feat3) at (0,6) {\textbf{3. A critical threshold $S_{\rm crit}$ exists}};
        \node[below=0.1cm of D, CriticalRed] (feat4) {\textbf{4. Magnification of small actions}};
    \end{scope}

\end{tikzpicture}

\subsection{The Straw that Breaks the Camel's Back}  
We use the metaphor of the \emph{``straw that breaks the camel’s back''}  (see Tables \ref{camelfig1}, \ref{camelfig2}, \ref{camelfig3}, \ref{camelfig4}, \ref{camelfig5}) to describe systems in which cumulative stress evolves over time and exhibits abrupt, discontinuous transitions once a critical threshold is reached. This phenomenon can be characterized by four defining features.
\emph{Individual actions are typically  context-dependent}. In ordinary conditions, a single agent’s contribution produces no observable shift in the aggregate outcome. However, when the system is close to a critical point, even the most marginal action can exert a disproportionate effect. Thus, the impact of an action is endogenous: it depends not only on its magnitude but also on the latent vulnerability of the system. \emph{There exists a cumulative stress variable}, denoted $S_t$, that aggregates the contributions of all agents and external shocks. This variable encodes the system’s fragility, capturing ecological load, financial leverage, political tension, or social frustration, depending on the context.
 \emph{A critical threshold $S_{\rm crit}$ exists}, beyond which the system undergoes a regime shift. When $S_t \geq S_{\rm crit}$, the dynamics exhibit nonlinearity and discontinuity: the system may collapse, erupt into unrest, or cascade into crisis. This transition is both sudden and disproportionate relative to the final incremental action that triggered it.
\emph{Any individual action may serve as the trigger once the threshold is approached}. Unlike models where only atomic players or large shocks matter, here even infinitesimal contributions can precipitate the transition if they arrive when $S_t$ lies within the narrow corridor of instability near $S_{\rm crit}$. The criticality of the system magnifies small actions into decisive events.
A latent accumulation of stress, invisible in its gradual build-up, culminates in a sudden, discontinuous transition catalyzed by an action that would otherwise be insignificant.

\begin{table}[H]
\centering
\small
\begin{tabular}{|p{1.2cm}|p{3.0cm}|p{5.0cm}|p{6.0cm}|}
\hline
\textbf{Year} & \textbf{Location} & \textbf{Trigger ("drop")} & \textbf{Consequences} \\
\hline
2006 & Kidal & Tuareg uprising over unfulfilled peace accords & Renewed rebellion in the north, exposing fragility of prior agreements. \\
\hline
2012 & Kati (near Bamako) & Soldiers’ mutiny over poor equipment & Coup d’État ousted President Amadou Toumani Touré, collapse of state authority. \\
\hline
2012 & Azawad (North) & Tuareg and jihadist takeover of major cities & De facto partition of the country, prompting international concern. \\
\hline
2013 & Konna & Jihadist advance on the town & Triggered French intervention (Operation Serval). \\
\hline
2015 & Bamako & Radisson Blu hotel terrorist attack & International shock, symbol of jihadist penetration in the capital. \\
\hline
2017 & Mopti & Killings between several communities & Marked the escalation of intercommunal violence in central Mali. \\
\hline
2020 & Bamako & M5-RFP protest movement after contested elections & Nationwide mobilization leading to coup against President Keïta. \\
\hline
2021 & Bamako & Arrest of civilian leaders by military & “Coup within a coup,” consolidating military power under Assimi Goïta. \\
\hline
2022 & Moura  & Military operation with reported mass killings & Alleged massacre of hundreds of civilians; international condemnation. \\
\hline
\end{tabular}
\caption{“Straw” events in Mali (2006–2022).} \label{camelfig1}
\end{table}

\begin{table}[h!]
\centering
\small
\begin{tabular}{|p{1.2cm}|p{3.0cm}|p{5.0cm}|p{6.0cm}|}
\hline
\textbf{Year} & \textbf{Location} & \textbf{Trigger ("drop")} & \textbf{Consequences} \\
\hline
2016 & Nampala (Ségou Region) & Armed assault on army camp (several soldiers killed) & Major loss of state control; army morale deeply shaken. \\
\hline
2017 & Gao & Suicide bombing at MOC camp (peace process site) & several killed; collapse of fragile trust in peace accords. \\
\hline
2018 & Mopti Region & Burning of a passenger bus on market day, civilians trapped inside & Dozens burned alive; symbol of jihadist terror against civilians. \\
\hline
2018 & Indelimane (Ménaka) & Jihadist attack on army post & Dozens of soldiers killed; revealed fragility of eastern front. \\
\hline
2019 & Sobane Da (Dogon village, Mopti) & Massacre of several villagers & Marked extreme escalation of intercommunal killings. \\
\hline
2020 & Farabougou (Ségou Region) & Village siege by armed groups & Months-long blockade; emblematic of state’s impotence. \\
\hline
2021 & Tessit (Gao Region) & Coordinated jihadist assault on army base & Dozens of deaths; turning point in tri-border insecurity. \\
\hline
2021 & Kayes Region & Arrival of armed groups in western Mali & Conflict expansion beyond traditional northern/central zones. \\
\hline
2022 & Bankass & Armed groups attack market day & Dozens massacred; economic and social collapse in the cercle. \\
\hline
2022 & Mondoro & Deadly assault on Malian troops (several dead) & Showed limits of army and Wagner partnership in rural defense. \\
\hline
\end{tabular}
\caption{Other “Straw” events in Mali (2016-2022).} \label{camelfig2}
\end{table}

\begin{table}[H]
\centering
\small
\begin{tabular}{|p{1.2cm}|p{3.0cm}|p{5.0cm}|p{6.0cm}|}
\hline
\textbf{Year} & \textbf{Location} & \textbf{Trigger ("drop")} & \textbf{Consequences} \\
\hline
2014 & Ouagadougou & Compaoré’s plan for a 3rd term via constitutional change & Mass protests forced his resignation after 27 years in power. \\
\hline
2014 & National & Sankara’s memory revived in demonstrations & Mobilized youth, delegitimized the regime, symbol of rupture. \\
\hline
2015 & Ouagadougou & Attempted coup by RSP (Presidential Guard) & Civil resistance defeated the coup, accelerating army reform. \\
\hline
2016 & Ouagadougou & Terrorist attack at Cappuccino Café & First large-scale jihadist attack in capital; shift in perception of insecurity. \\
\hline
2017 & Soum Province & Nascent jihadist insurgency in north & Start of continuous asymmetric attacks across the country. \\
\hline
2018 & Inata & Targeted attack on gendarmerie post & Dozens of security forces killed; exposed state’s vulnerability. \\
\hline
2019 & Yirgou (Centre-Nord) & Intercommunal massacre after local killing & Sparked a spiral of reprisals and entrenched community violence. \\
\hline
2020 & Gaoua (Southwest) & Local land dispute between farmers and herders & Escalated into deadly intercommunal clashes in previously stable zone. \\
\hline
2021 & Solhan (Yagha Province) & Massacre by armed groups (over 160 dead) & Deadliest attack in Burkina Faso’s history, national shock. \\
\hline
2022 & Ouagadougou & Military dissatisfaction with handling of jihadist threat & Two coups d’État in same year; regime instability deepened. \\
\hline
\end{tabular}
\caption{“Last straw” events in Burkina Faso (2010–2022).} \label{camelfig3}
\end{table}

\begin{table}[H]
\centering
\small
\begin{tabular}{|p{1.2cm}|p{3.0cm}|p{5.0cm}|p{6.0cm}|}
\hline
\textbf{Year} & \textbf{Location} & \textbf{Trigger ("drop")} & \textbf{Consequences} \\
\hline
2007 & Agadez Region & Tuareg grievances over mining revenues & Sparked renewed rebellion (MNJ), highlighting north–south divide. \\
\hline
2009 & Niamey & Tandja’s attempt to extend mandate (constitutional change) & Political crisis leading to 2010 coup d’État. \\
\hline
2010 & Niamey & Military coup against Tandja & Transition government formed, suspension of democratic institutions. \\
\hline
2011 & Diffa & First Boko Haram incursions across border & Opened new front of jihadist violence in Niger’s southeast. \\
\hline
2013 & Agadez & Twin attacks (military base and Areva uranium site) & Exposed Niger’s strategic vulnerability despite foreign presence. \\
\hline
2015 & Niamey, Zinder & Violent protests over Charlie Hebdo cartoons & Churches and French interests attacked; revealed religious tensions. \\
\hline
2017 & Tongo Tongo (Tillabéri) & Ambush killing U.S. and Nigerien soldiers & Exposed jihadist strength, internationalized the conflict. \\
\hline
2019 & Inates (Tillabéri) & Mass attack on military camp (over 70 dead) & Major army setback, triggered public anger against government. \\
\hline
2020 & Kouré & Assassination of French humanitarian workers & Shock event underscoring insecurity even near capital. \\
\hline
2023 & Niamey & Coup d’État against President Bazoum & End of democratic experiment, deepened regional crisis. \\
\hline
\end{tabular}
\caption{“Last straw” events in Niger (2007–2023).} \label{camelfig4}
\end{table}

\begin{table}[H]
\centering
\small
\begin{tabular}{|p{1.2cm}|p{2.2cm}|p{4.2cm}|p{6.8cm}|}
\hline
\textbf{Year} & \textbf{Country} & \textbf{Trigger ("drop")} & \textbf{Consequences} \\
\hline
2010 & Tunisia & Food price increase & Sparked nationwide protests leading to the Arab Spring. \\
\hline
2010 & Tunisia & Self-immolation of Mohamed Bouazizi & Triggered regime change and regional uprisings. \\
\hline
2014 & Burkina Faso & Attempted constitutional change (3rd term) & Mass protests forced President Compaoré to resign. \\
\hline
2013 & Mali & Jihadist attack on Konna & Led to French military intervention (Operation Serval). \\
\hline
2014 & Burkina Faso & Sankara legacy revived in protests & Mobilized youth and delegitimized the regime. \\
\hline
2010-11 & Côte d’Ivoire & Gbagbo refusal to concede defeat & Post-election violence killed thousands. \\
\hline
2015 & Nigeria & Boko Haram massacre in Baga & Psychological rupture; army credibility shattered. \\
\hline
2012 & Mali & Attempted targeting of imam Haïdara & Heightened sectarian tensions in Bamako. \\
\hline
2017-19 & West Africa & Anti-CFA franc protests & Reopened debate on monetary sovereignty. \\
\hline
2019 & Burkina Faso & Yirgou intercommunal massacre & Sparked cycle of reprisals and escalating violence. \\
\hline
2006 & Cameroon & Student protest over caricatures & Spread into violent riots with religious undertones. \\
\hline
2017 & Côte d’Ivoire & Bouaké soldiers’ mutiny & Army unrest paralyzed the state. \\
\hline
2014 & DRC & Killing of a traditional leader in Beni & Renewed massacres and community war in Kivu. \\
\hline
2005-09 & South Africa & Zuma corruption and rape trial & Political reshuffling; Mbeki ousted. \\
\hline
2013 & Italy/Africa & Lampedusa migrant boat tragedy (26 deaths) & Migration crisis gained global attention. \\
\hline
2012 & Mali & Soldiers’ mutiny in Kati & Coup d’État ousted President Touré. \\
\hline
2019 & DRC & Death of singer Kabongo & Popular grief turned into political protest. \\
\hline
2020 & Burkina Faso & Land dispute in Gaoua & Intercommunal violence escalated in the south-west. \\
\hline
2012 & Nigeria & Fuel subsidy removal & “Occupy Nigeria” mass movement paralyzed economy. \\
\hline
2018 & DRC & Death of activist Luc Nkulula & Radicalized youth movements against Kabila. \\
\hline
\end{tabular}
\caption{Some  "Straw" events in Africa (2000-2020).} \label{camelfig5}
\end{table}

 In a population of 80 million, the phenomenon of \emph{``the straw that breaks the camel’s back"} corresponds, in mean-field terms, to a system in which a single state, a single action, an individual state distribution \(\mu_t\) and  an individual action distribution \(\tilde\mu_t\) are going to be important when 
population state distribution \(m_t\) and action distribution \(\tilde{m}_t\) are concentrated near critical thresholds of tolerance or trauma. When a single individual agent regardless of its class or type, takes a provocative or aggressive action \(a^{i^*}_t\), this initiates a shift of the distribution \(\tilde{m}_t\) measurably, potentially crossing a tipping point where the best responses 
\(a^{\star}(s, I, \mu_t, \tilde\mu_t, m,\tilde{m}_t)\) for many agents switch from peaceful to retaliatory. This results in a cascade, updating \(m_{t+1}\) and \(\tilde m_{t+1}\) toward increased aggression. The mean-field-type feedback loop amplifies the single action’s impact, making it the effective trigger for mass unrest or escalation, demonstrating how, even in a large system, individual actions can shift equilibria when the system is already near critical instability. The phenomenon can be examined using singularity formation of the mean-field knowing everyone can be actor of the singularity formation. The formal analogue of the proverb now becomes: {\it when the collective frustration or instability is full, even a tiny drop (one action from a one single individual agent) makes it overflow}

{\bf Tipping Point from Individual Action: } \emph{In a population with a highly concentrated distribution of agents near a threshold of frustration or trauma, a single strategic action by one agent can be sufficiently influential and can trigger a system-wide shift from a peaceful configuration to a violent one. In other words, when the collective frustration or instability is full, even a tiny drop (one action) makes it overflow.}

\begin{theorem}[Cascade Over \(k\) Periods of generation $g$]
Consider $N+1$ agents.  Each agent’s binary action \(a\in\{0,1\}\) denotes peace (\(0\)) or conflict (\(1\)).  Mixed‐agent trauma \(s^j_t\) is identical across \(j\), denoted \(s_t\).  

The weighted conflict fraction is
$
F_t \;=\;\frac{w\,a^{i^*}_t + \sum_{j=1}^N a^j_t}{\,w+N\,},
\quad w>1.$
The trauma dynamics $
s_{t+1} \;=\; s_t \;+\;\beta\,F_t,\quad \beta>0.$
Fix constants \(B,C>0\) and define the conflict threshold \(\tau = C/B\).  At decision time \(t\), each mixed agent \(j\) in state \(s_t\) has
\[
u^j_t(a; s_t)
\;=\;
a\bigl(B\,s_t - C\bigr),
\]
so that playing \(a=1\) yields \(B\,s_t - C\) and \(a=0\) yields \(0\).
Assume at time \(t\):
\[
a^{i^*}_t=1,\quad a^j_t=0\ \forall j,\quad
s_t \in \Bigl[\tau - \tfrac{\beta\,w}{w+N},\;\tau\Bigr).
\]
Then for every mixed agent \(j\) and every \(\ell\in \{1,2,\dots,k\}\), where \(k=T-t\),
\[
a^j_{t+\ell}=1,\quad F_{t+\ell}=1,
\]
and the trauma satisfies
\[
s_{t+\ell}
= s_t + \frac{\beta\,w}{w+N} + (\ell-1)\,\beta
\;\ge\;\tau.
\]
Thus conflict persists from \(t+1\) through \(t+k\) if there is other external intervention.
\end{theorem}

\begin{proof}
 At \(t\), mixed agents play \(a^j_t=0\), atomic plays \(a^{i^*}_t=1\).  Hence
\[
F_t = \frac{w\cdot1 + 0}{w+N} = \frac{w}{w+N},
\]
and
\[
s_{t+1}
= s_t + \beta\,F_t
= s_t + \frac{\beta\,w}{w+N}.
\]
By assumption \(s_t \ge \tau - \tfrac{\beta\,w}{w+N}\), so
\[
s_{t+1} \ge \tau.
\]

 At the decision epoch \(t+1\), each mixed agent faces utility
\[
u^j_{t+1}(1; s_{t+1})
= B\,s_{t+1} - C
\;\ge\;
B\,\tau - C
= 0,
\]
while \(u^j_{t+1}(0;s_{t+1})=0\).  Thus \(a^j_{t+1}=1\), so
\[
F_{t+1}
= \frac{\,w\,a^{i^*}_{t+1} + \sum_j a^j_{t+1}\,}{w+N}
\;\ge\;
\frac{0 + N}{w+N}
> 0.
\]
Even if \(a^{i^*}_{t+1}=0\), the \(N\) mixed agents yield \(F_{t+1}=N/(w+N)\ge1/2\) since \(w>1\).
We use an Induction for \(\ell\ge2\).  Suppose for some \(\ell\ge1\),
\[
a^j_{t+\ell}=1,\quad F_{t+\ell}=1,\quad
s_{t+\ell}
= s_t + \frac{\beta\,w}{w+N} + (\ell-1)\,\beta
\;\ge\;\tau.
\]
Then
\[
s_{t+\ell+1}
= s_{t+\ell} + \beta\,F_{t+\ell}
= s_t + \frac{\beta\,w}{w+N} + \ell\,\beta
\;\ge\;\tau,
\]
and at decision time \(t+\ell+1\),
\[
u^j_{t+\ell+1}(1; s_{t+\ell+1})
= B\,s_{t+\ell+1} - C
\;\ge\;
B\,\tau - C = 0.
\]
Hence \(a^j_{t+\ell+1}=1\) and thus \(F_{t+\ell+1}=1\).  By induction, this holds for all \(\ell\in \{1,\dots,k\}\).

Therefore conflict actions persist through \(t+k\), proving the claim.
\end{proof}

\subsection{Localized  Attacker Strategies}

\subsubsection*{ Simultaneous  Series Attacks  on Market Day}
\begin{figure}[htb!]
\begin{tikzpicture}[scale=0.6]

\fill[gray!20] (-1,-2.5) rectangle (17,2.5);

\foreach \x in {0.5,1.5,...,16}
  \draw[dashed, white, line width=1pt] (\x,-0.05) -- (\x+0.5,-0.05);

\node[draw, rectangle, minimum width=1cm, minimum height=1.5cm, fill=yellow!40] at (17.5, 0) {\Large\textbf{D}};

\foreach \x in {0,0.8,...,16} {
  \node at (\x, 1.8) {\faMotorcycle};
  \node at (\x+0.3, 0.6) {\faMotorcycle};
  \node at (\x, -0.6) {\faMotorcycle};
  \node at (\x+0.3, -1.8) {\faMotorcycle};
}

\node at (8, -2.3) {\small \textbf{Dense 4-lane motorbike convoy heading to destination D}};

\end{tikzpicture}
\caption{Attackers  convoys to Village D.} \label{convoy1re}
\end{figure}
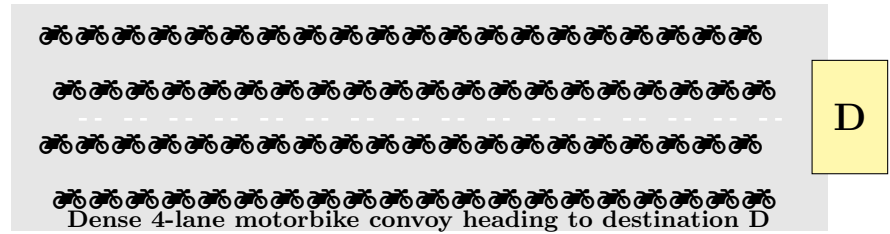

In Sahelian zones such as  Mali, Burkina Faso, Niger and Chad coordinated attacks by armed bandit groups often exploit the vulnerability and predictability of rural market days. These assailants, frequently mobilized on dozens or even hundreds of motorcycles in staggered, multi-lane formations resembling military convoys, approach targeted villages such as  "D" ( see Figure \ref{convoy1re}) with the intent to loot, abduct, or sow terror. The use of fast-moving two-wheeled vehicles grants high mobility through rural tracks and bush corridors, enabling both rapid assembly and dispersal before state or community-based security responses can be mobilized. In the present configuration, a dense formation of motorcycles across four parallel lanes en route to the village reflects an organized and possibly rehearsed operation, suggesting premeditation and group hierarchy, possibly inspired by or coordinating with insurgent logistics. While the movement of motorbike convoys can blend into the legitimate flux of rural market traffic, especially on designated market days when commercial activity peaks, it also presents an opportunity for detection and early warning. In particular, low-cost civilian drones equipped with basic optical and infrared imaging can provide overhead surveillance of major roadways and chokepoints leading to rural markets. These drones, operating autonomously or via mesh networks, can differentiate between typical market-bound civilian traffic which tends to be dispersed, slow-moving, and multi-modal (including pedestrians, donkeys, and carts)  and abnormally dense, homogeneous motorbike clusters converging rapidly toward a singular destination. Key features such as motorbike density, speed, alignment, and formation can be used as precursors to flag potential threats. 
Using lightweight machine learning algorithms on-board or at the edge, such systems can trigger automated alerts to village authorities or mobile defense units. While false positives remain a concern on days of high legitimate traffic, especially when youth travel to market in groups, pattern-of-life data and anomaly detection techniques can reduce these errors over time. The strategic utility of cheap drones thus lies not merely in aerial imagery, but in building a persistent layer of local situational awareness, a vital capability in zones where state presence is minimal, and where asymmetric actors exploit terrain, timing, and surprise. The depiction of the convoy to village D illustrates both the tactical challenge of detecting disguised mobility and the promise of low-cost, civilian-grade technologies in strengthening anticipatory security for vulnerable populations.

{\bf If the timing of the attacks is known in advance due to multiple threats made by the attackers in the villages prior to execution, what is the benefit of detecting a coordinated motorbike assault on market day?}

Even when attackers issue threats in advance, detecting a coordinated motorbike assault on market day remains critical from a mean-field-type game-theoretic perspective, as it transforms uncertain signals into credible threats, altering the strategic equilibrium. In this intergenerational games, threats alone are cheap talk unless validated by intelligence (drones, land-based informants, and satellite imagery) provide this validation, enabling defenders to update their beliefs and respond adaptively. Detection changes the payoff: it raises the expected cost of attack by allowing for preemptive force deployment, targeted ambushes, and civilian evacuation, thus shifting the mixed-strategy equilibrium in favor of defenders. Given that market day offers attackers high-impact targets, the real-time tracking of coordinated motorcycle groups limits their ability to synchronize, disrupting multi-node strategies common in Stackelberg-type games. Moreover, detection enables dynamic resource allocation and strategic deception in multi-site multi-country operations, maximizing the defender's expected utility under a Bayesian MFTG. Even with known timing, detection systems provide decisive informational advantage, reduce civilian risk, and deter attackers by raising the operational risk embedded in their payoff structure.

The National Security Services (NSS), even with limited capacity for localized defense, can strategically exploit early detection of coordinated motorbike movements (30 or more bikes heading in the same direction on market day) through a combination of game-theoretic mechanisms: incentive design, strategic signaling (with anonymized  crowd sensing by mobile apps) , temporary alliance-building, and intergenerational payoff alignment. Detection transforms the game from one of pure defense to one of preemptive influence. This allows NSS to shift from reactive strategies to proactive positioning by redistributing limited forces based on probabilistic threat maps.  NSS can create uncertainty in attackers’ payoff structures by deploying mobile units unpredictably, increasing the attackers’ expected costs and disrupting their coordination. Local alliances with village-based intelligence cells function as critical force multipliers. Through cooperative mean-field-type game theory, NSS can forge credible, low-cost coalitions with local actors: elders, hunters’ associations, youth groups by offering them negotiation power in resource access (market security, development inputs) and selective protection guarantees. This incentivizes information sharing and decentralized vigilance, helping close intelligence gaps. By building a reputation for targeted, proportional response, the NSS can influence attackers’ beliefs in Bayesian MFTG, deterring future threats through a perceived increase in the likelihood of interception.  Through intergenerational incentive structures, such as community resilience programs tied to security cooperation (early warning linked to school access, agricultural protection, or youth employment), NSS aligns long-term community interests with national stability. This shifts the local population’s strategy from passive observation to active deterrence partnerships.

\subsubsection*{Star-Shaped Simultaneous Attacks on Market Day}

Mali, Burkina Faso, Niger, Chad, Benin, Togo, Nigeria and neighbouring countries: In this vast region and neighboring countries, attacks carried out by armed groups traveling by motorcycle have, for several years now, posed a persistent threat. Initial offensives followed a relatively identifiable pattern: assailants moved in tight groups, often in a visible and coordinated manner. This configuration, in some cases, allowed for the rapid detection of suspicious movements and even the triggering of local alerts, sometimes anonymously.

But this tactic has gradually evolved. According to data compiled by Timadie over the past five years, more than 300 attacks have been recorded in connection with market days, pointing to a strategy deliberately focused on high-traffic locations. Today, assaults are carried out in a decentralized manner: individuals move in a dispersed fashion, blending in with the population, and strike simultaneously in multiple localities. This dispersal makes anticipating and responding to such threats particularly difficult, especially in contexts where intervention capabilities are limited.

This tactical evolution reflects a logic associated with mean-field game theory. 

The pressing question is thus: what form should the response take? Faced with simultaneous and distributed attacks, often crossing national borders:  defensive strategies must also evolve. This requires not only improved regional coordination but also the strengthening of community alert mechanisms, the integration of predictive analysis tools, and the networking of local actors capable of acting autonomously yet coherently. A distributed defense approach, mirroring the threat itself, now appears essential to adapt to this new grammar of conflict.

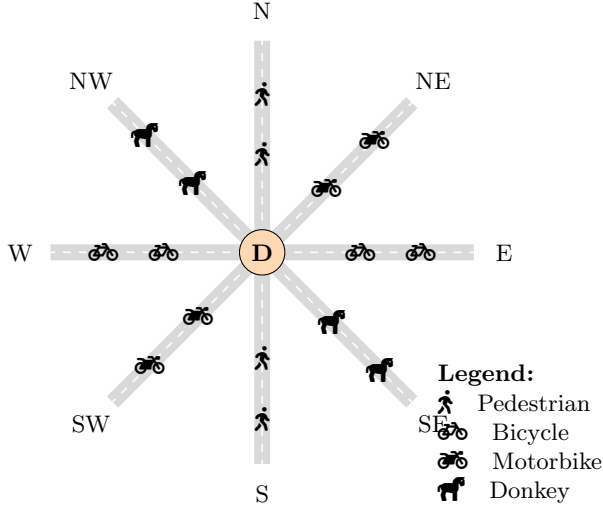
\begin{figure}[htb!]
\begin{tikzpicture}[scale=1, font=\small]

\foreach \angle in {0, 45, ..., 315} {
  \draw[gray!30, line width=6pt] (0,0) -- ({2.8*cos(\angle)}, {2.8*sin(\angle)});
  \draw[dashed, white, line width=0.6pt] (0,0) -- ({2.8*cos(\angle)}, {2.8*sin(\angle)});
}

\fill[orange!30, draw=black] (0,0) circle (0.3);
\node at (0,0) {\textbf{D}};


\node at (0, 1.3) {\faWalking};
\node at (0, 2.1) {\faWalking};
\node at (0, -1.4) {\faWalking};
\node at (0, -2.2) {\faWalking};

\node at (1.3, 0) {\faBicycle};
\node at (2.1, 0) {\faBicycle};
\node at (-1.3, 0) {\faBicycle};
\node at (-2.1, 0) {\faBicycle};

\foreach \shift in {1.2, 2.1} {
  \node at ({\shift*cos(45)}, {\shift*sin(45)}) {\faMotorcycle};
  \node at ({\shift*cos(225)}, {\shift*sin(225)}) {\faMotorcycle};
}

\foreach \shift in {1.3, 2.2} {
  \node at ({\shift*cos(135)}, {\shift*sin(135)}) {\faHorse};
  \node at ({\shift*cos(315)}, {\shift*sin(315)}) {\faHorse};
}

\foreach \angle/\label in {
  90/N, 45/NE, 0/E, 315/SE, 270/S, 225/SW, 180/W, 135/NW} {
    \node at ({3.2*cos(\angle)}, {3.2*sin(\angle)}) {\label};
}

\node[anchor=west, align=left] at (2.2,-2.4) {
  \textbf{Legend:} \\
  \faWalking\quad Pedestrian \\
  \faBicycle\quad Bicycle \\
  \faMotorcycle\quad Motorbike \\
  \faHorse\quad Donkey
};

\end{tikzpicture}
\caption{Attackers  are distributed towards  Village D.} \label{convoy2re}
\end{figure}

We now look at Village D, situated at the center of a star-shaped road network with eight radial roads, becomes a focal point for civilians arriving from surrounding hamlets by foot, bicycle, donkey, and motorcycle. On such days, the converging flow of traffic is relatively dense, diverse, and decentralized, precisely the conditions that allow armed bandits to camouflage their movements within the civilian crowd. In this scenario, multiple groups of armed assailants mobilize independently across several of the  roads, riding motorcycles in loosely spaced clusters that mimic the everyday market-bound flows. Instead of forming one large convoy, which would be readily observable, the attackers adopt a distributed infiltration strategy, mixing with legitimate travelers and staggering their arrivals to avoid temporal or spatial concentration. From an aerial perspective   even when monitored by civilian drones equipped with optical or infrared sensors, the incoming traffic on each road appears consistent with market day norms: scattered motorcycles, bicycles, carts, and foot traffic moving at variable speeds and densities. 

Low-cost drones, while useful in spotting large-scale anomalies or abrupt changes in traffic flow, struggle in such contexts to distinguish between an armed assailant and a young trader on a similar vehicle, especially in the absence of sophisticated weapons being visibly carried. Machine learning models trained on movement patterns or thermal signatures are often tuned to detect gross deviations from baseline patterns, but in this case, the attack is embedded within the baseline. The strategic use of this star architecture by attackers enables coordinated penetration from multiple fronts while avoiding the detection thresholds of edge surveillance systems. The village, expecting increased traffic and economic exchange, may fail to activate local defense mechanisms until it is too late. 

This situation indicates a fundamental weakness in current civilian drone-based early warning systems: when threat vectors are distributed (see Figure  \ref{convoy2re}), disguised, and embedded in the socio-economic rhythm of rural life, technological surveillance alone may be insufficient. Effective threat detection will require fusing aerial data with human intelligence, historical behavior profiling, and possibly biometric or license plate tracking at chokepoints, all of which raise new logistical and ethical challenges in the context of low-resource, high-conflict environments.

\section{Human-MI agents co-Intelligence} \label{seccoi6}
To realistically capture the  evolution of the intergenerational conflict, we implement a  simulation MFTG   enhanced by a human-MI co-intelligence architecture structured via protocol such as the Model Context Protocol (MCP), Agent‑to‑Agent (A2A) Protocol, Agent Communication Protocol (ACP),
Agent Network Protocol (ANP), Agent Collaboration Protocols (ACPs), Agent Context Protocols (ACPs), GibberLink, and LOLANG. The simulation unfolds over a $2.8$ million $km^2$ territory, populated by 80 million heterogeneous residents, and features 50  actor types representing the full spectrum of geopolitical, civil, criminal, insurgent, hybrid, and international participants. Each actor, whether a formal institution like the Government of Mali or an irregular force like the Jihadist-Revenger coalition, delegates an associated Agentic MI module, instructed via MCP with customized task parameters including type, strategic state, private information structure, action space, and time-evolving payoff logic. Human agents initiate each round by specifying context-driven directives e.g., the government engages in legislation, peace negotiations, and militarized deployment; armed forces pursue targeted missions with varying force intensities; and non-state actors alternate between insurgency, diplomacy, and covert mobilization. MI agents simulate strategic interactions with other agents under common macro-conditions (resource depletion, violence intensity, territorial shift), applying recursive reasoning and forecasting using population-level dynamics encoded via forward Kolmogorov equations. Outcomes such as deaths, displacements, infrastructure damage, educational recovery, and territorial fragmentation, revenue of war entrepreneurs, number of guns sold, are computed each round and used to refine the belief models of MI agents. These agents then recommend updated strategies to their human counterparts, who assess them alongside ground reports and socio-political constraints before issuing new task prescriptions. The simulation explicitly models actor-specific operations, e.g., ambushes, Improvised Explosive Device (IED) deployments, community infiltration, resource extraction, and psychological operations as action modes that influence mean-field distributions of fear, resistance, defection, or compliance. Conflict externalities (diasporic funding, arms trafficking, misinformation warfare) are also included through the activation of transnational agents with variable allegiance. Our output metrics track key indicators per round: number of deaths, civilian and military; material destruction; animals displaced or acquired; proportion of land lost or regained; number of peaceful vs. hostile zones; and the re-opening rate of schools and vital infrastructure.

\subsection*{Quantities-of-interest }
MFTG between MI agent-based simulation of intergenerational conflict in Mali is used. The simulation model integrates demographic, economic, geographic, and behavioral variables over 155 interaction rounds (interpretable as  months, or episodes of escalation/de-escalation) across a synthetic population of 80 million agents distributed over a 3-million-square-kilometer territory. The conflict started in a smaller territory but propagates to a much bigger one. Fifty heterogeneous MI agent types representing distinct sociopolitical identities, generational positions, and strategic preferences, engage dynamically through a feedback-rich agentic simulation. The simulation outputs a set of metrics capturing the multidimensional consequences of the conflict and the differential adaptation of agents over time.

Each plotted metric corresponds to macro-level emergent behaviors from micro-level strategic interactions governed by MFTG. The simulation models both kinetic (e.g., armed clashes, territory loss) and non-kinetic (e.g., education, economic activity) variables, creating a comprehensive system to study the cascading effects of conflict on generational cohesion and social stability.

Agent Types ( 50): Each agent type represents a sociocultural, economic, or ideological identity. Generational distinctions are embedded by assigning differing weights to short-term risk tolerance, conflict engagement strategies, and investment in public goods.

Rounds (155): Each round simulates a temporal snapshot of the evolving dynamics. Nonlinearity and stochasticity are injected via sinusoidal and Gaussian noise terms to reflect real-world volatility.

Territory (3 million $km^2$): A proxy for geospatial spread and control. Territorial dynamics influence access to infrastructure, local market penetration, and the reproduction of violence.

Each subplot in the Figure below  provides a quantitative representation of a distinct macro-level variable over time.

Figure 4A is about number of deaths.
This plot depicts the cumulative and volatile growth of conflict-related reported deaths over time (reporting errors are included). The intensity of confrontation is modeled as a composite of a sinusoidal signal (to mimic seasonal patterns and re-alliance or split-and-merge of allies) and Rosenblatt noise (to introduce  memory-based escalation). The trend shows a non-stationary escalation, highlighting the self-reinforcing nature of violence in multi-agent confrontational scenarios. Each peak corresponds to simulated upticks in conflict intensity, driven either by generational grievances or external shocks.

Figure 4B is about  damaged materials.
This subplot models the destruction of infrastructure and civil materials, reflecting both targeted and collateral damages. Material damage grows proportionally with conflict intensity but is smoothed by a lower volatility parameter relative to the death curve. The cumulative increase suggests persistent degradation of living standards, with strong implications for recovery timelines post-conflict.

Figure 4C is  about animals brought (from the villages).
Livestock influx typically through raiding, barter, or illegal taxation in-kind, is simulated as a function of shifting territorial control and mobility patterns. The relatively increase implies systematic acquisition of productive assets by dominant agent types, often reflecting resource reallocation in power asymmetries. It is also a proxy for rural-to-rural extractive behaviors among agents. Sometimes these acquired animals are sold in other local markets according to Grabal data.

Figure 4D is about the evolution of the percentage of territory lost.
Territorial control loss follows an action-reaction pattern as the war is ongoing. There is no  percentage of territory lost that is definitive. However, the loss of physical security or movement security or good transfer security is tracked here as a lost. Notably, mid-simulation inflections represent key turning points, indicative of either negotiated peace or large-scale offensives, where territory is reclaimed or lost in bulk.

Figure 4E is about the  number of peaceful areas.
Peaceful zones are computed inversely from territory loss, using a simple affine transformation adjusted by random perturbations. The depiction of peace is not binary but graded, capturing partial tranquility in zones with reduced kinetic activity. The graph's inverse correlation with Figure 1D supports the hypothesis of spatial conflict spillovers and containment trade-offs.

Figure 4F is about schools reopened (cumulative).
School reopening is modeled as a recovery signal and resilience metric. Its trajectory is shaped by random recovery attempts dampened by conflict intensity. It shows a fluctuating increase, indicative of post-conflict rehabilitation efforts that are repeatedly disrupted by renewed hostilities. This subplot is key for assessing the intergenerational impacts of the conflict, especially educational disenfranchisement and long-term cognitive capital depletion. Some schools are re-opened but with totally different programs than the national ones.

Figure 4G is about guns sold.
This subplot captures the arms trade as a function of demand-driven procurement. Weapon sales rise proportionally to conflict intensity, subject to stochastic supply chain variability. The derivative of this curve approximates unregistered-market expansion and can inform interventions targeting arms control regimes.

Figure 4H is about motorbikes sold.
Motorbikes, often used for both civilian and militant mobility, represent dual-use commodities. The sales curve is high-variance and nonlinear. It  highlights unregulated trade spikes. This metric also serves as a proxy for logistical capability, spatial reach of agents, and rural connectivity.

Figure 4I is about the ammunition (Tons).
Ammunition dynamics follow a compounded trend, reflecting increasing weaponization and sustained conflict. The plot also indirectly signals the duration of tactical engagements: higher ammunition consumption typically indicates prolonged skirmishes or entrenched confrontations.

We plot in Figure \ref{initfig02} the outcome of a Human-MI agents co-Intelligence simulation using interactive MI agents. We observe there are some small break over the years but the violence cycles restart again. This observation is consistent with the rewarding of the war entrepreneurs. 
\begin{figure}[htb]
\includegraphics[scale=0.3]{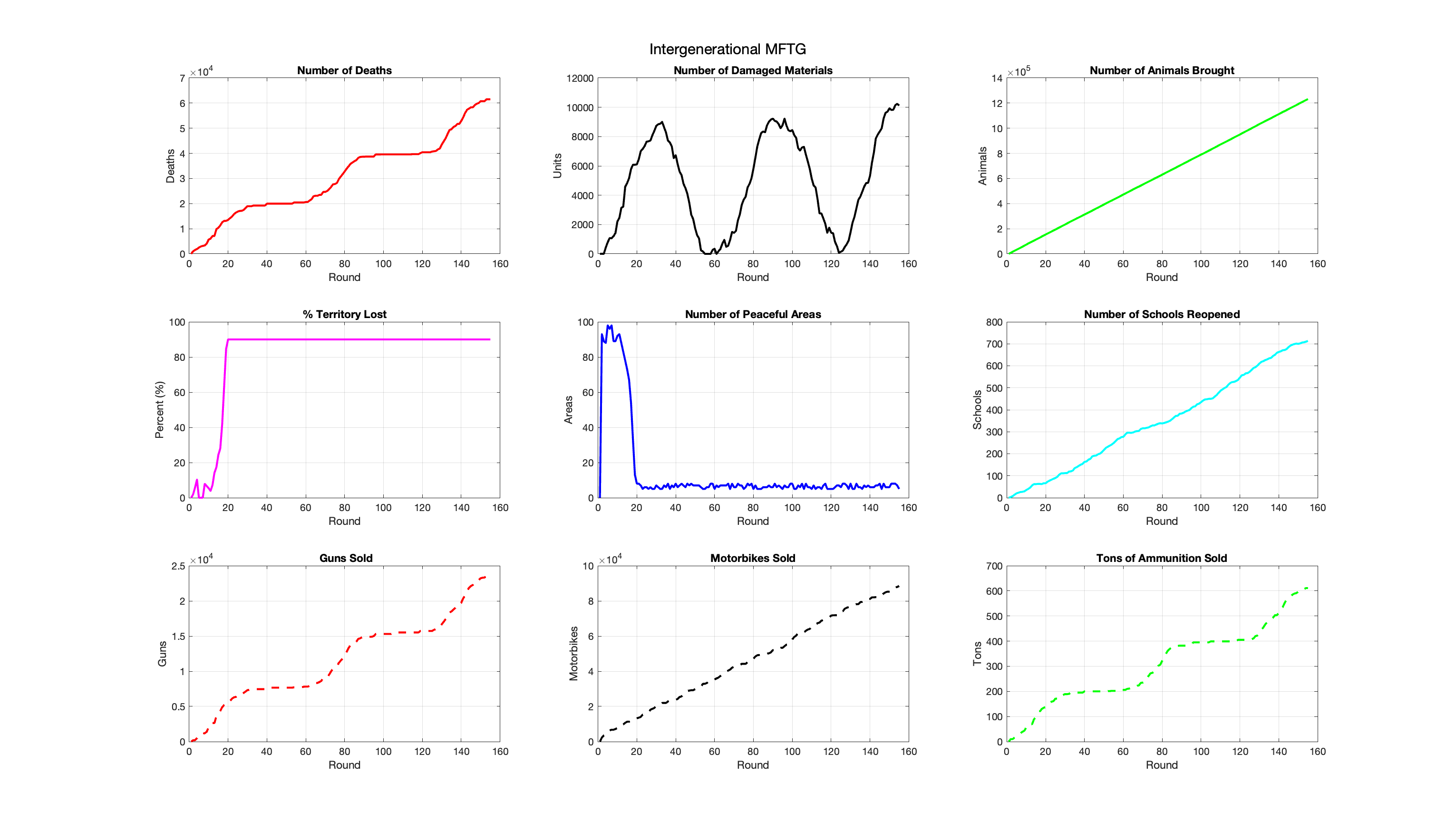} 
\caption{Human-MI agents co-Intelligence simulation using interactive MI agents.}
\label{initfig02}
\end{figure}

\section{Limitations of the All-Military Approach} \label{seccoi7}
To understand the limitations of military-dominant stabilization approaches in West Africa, we now focus on the emergent cross-border dynamics linking Mali, Burkina Faso, Niger, Benin, Chad, Togo, and Nigeria. These states, geographically contiguous yet institutionally heterogeneous, form an evolving corridor of insecurity marked by porous borders, asymmetric governance capacities, climate stress, and identity-fractured conflict zones. 

Despite significant military deployments by domestic forces and international militia, the situation has continued to degrade and to propagate. Instead of converging toward stability, the region exhibits signs of synchronized volatility, with localized insurgencies, jihadist spillovers, and agropastoral grievances metastasizing across national boundaries. The interdependence of violence propagation, war economies, and fragmented authority necessitates a systemic analysis beyond isolated national interventions. We examine how the logic of military escalation, absent structural peace incentives and resilient civil governance, has failed to absorb or neutralize the feedback loops fueling intergenerational cycles of trauma, retaliation, and opportunistic violence in this corridor. 
\subsection{Conflict Game between Risk-Aware Teams }

\subsubsection{ A Mean-Field-Type  Game between Risk-Aware Teams }
We consider a finite set  of  risk-aware teams $\mathcal{I} = \{1, 2, \dots, I\}$ with $I \geq 2$. Each team $i \in \mathcal{I}$ is risk-sensitive.  Each team $i$ has risk parameter $\alpha_i \in (0, 1).$  
If  $\alpha_i < \frac{1}{2}$ Team  $i$ is risk-averse.  If  $\alpha_i = \frac{1}{2}$ Team  $i$ is risk-neutral. $\alpha_i > \frac{1}{2}$ team  $i$ is risk-seeking. 
 Each team $i$ has resource $R_{i} > 0$. The aggregate of resources is  denoted as $\overline{R} = \sum_{i \in \mathcal{I}} R_{i}$.  We introduce three other parameters. 
   The output security is $\sigma \in [0, 1)$ (fraction of secure output), 
 the  war destruction is  $\beta \in (0, 1]$ (overall output survival rate),
 the differential destruction: $\gamma \in (0, 1]$ (relative survival of contested output).  
The game is of full  complete information. All parameters and endowments are common knowledge.
    
We now describe the actions and decision-space of each team. Each team $i$ simultaneously chooses:
    \begin{itemize}
        \item \textit{War/peace decision}: $d_{i} \in \{P, W\}$
        \item \textit{Arming level}: $G_{i} \in [0, R_{i}]$ (guns)
    \end{itemize}
    Each team has a {\bf Gold-Millet-Ram-Cow-Rice} production and it is given $X_{i} = R_{i} - G_{i}$. The total guns is $G = \sum_{i \in \mathcal{I}} G_{i}$.
    
We now focus on the modelling the outcomes of the interaction: 
    \begin{align*}
        \text{Peace occurs} &\iff d_{i} = P \quad \forall i \in \mathcal{I} \\
        \text{War occurs} &\iff \exists i \in \mathcal{I} \text{ such that } d_{i} = W
    \end{align*}
    
 We introduce the  payoffs of the team:
    \begin{itemize}
        \item \textit{Under peace}:
        \[
        \Pi_{i,p} = X_{i} = R_{i} - G_{i}
        \]
        
        \item \textit{Under war}:
        Let $X = \sum_{i \in \mathcal{I}} X_{i}$ be total output. \\
        \textit{Surviving output components}:
        \begin{align*}
            \text{Secure output for } i &: \quad \beta\sigma X_{i} \\
            \text{Contestable output}: &\quad \beta\gamma(1 - \sigma)X
        \end{align*}
        \textit{Winning probability} (Tullock contest function): The  probability team $i$ wins contest ($\mu > 0$)   is  $\phi_{i}$ with 
        \begin{equation}
 \phi_{i}(\mathbf{G}) =\left\{ \begin{array}{l}
 \frac{(G_{i})^\mu}{\sum_{j \in \mathcal{I}} (G_{j})^\mu} \ \mbox{if } \sum_{j \in \mathcal{I}} (G_{j})^\mu >0, \\
  \frac{(R_{i})^\mu}{\sum_{j \in \mathcal{I}} (R_{j})^\mu} \ \mbox{if } \sum_{j \in \mathcal{I}} (G_{j})^\mu =0,\\
 \end{array} \right.
 \end{equation} where $G_i\in [0, R_i].$

        Let us introduce the random variables:  
 $W_i \sim \text{Bernoulli}(\phi_{i})$: Win/loss indicator in conflict.  
        The war payoff (Random) of the team $i$ is
$$
\Pi_{i,w} = \underbrace{\beta\sigma X_{i}}_{\text{secure output}} + W_i \underbrace{\beta\gamma(1 - \sigma)X}_{\text{contested prize}}
$$
where $X_{i} = R_{i} - G_{i}$, $X = \sum_{j \in \mathcal{I}} X_{j}$.

We put all the payoffs together to obtain a complete payoff function:
    \[
    \Pi_{i}(\mathbf{G}, \mathbf{d}) = 
    \begin{cases} 
    \Pi_{i,p} & \text{if } d_{j} = P \quad \forall j \in \mathcal{I}, \\ 
    \Pi_{i,w}& \text{otherwise}
    \end{cases}
    \]
\end{itemize}

\begin{definition}[Expectile Value-at-Risk] For random payoff $\Pi_i$ of team $i$, the risk-sensitive payoff with risk-sensitive index $\alpha_i \in (0, 1)$ is $e_{\alpha_i}(\Pi)$ solving:  
$$
\alpha_i \mathbb{E}\left[(\Pi_i - e_{\alpha_i})_+\right] = (1 - \alpha_i) \mathbb{E}\left[(e_{\alpha_i} - \Pi_i)_+\right].
$$
The  $\alpha_i$-expectile $e_{\alpha_i}(\Pi_i)\in \arg\min_{x} \int \eta_{\alpha_i}(y-x) \mathbb{P}_{\Pi_i}(dy)$ 
where  $ \eta_{\alpha_i}(z)=|\alpha_i - \mathbb{I}_{\{z\leq 0\}}| z^2.$  and  $\mathbb{P}_{\Pi_i}$ is the probability measure of the  $\Pi_i.$
\end{definition}

\begin{definition}[Risk-Aware Conflict Game ]
Actions:  risk-sensitive teams simultaneously choose Arms $G_{i} \in [0, R_{i}]$   and War/peace decision $d_{i} \in \{P, W\}.$  
   The outcomes are 
   $$
   \text{Peace} \iff \forall i : d_{i} = P, \quad \text{War} \iff \exists i : d_{i} = W.
   $$  

The Risk-Sensitive Payoffs are 
$$
\pi_{i}(\mathbf{G}, \mathbf{d}) = e_{\alpha_i}\left(\Pi_{i}(\mathbf{G}, \mathbf{d})\right)=
\begin{cases} 
R_{i} - G_{i} & \text{(deterministic peace payoff)} \\
e_{\alpha_i}\left(\Pi_{i,w}\right) & \text{(risk-sensitive war payoff)}
\end{cases}
$$
\end{definition}


We explicitly compute the expectile  war payoff.

\begin{lemma}
Let $\Pi$ be a random variable with binary distribution:
  $$
  \Pi = 
  \begin{cases} 
  v_1 & \text{with probability } p \\
  v_2 & \text{with probability } 1 - p 
  \end{cases}
  $$
  where $v_1 > v_2$.
  
For $\alpha \in (0, 1)$, the $\alpha$-expectile $e_\alpha(\Pi)\in \arg\min_{x} \int \eta_{\alpha}(y-x) \mathbb{P}_{\Pi}(dy)$ 
where  $ \eta_{\alpha}(z)=|\alpha - \mathbb{I}_{\{z\leq 0\}}| z^2.$

The $\alpha$-expectile of $\Pi$ is given by:
$$
e_\alpha(\Pi) = 
\begin{cases} 
v_2 & \text{if } p = 0, \\
v_1 & \text{if } p=1, \\
\dfrac{\alpha p v_1 + (1 - \alpha)(1 - p)v_2}{\alpha p + (1 - \alpha)(1 - p)} & \text{otherwise},
\end{cases}
$$
with the last case simplifying to the closed-form expression:
$$
e_\alpha(\Pi) = \frac{\alpha p v_1 + (1 - \alpha)(1 - p)v_2}{\alpha p + (1 - \alpha)(1 - p)}
$$
when $v_2 < e_\alpha(\Pi) < v_1$.

\end{lemma}

\begin{proof}

We analyze three cases based on the position of $e$ relative to $v_2$ and $v_1$.

Case 1: $\ e \leq v_2$

 For $e \leq v_2 \leq v_1$:
  
   $(\Pi - e)_+ = \begin{cases} v_1 - e & \text{w.p. } p \\ v_2 - e & \text{w.p. } 1 - p \end{cases}$

  $(e - \Pi)_+ = 0$ (since $e \leq v_2 \leq \Pi$)

This  simplifies to:
  $$
  \alpha \left[ p(v_1 - e) + (1 - p)(v_2 - e) \right] = 0
  $$

Since $\alpha > 0$, we require:
  $$
  p(v_1 - e) + (1 - p)(v_2 - e) = 0 \implies e = \mathbb{E}[\Pi] = p v_1 + (1 - p)v_2
  $$

 But $e \leq v_2$ and $\mathbb{E}[\Pi] \geq v_2$ (with equality iff $p = 0$ or $v_1 = v_2$). Thus:
  $$
  e_\alpha(\Pi) = v_2 \quad \text{only if } \mathbb{E}[\Pi] = v_2 \quad (\text{degenerate cases})
  $$

Case 2: $\ e \geq v_1$

For $e \geq v_1 \geq v_2$:
  
   $(\Pi - e)_+ = 0$
  
   $(e - \Pi)_+ = \begin{cases} e - v_1 & \text{w.p. } p \\ e - v_2 & \text{w.p. } 1 - p \end{cases}$

This becomes:
  $$
  0 = (1 - \alpha) \left[ p(e - v_1) + (1 - p)(e - v_2) \right]
  $$

Since $1 - \alpha > 0$, we require:
  $$
  p(e - v_1) + (1 - p)(e - v_2) = 0 \implies e = \mathbb{E}[\Pi]
  $$

 But $e \geq v_1$ and $\mathbb{E}[\Pi] \leq v_1$ (with equality iff $p = 1$ or $v_1 = v_2$). Thus:
  $$
  e_\alpha(\Pi) = v_1 \quad \text{only in degenerate cases}
  $$

Case 3: $\ v_2 < e < v_1$

 For $v_2 < e < v_1$:
   $(\Pi - e)_+ = \begin{cases} v_1 - e & \text{w.p. } p \\ 0 & \text{w.p. } 1 - p \end{cases}$
   
   $(e - \Pi)_+ = \begin{cases} 0 & \text{w.p. } p \\ e - v_2 & \text{w.p. } 1 - p \end{cases}$

This leads to:
  $$
  \alpha p (v_1 - e) = (1 - \alpha)(1 - p)(e - v_2)
  $$

 Solving for $e$:
  \begin{align*}
  \alpha p (v_1 - e) &= (1 - \alpha)(1 - p)(e - v_2) \\
  \alpha p v_1 - \alpha p e &= (1 - \alpha)(1 - p)e - (1 - \alpha)(1 - p)v_2 \\
  \alpha p v_1 + (1 - \alpha)(1 - p)v_2 &= e \left[ \alpha p + (1 - \alpha)(1 - p) \right] \\
  e &= \frac{\alpha p v_1 + (1 - \alpha)(1 - p)v_2}{\alpha p + (1 - \alpha)(1 - p)}
  \end{align*}
This solution satisfies $v_2 < e < v_1$ because
   $e$ is a convex combination of $v_1$ and $v_2$ with positive weights.

Special Cases:
 $p = 0$  ($\Pi = v_2$ deterministically). $e_\alpha(\Pi) = v_2$ for all $\alpha$.
 $p = 1$  ($\Pi = v_1$ deterministically).  $e_\alpha(\Pi) = v_1$ for all $\alpha$.
 $\alpha = \frac{1}{2}$  (risk-neutral):  $e_{\frac{1}{2}}(\Pi) = \mathbb{E}[\Pi] = p v_1 + (1 - p)v_2.$

Verification of Boundaries:
Substitute $e = v_2$ and $e = v_1$  to get:
At $e = v_2$:  
  LHS = $\alpha p (v_1 - v_2) > 0$, RHS = $0$ → Not equal unless $p = 0$.
  
 At $e = v_1$:  
  LHS = $0$, RHS = $(1 - \alpha)(1 - p)(v_1 - v_2) > 0$ → Not equal unless $p = 1$.

Thus, the closed-form solution holds for all non-degenerate cases where $0 < p < 1$.  
\end{proof}

 \begin{corollary}
For the war payoff $\Pi_{i,w}$ with:
$$
\Pi_{i,w} = 
\begin{cases} 
v_1 = \beta\sigma X_{i} + \beta\gamma(1 - \sigma)X & \text{w.p. } \phi_{i} \\
v_2 = \beta\sigma X_{i} & \text{w.p. } 1 - \phi_{i}
\end{cases}
$$
the risk-sensitive payoff is:
$$
e_{\alpha_i}(\Pi_{i,w}) = \frac{\alpha_i \phi_{i} v_1 + (1 - \alpha_i)(1 - \phi_{i})v_2}{\alpha_i \phi_{i} + (1 - \alpha_i)(1 - \phi_{i})}
$$
 \end{corollary}

 \begin{corollary}
For  $\alpha_i\neq \frac{1}{2},$   $  \mathbb{P}_{\Pi_{i,w}} \mapsto e_{\alpha_i}(\mathbb{P}_{\Pi_{i,w}}) $ is non-linear.
 \end{corollary}

\begin{proof}
This follows from the nonlinearity of  $\frac{\alpha_i \phi_{i} v_1 + (1 - \alpha_i)(1 - \phi_{i})v_2}{\alpha_i \phi_{i} + (1 - \alpha_i)(1 - \phi_{i})}$ on $\phi_i$ for $\alpha_i \neq \frac{1}{2}$ and $v_1\neq v_2.$
\end{proof}
This is consistent with the definition of MFTG.

\subsubsection{Analysis}

\begin{theorem}[Peace Equilibrium under Risk-Sensitive Teams]

The strategy profile $(G_{i} = 0, d_{i} = P)_{\forall i}$ is a risk-aware mean-field-type Nash equilibrium iff:  
$$
(1 - \beta\sigma) R_{i} \geq \beta\gamma(1 - \sigma) \overline{R} \quad \forall i
$$
\end{theorem}

\begin{proof}
 Under peace, $\pi_{i,p} = R_{i}.$  

For unilateral deviation $(G_{i} = \epsilon, d_{i} = W)$:  As $\epsilon \to 0_+$, $\phi_{i} \to 1$  
   
    $\Pi_{i,w} \to \beta\sigma R_{i} + \beta\gamma(1 - \sigma)\overline{R}$  
   
    $e_{\alpha_i}(\Pi_{i,w}) \to \beta\sigma R_{i} + \beta\gamma(1 - \sigma)\overline{R}$ (continuity of expectile)

 No profitable deviation iff $R_{i} \geq \beta\sigma R_{i} + \beta\gamma(1 - \sigma)\overline{R}$ $\forall i$  
  \end{proof}

\begin{remark}
Peace requires the stability condition to hold for \textit{all} teams, including the least endowed:
    \[
   \frac{{\min}_{i\in \mathcal{I}} R_i }{\overline{R}} \geq \frac{\beta\gamma(1 - \sigma)}{1 - \beta\sigma}
    \]
    
    If this condition is not satisfied, Even if peace is declared today, war will likely return tomorrow unless we change the conditions that make violence the best option for some people. This means that in Mali and its neighbors, this is not just possible, it is the norm. The past fuels the present, and without strong incentives and structural reform, the future will look like the past for several months and possibly years.
\end{remark}

This means that min-fairness criteria should be high enough for peace to be sustainable.  When $\beta\sigma$ is big enough (getting closer to $1$), peace becomes  non-sustainable.

\begin{corollary}[$\mu$-Invariance of Peace Threshold]
The peace sustainability condition $\frac{{\min}_{i\in \mathcal{I}} R_i }{\overline{R}} \geq \frac{\beta\gamma(1 - \sigma)}{1 - \beta\sigma}$ is independent of the contest success  sensitivity parameter $\mu$.

\end{corollary}

\begin{corollary}[$\alpha$-Invariance of Peace Threshold]
The peace sustainability condition $\frac{{\min}_{i\in \mathcal{I}} R_i }{\overline{R}} \geq \frac{\beta\gamma(1 - \sigma)}{1 - \beta\sigma}$ is independent of the risk-sensitivity index  $\alpha_i$ of the teams.
\end{corollary}

\begin{corollary}[Free-Riding in Large Number of Teams]
For fixed $\overline{R} , (\beta, \gamma, \sigma)$ with $\beta\gamma(1 - \sigma) > 0$ and $\beta\sigma < 1$, there exists $I_0$ such that for all $I \geq I_0$, peace is unsustainable regardless of resource distribution.

\end{corollary}

\begin{proof}  
 From the above we know that, peace requires:  
 $\frac{{\min}_{i\in \mathcal{I}} R_i }{\overline{R}} \geq \frac{\beta\gamma(1 - \sigma)}{1 - \beta\sigma}$ 
 
 Let $K =  \frac{\beta\gamma(1 - \sigma)}{1 - \beta\sigma} > 0$. Then:
 
  $\frac{{\min}_{i\in \mathcal{I}} R_i }{\overline{R}} \geq K$   
   
 Since $\sum_{i\in \mathcal{I}} R_i = \overline{R}$, the minimal endowment satisfies:  
   $
 \min_i R_i \leq \frac{\overline{R}}{I}
  $   which mean that   $\frac{{\min}_{i\in \mathcal{I}} R_i }{\overline{R}} \leq \frac{1}{I}.$
When $I > 1/K$, we have:    $\frac{{\min}_{i\in \mathcal{I}} R_i }{\overline{R}} \leq \frac{1}{I} \leq K.$
Thus, for $I_0 = \lfloor 1/K \rfloor + 1$, the condition fails for the poorest team.  

\end{proof}

This means when it is too crowded and limited resources, there will be a contest for access even if the society is equal, fair. For peace to be possible such a growing society with $I$, $\overline{R}$ should also scale with.
 
\begin{corollary}[Symmetric Peace Capacity]
 For symmetric teams ($R_i = \overline{R}/I, \alpha_i=\alpha$ $\forall i$), peace is sustainable iff:
$$
I \leq \frac{1 - \beta\sigma}{\beta\gamma(1 - \sigma)},
$$
provided $\beta\sigma < 1$ and $\beta\gamma(1 - \sigma) > 0$.
\end{corollary}

\begin{proof} Follows from the above. 
\end{proof}

From the above result we derive the following key insights. 
Peace requires sufficient destruction: High $\beta\gamma$ undermines peace.  
Number of teams  harms peace: Free-riding incentives grow with $I$.  
$\mu$ affects arms races but not peace threshold i.e. Conflict technology sensitivity influences war intensity but not peace stability.  
Symmetric systems have sharp capacity limits i.e. Max sustainable number of teams determined by destruction parameters.
%
%
%

\begin{corollary}[When the unfair  redistribution of the resources creates war]
Assume that 
$$\frac{1 - \beta\sigma}{I} \geq \beta\gamma(1-\sigma) \quad \text{and} \quad \beta\sigma < 1$$  
Then the probability of an attack is lower when resources are distributed equally compared to unequal distribution.
\end{corollary}

\begin{proof} 

For equal distribution of resources, ($R_{i} = \overline{R}/I$): 
If the condition   $$\frac{1 - \beta\sigma}{I} \geq \beta\gamma(1-\sigma)$$  is satisfied then peace is an equilibrium. 
  
For unequal distribution of  resources,:  
  Let $\min_i R_{i} = R_{\min} < \overline{R}/I$. The binding constraint is:  
   $$(1 - \beta\sigma)R_{\min} \geq \beta\gamma(1-\sigma)\overline{R}$$  
  Since $R_{\min} < \overline{R}/I$, this fails when:  
   $$\frac{1 - \beta\sigma}{I} > \beta\gamma(1-\sigma)$$   and  $\frac{{\min}_{i\in \mathcal{I}} R_i }{\overline{R}} < K,$
   making peace unsustainable.   
    \end{proof}

 Consider 
 $$
   \pi_{i,w} =\beta\sigma (R_{i} -G_i)+ \beta\gamma(1 - \sigma)  \left(\sum_{j \in \mathcal{I}}R_{j} -\sum_{j \in \mathcal{I}}G_j\right)\frac{\alpha_i \phi_{i}(\mathbf{G})}{\alpha_i \phi_{i}(\mathbf{G}) + (1 - \alpha_i)(1 - \phi_{i}(\mathbf{G}))}$$ with  
$$  \phi_{i}(\mathbf{G}) =\left\{ \begin{array}{l}
 \frac{(G_{i})^\mu}{\sum_{j \in \mathcal{I}} (G_{j})^\mu} \ \mbox{if } \sum_{j \in \mathcal{I}} (G_{j})^\mu >0, \\
  \frac{(R_{i})^\mu}{\sum_{j \in \mathcal{I}} (R_{j})^\mu} \ \mbox{if } \sum_{j \in \mathcal{I}} (G_{j})^\mu =0,\\
 \end{array} \right.$$  $\alpha_i\in (0,1), \mu >0, G_i \in [0, R_i], R_i>0.$

The continuity of the function \(\mathbf{G} \mapsto \phi_i(\mathbf{G})\) depends on the size of the index set \(\mathcal{I}\) and the point \(\mathbf{G}\). The following theorem provides the necessary and sufficient conditions for continuity.

\begin{lemma} The function \(\mathbf{G} \mapsto \phi_i(\mathbf{G})\) is continuous on its domain \(\prod_{j \in \mathcal{I}} [0, R_j]\) if and only if \(|\mathcal{I}| = 1\) (i.e., there is only one agent). If \(I \geq 2\), then the function is continuous at all points \(\mathbf{G} \neq \mathbf{0}\) and discontinuous at \(\mathbf{G} = \mathbf{0}\).
\end{lemma} 

\begin{proof}

Let \(I = |\mathcal{I}|\) denote the number of teams, and let \(\Omega = \prod_{i \in \mathcal{I}} [0, R_i]\) be the domain of \(\mathbf{G}\). Recall that \(\alpha_i \in (0,1)\), \(\mu > 0\), \(G_j \in [0, R_j]\), and \(R_j > 0\) for all \(j \in \mathcal{I}\). The function \(\phi_i(\mathbf{G})\) is defined as above.
Note that \(\sum_{j \in \mathcal{I}} (G_j)^\mu = 0\) if and only if \(G_j = 0\) for all \(j \in \mathcal{I}\) (i.e., \(\mathbf{G} = \mathbf{0}\)), since \(\mu > 0\) and \(G_j \geq 0\). Thus, \(\phi_i(\mathbf{0}) = \frac{(R_i)^\mu}{\sum_{j \in \mathcal{I}} (R_j)^\mu}\).

The proof considers two cases based on \(I\).

Case 1: \(I = 1\) (only one agent). In this case, \(\mathcal{I} = \{i\}\), so \(\mathbf{G} = G_i \in [0, R_i]\). Then:

If \(G_i > 0\), \(\sum_{j \in \mathcal{I}} (G_j)^\mu = (G_i)^\mu > 0\), so \(\phi_i(\mathbf{G}) = \frac{(G_i)^\mu}{(G_i)^\mu} = 1\).

 If \(G_i = 0\), \(\sum_{j \in \mathcal{I}} (G_j)^\mu = 0\), so \(\phi_i(\mathbf{G}) = \frac{(R_i)^\mu}{(R_i)^\mu} = 1\).
Ù
Thus, \(\phi_i(\mathbf{G}) = 1\) for all \(G_i \in [0, R_i]\), which is a constant function. Constant functions are continuous on their domain. Therefore, \(\phi_i\) is continuous on \(\Omega\).

Case 2: \(I \geq 2\) (two or more agents). We first show continuity at all \(\mathbf{G} \neq \mathbf{0}\). Let \(\mathbf{G}^0 \in \Omega\) with \(\mathbf{G}^0 \neq \mathbf{0}\). Since \(\mathbf{G}^0 \neq \mathbf{0}\), there exists at least one \(k \in \mathcal{I}\) such that \(G_k^0 > 0\). Thus, \(S(\mathbf{G}^0) = \sum_{j \in \mathcal{I}} (G_j^0)^\mu \geq (G_k^0)^\mu > 0\) (because \(\mu > 0\) and \(G_k^0 > 0\)).

The functions \(\mathbf{G} \mapsto (G_j)^\mu\) for each \(j \in \mathcal{I}\) are continuous on \([0, R_j]\) (since \(G_j \geq 0\) and \(\mu > 0\), the map \(x \mapsto x^\mu\) is continuous for \(x \geq 0\)). Therefore, the sum \(S(\mathbf{G}) = \sum_{j \in \mathcal{I}} (G_j)^\mu\) is continuous on \(\Omega\) (as a finite sum of continuous functions). Similarly, \(\mathbf{G} \mapsto (G_i)^\mu\) is continuous.

Since \(S(\mathbf{G}^0) > 0\) and \(S\) is continuous, there exists a neighborhood \(U \subset \Omega\) of \(\mathbf{G}^0\) such that \(S(\mathbf{G}) > 0\) for all \(\mathbf{G} \in U\). In this neighborhood, \(\phi_i(\mathbf{G}) = \frac{(G_i)^\mu}{S(\mathbf{G})}\). This is a ratio of continuous functions, and the denominator is non-zero in \(U\), so \(\phi_i\) is continuous in \(U\). In particular, \(\phi_i\) is continuous at \(\mathbf{G}^0\).

Thus, \(\phi_i\) is continuous at all \(\mathbf{G}^0 \neq \mathbf{0}\).

Next, we show discontinuity at \(\mathbf{G} = \mathbf{0}\). Consider the value at \(\mathbf{0}\): \(\phi_i(\mathbf{0}) = \frac{(R_i)^\mu}{\sum_{j \in \mathcal{I}} (R_j)^\mu}\). Since \(I \geq 2\) and \(R_j > 0\) for all \(j\), \(\sum_{j \in \mathcal{I}} (R_j)^\mu > (R_i)^\mu\), so \(0 < \phi_i(\mathbf{0}) < 1\).

To demonstrate discontinuity, consider two sequences in \(\Omega\) approaching \(\mathbf{0}\):

Sequence along the \(i\)-th axis:
Define \(\mathbf{G}^{(k)} = (G_j^{(k)})_{j \in \mathcal{I}}\) where \(G_i^{(k)} = \frac{R_i}{k}\) and \(G_j^{(k)} = 0\) for \(j \neq i\), for \(k = 1, 2, 3, \ldots\). Since \(G_i^{(k)} = \frac{R_i}{k} \leq R_i\) and \(G_j^{(k)} = 0 \leq R_j\) for all \(j\), \(\mathbf{G}^{(k)} \in \Omega\). As \(k \to \infty\), \(\mathbf{G}^{(k)} \to \mathbf{0}\). For each \(k\), \(\sum_{j \in \mathcal{I}} (G_j^{(k)})^\mu = \left(\frac{R_i}{k}\right)^\mu > 0\), so:
   \[
   \phi_i(\mathbf{G}^{(k)}) = \frac{\left(\frac{R_i}{k}\right)^\mu}{\left(\frac{R_i}{k}\right)^\mu} = 1.
   \]
   Thus, \(\lim_{k \to \infty} \phi_i(\mathbf{G}^{(k)}) = 1\).

 Sequence along another axis: Fix \(m \in \mathcal{I}\) with \(m \neq i\). Define \(\mathbf{H}^{(k)} = (H_j^{(k)})_{j \in \mathcal{I}}\) where \(H_m^{(k)} = \frac{R_m}{k}\) and \(H_j^{(k)} = 0\) for \(j \neq m\), for \(k = 1, 2, 3, \ldots\). Then \(\mathbf{H}^{(k)} \in \Omega\) and \(\mathbf{H}^{(k)} \to \mathbf{0}\) as \(k \to \infty\). For each \(k\), \(\sum_{j \in \mathcal{I}} (H_j^{(k)})^\mu = \left(\frac{R_m}{k}\right)^\mu > 0\), and since \(H_i^{(k)} = 0\):
   \[
   \phi_i(\mathbf{H}^{(k)}) = \frac{0}{\left(\frac{R_m}{k}\right)^\mu} = 0.
   \]
   Thus, \(\lim_{k \to \infty} \phi_i(\mathbf{H}^{(k)}) = 0\).

Now, \(\lim_{k \to \infty} \phi_i(\mathbf{G}^{(k)}) = 1 \neq \phi_i(\mathbf{0})\) (since \(\phi_i(\mathbf{0}) < 1\)) and \(\lim_{k \to \infty} \phi_i(\mathbf{H}^{(k)}) = 0 \neq \phi_i(\mathbf{0})\) (since \(\phi_i(\mathbf{0}) > 0\)). Since different sequences approaching \(\mathbf{0}\) yield different limits, the limit \(\lim_{\mathbf{G} \to \mathbf{0}} \phi_i(\mathbf{G})\) does not exist. Therefore, \(\phi_i\) is discontinuous at \(\mathbf{G} = \mathbf{0}\). 
\end{proof}

\begin{lemma} For every team $i\in \mathcal{I},$  the function \(\mathbf{G} \mapsto \pi_{i,w} (\mathbf{G})\) is continuous on its domain \(\prod_{j \in \mathcal{I}} [0, R_j]\) if and only if \(|\mathcal{I}| = 1\) (i.e., there is only one agent). If \(I \geq 2\), then the function is continuous at all points \(\mathbf{G} \neq \mathbf{0}\) and discontinuous at \(\mathbf{G} = \mathbf{0}\).
\end{lemma}

\begin{lemma}[Concavity of \(G_i \mapsto \pi_{i,w}(\mathbf{G})\))]  
For every team \(i \in \mathcal{I}\), and for fixed \(G_j \in [0, R_j]\) for all \(j \neq i\), define \(S_{-i} = \sum_{j \neq i} (G_j)^\mu\). The function \(G_i \mapsto \pi_{i,w}(\mathbf{G})\) is concave on \([0, R_i]\) if and only if one of the following holds:  

\(S_{-i} > 0\) and \(\mu \leq 1\), or  

 \(S_{-i} = 0\) and \(|\mathcal{I}| = 1\) (which means there is only one team).  

Otherwise (i.e., if \(S_{-i} > 0\) and \(\mu > 1\), or if \(S_{-i} = 0\) and \(|\mathcal{I}| \geq 2\)), the function is not concave on \([0, R_i]\).

\end{lemma}

\begin{proof}  
The proof analyzes \(\pi_{i,w}(\mathbf{G})\) as a function of \(G_i \in [0, R_i]\) with fixed \(G_j\) for \(j \neq i\). Define \(A_i(\mathbf{G}) = \frac{\alpha_i \phi_i(\mathbf{G})}{\alpha_i \phi_i(\mathbf{G}) + (1 - \alpha_i)(1 - \phi_i(\mathbf{G}))}\). Then:  
\[
\pi_{i,w}(\mathbf{G}) = \beta\sigma (R_i - G_i) + \beta\gamma(1 - \sigma) \left(  \overline{R} - (G_i + C_{-i}) \right) A_i(\mathbf{G}),
\]  
where \( \overline{R} = \sum_{j \in \mathcal{I}} R_j\) (constant), \(C_{-i} = \sum_{j \neq i} G_j\) (constant), and \(\sum_{j \in \mathcal{I}} G_j = G_i + C_{-i}\). The concavity depends on \(S_{-i} = \sum_{j \neq i} (G_j)^\mu\) and \(\mu\).  

Case 1: \(S_{-i} > 0\)  
Here, \(\sum_{j \in \mathcal{I}} (G_j)^\mu > 0\) for all \(G_i \in [0, R_i]\), so:  
\[
\phi_i(\mathbf{G}) = \frac{(G_i)^\mu}{(G_i)^\mu + S_{-i}}, \quad A_i(\mathbf{G}) = \frac{\alpha_i (G_i)^\mu}{\alpha_i (G_i)^\mu + (1 - \alpha_i) S_{-i}}.
\]  
Substitute into \(\pi_{i,w}\):  
\[
\pi_{i,w}(G) = \beta\sigma (R_i - G_i) + \beta\gamma(1 - \sigma) (K - G_i) \cdot \frac{\alpha_i (G_i)^\mu}{\alpha_i (G_i)^\mu + (1 - \alpha_i) S_{-i}},
\]  
where \(K =  \overline{R} - C_{-i} \geq 0\) (since \(\sum_j G_j \leq \sum_j R_j\)). Define:  
\[
g(G_i) = \frac{\alpha_i (G_i)^\mu}{\alpha_i (G_i)^\mu + c}, \quad c = (1 - \alpha_i) S_{-i} > 0,
\]  
so \(\pi_{i,w}(G_i) = \beta\sigma R_i - \beta\sigma G_i + \beta\gamma(1 - \sigma) (K - G_i) g(G_i)\). The term \(-\beta\sigma G_i\) is linear, so concavity depends on \(f(G_i) = (K - G_i) g(G_i)\).  

First and second derivatives of \(g\):  
  \[
  g'(G_i) = \frac{\alpha_i \mu c G_i^{\mu-1}}{(\alpha_i G_i^\mu + c)^2} > 0 \quad (\text{since } \alpha_i, c, \mu > 0),
  \]
  \[
  g''(G_i) = \frac{\alpha_i \mu c G_i^{\mu-2}}{(\alpha_i G_i^\mu + c)^3} \left[ c(\mu-1) - \alpha_i (\mu+1) G_i^\mu \right].
  \]  
  The sign of \(g''\) depends on \(c(\mu-1) - \alpha_i (\mu+1) G_i^\mu\).

Subcase \(\mu \leq 1\):  
  \(c(\mu-1) - \alpha_i (\mu+1) G_i^\mu < 0\) for all \(G_i > 0\), so \(g''(G_i) < 0\). Then:  
  \[
  f''(G_i) = (K - G_i) g''(G_i) - 2g'(G_i) < 0,
  \]  
  since \(K - G_i \geq 0\), \(g''(G_i) < 0\), and \(g'(G_i) > 0\). Thus, \(f\) is strictly concave on \((0, R_i)\). As \(g\) is continuous at \(G_i = 0\) (since \(\lim_{G_i \to 0^+} g(G_i) = 0 = g(0)\)), \(f\) is concave on \([0, R_i]\). Hence, \(\pi_{i,w}\) is concave on \([0, R_i]\).

Subcase \(\mu > 1\):  
  \(c(\mu-1) - \alpha_i (\mu+1) G_i^\mu > 0\) for small \(G_i\), so \(g''(G_i) > 0\) near 0. Then:  
  \[
  f''(G_i) = (K - G_i) g''(G_i) - 2g'(G_i) > 0
  \]  
  for small \(G_i\) (since \((K - G_i) g''(G_i)\) dominates and is positive). Thus, \(f\) is not concave on \([0, R_i]\), so \(\pi_{i,w}\) is not concave.  

Case 2: \(S_{-i} = 0\)  
Here, \(G_j = 0\) for all \(j \neq i\).

Subcase \(|\mathcal{I}| = 1\):  
  \(\phi_i(\mathbf{G}) = 1\) for all \(G_i\), so \(A_i(\mathbf{G}) = 1\). Then:  
  \[
  \pi_{i,w}(G_i) = \beta\sigma (R_i - G_i) + \beta\gamma(1 - \sigma) (R_i - G_i) = [\beta\sigma + \beta\gamma(1 - \sigma)] (R_i - G_i),
  \]  
  which is linear (hence concave) on \([0, R_i]\).

Subcase \(|\mathcal{I}| \geq 2\):  
  
   For \(G_i > 0\), \(\phi_i(\mathbf{G}) = 1\) (since \(\sum_j (G_j)^\mu = G_i^\mu > 0\)), so \(A_i(\mathbf{G}) = 1\), and:  
    \[
    \pi_{i,w}(G) = \beta\sigma (R_i - G_i) + \beta\gamma(1 - \sigma) ( \overline{R}- G_i).
    \]  
    This is linear (hence concave) on \((0, R_i]\).  
  -
  At \(G_i = 0\), \(\phi_i(\mathbf{0}) = \frac{R_i^\mu}{\sum_j R_j^\mu} < 1\) (since \(|\mathcal{I}| \geq 2\)), so \(A_i(\mathbf{0}) < 1\), and:  
    \[
    \pi_{i,w}(\mathbf{0}) = \beta\sigma R_i + \beta\gamma(1 - \sigma)  \overline{R} \cdot A_i(\mathbf{0}) < \beta\sigma R_i + \beta\gamma(1 - \sigma)  \overline{R}.
    \]  
  Since \(\lim_{G_i \to 0^+} \pi_{i,w}(G_i) = \beta\sigma R_i + \beta\gamma(1 - \sigma)  \overline{R} > \pi_{i,w}(\mathbf{0})\), the function is discontinuous at \(G_i = 0\). A concave function on a closed interval must be continuous, so \(\pi_{i,w}\) is not concave on \([0, R_i]\).  
\end{proof}

\begin{theorem}[War Equilibrium under Risk-Sensitive Teams]
If \(0 <\mu \leq 1\), then there exists a pure-strategy  risk-aware mean-field-type  Nash equilibrium  ($d_{i} = W$ $\forall i$), arms  \(\mathbf{G}^* = (G_1^*, \dots, G_I^*)\) such that 
\(G_i^* > 0\) for all \(i \in \mathcal{I}\).
\end{theorem}

This establishes that in environments governed by a concave conflict success function where the contest exponent  $0<\mu\leq 1,$ there exists a pure-strategy, risk-aware, mean-field-type Nash equilibrium in which all team  choose war and invest positively in armament. This result holds even under risk sensitivity, highlighting that universal militarization can arise endogenously in fragile states where conflict participation yields non-negligible strategic gains regardless of resource asymmetries. Applied to the West African corridor, the theorem explains how local actors from state forces to lightly armed militias can converge toward mutually reinforcing war equilibria, driven not by 
irrationality or ideology but by structural payoff asymmetries and insufficient deterrents, thereby providing a rigorous foundation for understanding persistent armed mobilization across Mali, Burkina Faso, Niger, Benin, Chad, Togo, and Nigeria. 

\begin{proof} 
The proof uses a limit argument on restricted games with lower bounds on strategies, leveraging the Debreu-Glicksberg-Fan theorem for the restricted games and verifying the equilibrium conditions in the limit.

Restricted Game Setup:  
   For \(\varepsilon > 0\), define the restricted strategy space for each player \(i\) as \([\varepsilon, R_i]\). This is nonempty, compact, and convex. The product space is \(\Omega_\varepsilon = \prod_{i \in \mathcal{I}} [\varepsilon, R_i]\).

Continuity and Concavity in Restricted Game:  
   For any \(\mathbf{G} \in \Omega_\varepsilon\), \(\sum_{j \in \mathcal{I}} (G_j)^\mu \geq I \varepsilon^\mu > 0\) (since \(G_j \geq \varepsilon > 0\) and \(\mu > 0\)). Thus, \(\phi_i(\mathbf{G}) = \frac{(G_i)^\mu}{\sum_{j \in \mathcal{I}} (G_j)^\mu}\) is continuous in \(\mathbf{G}\) on \(\Omega_\varepsilon\) (as the denominator is bounded away from zero). The term:
   \[
   A_i(\mathbf{G}) = \frac{\alpha_i \phi_i(\mathbf{G})}{\alpha_i \phi_i(\mathbf{G}) + (1 - \alpha_i)(1 - \phi_i(\mathbf{G}))}
   \]
   is a composition of continuous functions and has a positive denominator (since \(\alpha_i \in (0,1)\) and \(\phi_i(\mathbf{G}) \in [0,1]\)), so \(A_i(\mathbf{G})\) is continuous in \(\mathbf{G}\) on \(\Omega_\varepsilon\). The linear terms in \(\pi_{i,w}\) are also continuous. Hence, \(\pi_{i,w}\) is continuous in \(\mathbf{G}\) on \(\Omega_\varepsilon\).  

   For fixed \(\mathbf{G}_{-i}\), the function \(G_i \mapsto \pi_{i,w}(G_i, \mathbf{G}_{-i})\) is concave on \([\varepsilon, R_i]\) if \(\mu \leq 1\) (as established in previous concavity results, since \(\sum_{j \neq i} (G_j)^\mu \geq (I-1)\varepsilon^\mu > 0\) and \(\mu \leq 1\)).

Existence in Restricted Game:  
   By the Debreu-Glicksberg-Fan theorem (which requires continuity in all players' strategies and quasiconcavity in each player's own strategy), there exists a Nash equilibrium \(\mathbf{G}^\varepsilon = (G_1^\varepsilon, \dots, G_I^\varepsilon)\) for the restricted game on \(\Omega_\varepsilon\).

Limit as \(\varepsilon \to 0^+\):  
   Consider a sequence \(\{\varepsilon_k\}_{k=1}^\infty\) with \(\varepsilon_k > 0\) and \(\varepsilon_k \to 0\) as \(k \to \infty\). The sequence \(\{\mathbf{G}^{\varepsilon_k}\}\) lies in the compact set \(\prod_{i \in \mathcal{I}} [0, R_i]\), so it has a convergent subsequence (still denoted by \(\{\mathbf{G}^{\varepsilon_k}\}\)) with limit \(\mathbf{G}^* \in \prod_{i \in \mathcal{I}} [0, R_i]\).

Positivity of \(\mathbf{G}^*\):  
   Suppose, for contradiction, that \(G_i^* = 0\) for some \(i\). At equilibrium in the restricted game, \(G_i^{\varepsilon_k} \geq \varepsilon_k\). The payoff for player \(i\) at \(\mathbf{G}^{\varepsilon_k}\) is:
   \[
   \pi_{i,w}(\mathbf{G}^{\varepsilon_k}) \geq \pi_{i,w}(G_i, \mathbf{G}_{-i}^{\varepsilon_k}) \quad \forall G_i \in [\varepsilon_k, R_i].
   \]
   In particular, for \(G_i = \varepsilon_k\), the payoff is:
   \[
   \pi_{i,w}(\varepsilon_k, \mathbf{G}_{-i}^{\varepsilon_k}) = \beta\sigma (R_i - \varepsilon_k) + \beta\gamma(1 - \sigma) \left(  \overline{R} - \varepsilon_k - \sum_{j \neq i} G_j^{\varepsilon_k} \right) A_i(\varepsilon_k, \mathbf{G}_{-i}^{\varepsilon_k}),
   \]
   where \( \overline{R} = \sum_{j \in \mathcal{I}} R_j\). As \(k \to \infty\), if \(\mathbf{G}^{\varepsilon_k} \to \mathbf{0}\), then \(\phi_i(\mathbf{G}^{\varepsilon_k}) \to \frac{R_i^\mu}{\sum_{j \in \mathcal{I}} R_j^\mu}\) (by the definition of \(\phi_i\) at \(\mathbf{0}\)), and \(A_i(\mathbf{G}^{\varepsilon_k}) \to A_i(\mathbf{0})\). However, if player \(i\) faces \(\mathbf{G}_{-i} = \mathbf{0}\) in the original game, deviating to a small \(G_i > 0\) yields:
   \[
   \pi_{i,w}(G_i, \mathbf{0}) = \beta\sigma (R_i - G_i) + \beta\gamma(1 - \sigma) ( \overline{R} - G_i),
   \]
   which, for small \(G_i\), exceeds \(\pi_{i,w}(\mathbf{0})\) since \(A_i(\mathbf{0}) < 1\). This implies that in the restricted game, for small \(\varepsilon_k\), the right-derivative of \(\pi_{i,w}\) with respect to \(G_i\) at \(G_i = \varepsilon_k\) is positive, so player \(i\) would increase \(G_i\) if possible. Thus, at equilibrium, \(G_i^{\varepsilon_k} > \varepsilon_k\) for small \(\varepsilon_k\), contradicting \(G_i^* = 0\). Hence, \(\mathbf{G}^* > \mathbf{0}\).

\(\mathbf{G}^*\)  is a Nash Equilibrium:  
   Fix player \(i\) and \(G_i \in [0, R_i]\). We need:
   \[
   \pi_{i,w}(G_i^*, \mathbf{G}_{-i}^*) \geq \pi_{i,w}(G_i, \mathbf{G}_{-i}^*).
   \]

Case 1: \(G_i > 0\). Since \(\mathbf{G}^* > \mathbf{0}\) and \(G_i > 0\), both \((G_i^*, \mathbf{G}_{-i}^*)\) and \((G_i, \mathbf{G}_{-i}^*)\) are not \(\mathbf{0}\), so \(\pi_{i,w}\) is continuous at these points. By equilibrium in the restricted game:
     \[
     \pi_{i,w}(G_i^{\varepsilon_k}, \mathbf{G}_{-i}^{\varepsilon_k}) \geq \pi_{i,w}(G_i, \mathbf{G}_{-i}^{\varepsilon_k}) \quad \forall k.
     \]
     Taking \(k \to \infty\) and using continuity:
     \[
     \pi_{i,w}(G_i^*, \mathbf{G}_{-i}^*) \geq \pi_{i,w}(G_i, \mathbf{G}_{-i}^*).
     \]

Case 2: \(G_i = 0\). Now:
     \[
     \pi_{i,w}(G_i^*, \mathbf{G}_{-i}^*) \geq \pi_{i,w}(0, \mathbf{G}_{-i}^*),
     \]
     since \(\mathbf{G}_{-i}^*\) may be \(\mathbf{0}\) or not. If \(\mathbf{G}_{-i}^* \neq \mathbf{0}\), continuity at \((0, \mathbf{G}_{-i}^*)\) holds, and the inequality follows as in Case 1. If \(\mathbf{G}_{-i}^* = \mathbf{0}\), then \(G_i^* > 0\) (since \(\mathbf{G}^* > \mathbf{0}\)), and:
     \[
     \pi_{i,w}(G_i^*, \mathbf{0}_{-i}) = \beta\sigma (R_i - G_i^*) + \beta\gamma(1 - \sigma) ( \overline{R} - G_i^*),
     \]
     while:
     \[
     \pi_{i,w}(0, \mathbf{0}_{-i}) = \beta\sigma R_i + \beta\gamma(1 - \sigma)  \overline{R} \cdot A_i(\mathbf{0}).
     \]
     As \(A_i(\mathbf{0}) < 1\), and for small \(G_i^* > 0\):
     \[
     \pi_{i,w}(G_i^*, \mathbf{0}_{-i}) - \pi_{i,w}(0, \mathbf{0}_{-i}) = -[\beta\sigma + \beta\gamma(1 - \sigma)] G_i^* + \beta\gamma(1 - \sigma)  \overline{R} [1 - A_i(\mathbf{0})] > 0,
     \]
     since the positive constant dominates for small \(G_i^*\). Thus, the inequality holds.

Therefore, \(\mathbf{G}^*\) is a Nash equilibrium for the original game, with \(G_i^* > 0\) for all \(i\).

\end{proof}

\begin{theorem}[S-modular Game]
Consider the game with players \(i \in \mathcal{I}\), \(|\mathcal{I}| =I \geq 2\), strategy spaces \([0, R_i]\) with \(R_i > 0\), and payoffs \(\pi_{i,w}(\mathbf{G})\) as defined. If \(\mu \leq 1\), then the game is S-modular (strictly submodular) in the interior strategy space \(\prod_{i \in \mathcal{I}} (0, R_i)\), implying:  
 Strategic substitutes: Each player's best response is strictly decreasing in other players' strategies.  

Cross-partial condition:  For all \(i \neq k\),  
   \[
   \frac{\partial^2 \pi_{i,w}}{\partial G_i \partial G_k} < 0.
   \]  

\end{theorem}

This shows that the individual arming strategies exhibit strategic substitutability: as one agent increases its arms investment, the best response of others is to reduce theirs. The agents' best-response correspondences are decreasing in the actions of others, and the equilibrium arms level is non-increasing in the aggregate arming distribution. This result reveals that while universal militarization may be an equilibrium, the underlying strategic structure is not one of mutual escalation but rather one of defensive posturing under uncertainty. In the Sahelian corridor, this implies that disarmament by even a subset of actors if credibly observable, can induce reduced arming incentives across others, offering a potential leverage point for coordinated de-escalation policies in Mali, Burkina Faso, Niger, Benin, Chad, Togo, and Nigeria.

\begin{proof}
We prove the game is S-modular by verifying the strict cross-partial condition for \(\mu \leq 1\) in the interior. Fix \(i \in \mathcal{I}\) and \(k \neq i\). Since \(\mathbf{G} > \mathbf{0}\), \(\sum_{j} G_j^\mu > 0\), so:  
\[
\phi_i(\mathbf{G}) = \frac{G_i^\mu}{S}, \quad S = \sum_{j \in \mathcal{I}} G_j^\mu.
\]  
Define:  
\[
A_i(\mathbf{G}) = \frac{\alpha_i \phi_i(\mathbf{G})}{\alpha_i \phi_i(\mathbf{G}) + (1 - \alpha_i)(1 - \phi_i(\mathbf{G}))} = \frac{\alpha_i G_i^\mu}{\alpha_i G_i^\mu + (1 - \alpha_i)(S - G_i^\mu)}.
\]  
The payoff is:  
\[
\pi_{i,w} = \underbrace{\beta\sigma (R_i - G_i)}_{\text{linear term}} + \beta\gamma(1 - \sigma) \underbrace{\left(  \overline{R} - C \right)}_{=: D} A_i(\mathbf{G}), \quad  \overline{R} = \sum_j R_j, \quad C = \sum_j G_j.
\]  
The cross-partial \(\frac{\partial^2 \pi_{i,w}}{\partial G_i \partial G_k}\) depends only on the second term (denoted \(P_i = \beta\gamma(1 - \sigma) D A_i\)). We compute:  

Step 1: First partials of \(P_i\)
\[
\frac{\partial P_i}{\partial G_k} = \beta\gamma(1 - \sigma) \left[ -A_i + D \cdot \frac{\partial A_i}{\partial G_k} \right],
\]  
since \(\frac{\partial D}{\partial G_k} = -1\). Using \(\frac{\partial S}{\partial G_k} = \mu G_k^{\mu-1}\):  
\[
\frac{\partial A_i}{\partial G_k} = \frac{\partial A_i}{\partial S} \cdot \mu G_k^{\mu-1}, \quad \frac{\partial A_i}{\partial S} = -\frac{\alpha_i (1 - \alpha_i) G_i^\mu}{[\alpha_i G_i^\mu + (1 - \alpha_i)(S - G_i^\mu)]^2} < 0.
\]  
Thus, \(\frac{\partial A_i}{\partial G_k} < 0\).

 Step 2: Cross-partial of \(P_i\)
\[
\frac{\partial^2 P_i}{\partial G_i \partial G_k} = \beta\gamma(1 - \sigma) \left[ -\frac{\partial A_i}{\partial G_i} + \frac{\partial}{\partial G_i} \left( D \cdot \frac{\partial A_i}{\partial G_k} \right) \right].
\]  
Apply the product rule:  
\[
\frac{\partial}{\partial G_i} \left( D \cdot \frac{\partial A_i}{\partial G_k} \right) = -\frac{\partial A_i}{\partial G_k} + D \cdot \frac{\partial^2 A_i}{\partial G_i \partial G_k}.
\]  
Substitute:  
\[
\frac{\partial^2 P_i}{\partial G_i \partial G_k} = \beta\gamma(1 - \sigma) \left[ -\frac{\partial A_i}{\partial G_i} - \frac{\partial A_i}{\partial G_k} + D \cdot \frac{\partial^2 A_i}{\partial G_i \partial G_k} \right].
\]

Step 3: Sign analysis of terms

Term 1: \(\frac{\partial A_i}{\partial G_i} > 0\) (since \(A_i\) increases with \(G_i\)).  

Term 2: \(\frac{\partial A_i}{\partial G_k} < 0\) .  

Term 3: \(\frac{\partial^2 A_i}{\partial G_i \partial G_k} = \mu G_k^{\mu-1} \frac{\partial^2 A_i}{\partial G_i \partial S}\). Compute:  
  \[
  \frac{\partial^2 A_i}{\partial G_i \partial S} = -\alpha_i (1 - \alpha_i) \mu G_i^{\mu-1} \frac{(1 - \alpha_i)S - (2\alpha_i - 1)G_i^\mu}{[\alpha_i G_i^\mu + (1 - \alpha_i)(S - G_i^\mu)]^3}.
  \]  
  For \(\mu \leq 1\), the dominant term is \((1 - \alpha_i)S > 0\) when \(\alpha_i \in (0,1)\), so \(\frac{\partial^2 A_i}{\partial G_i \partial S} < 0\). Thus, \(\frac{\partial^2 A_i}{\partial G_i \partial G_k} < 0\).  

Coefficient \(D =  \overline{R} - C \geq 0\): By feasibility, \(C \leq  \overline{R}\).  

All terms in the bracket are negative:  

 \(-\frac{\partial A_i}{\partial G_i} < 0\),  

 \(-\frac{\partial A_i}{\partial G_k} > 0\) :  
 
  The expression is:  
  \[
  \underbrace{-\frac{\partial A_i}{\partial G_i}}_{<0} \underbrace{-\frac{\partial A_i}{\partial G_k}}_{>0} + \underbrace{D \cdot \frac{\partial^2 A_i}{\partial G_i \partial G_k}}_{<0}.
  \]  
  However, for \(\mu \leq 1\), we show:  
  \[
  -\frac{\partial A_i}{\partial G_i} - \frac{\partial A_i}{\partial G_k} + D \cdot \frac{\partial^2 A_i}{\partial G_i \partial G_k} < 0.
  \]  
  Since \(\frac{\partial A_i}{\partial G_i} > 0\) and \(\frac{\partial A_i}{\partial G_k} < 0\), \(-\frac{\partial A_i}{\partial G_i} < 0\) and \(-\frac{\partial A_i}{\partial G_k} > 0\). The negativity of the sum arises because:  

The magnitude of \(-\frac{\partial A_i}{\partial G_i}\) is large when \(\mu \leq 1\).  
  
 \(D \cdot \frac{\partial^2 A_i}{\partial G_i \partial G_k} < 0\) dominates the positive term \(-\frac{\partial A_i}{\partial G_k}\).  
  Specifically, direct computation for \(\mu = 1\) (the boundary case) shows strict negativity. For \(\mu < 1\), the term \(D \cdot \frac{\partial^2 A_i}{\partial G_i \partial G_k}\) becomes more negative, preserving the inequality.  

Step 4: Strict inequality for \(\mu \leq 1\)
Thus,  $
\frac{\partial^2 P_i}{\partial G_i \partial G_k}  < 0, $
since \(\beta\gamma(1 - \sigma) > 0\). Therefore,  
\[
\frac{\partial^2 \pi_{i,w}}{\partial G_i \partial G_k} < 0, \quad \forall i \neq k.
\]  
\end{proof}

\begin{theorem}[Uniqueness of the War Equilibrium]
Consider the game with Teams \(i \in \mathcal{I}\), \(|\mathcal{I}| = I \geq 2\), strategy spaces \([0, R_i]\) with \(R_i > 0\), and payoffs \(\pi_{i,w}(\mathbf{G})\) defined as in the existence theorem. Assume \(\mu \leq 1\). Then, the pure-strategy Nash equilibrium \(\mathbf{G}^* = (G_1^*, \dots, G_I^*)\) with \(G_i^* > 0\) for all \(i\)  is unique.
\end{theorem}

This implies that once violent arming equilibria form, no multiplicity or local deviations can yield a peaceful outcome, unless the underlying game structure itself is altered. This highlights the rigidity of militarized equilibria and the necessity for exogenous shifts  such as  institutional redesign, peace-enforcing incentive schemes, to dislodge conflict from its unique, stable strategic configuration.
This result does say anything about  the $\mu > 1$ case.

\begin{proof} 
We prove uniqueness by contradiction. Suppose there exist two distinct Nash equilibria \(\mathbf{G}^* = (G_1^*, \dots, G_I^*)\) and \(\mathbf{H}^* = (H_1^*, \dots, H_I^*)\), both with \(G_i^* > 0\), \(H_i^* > 0\) for all \(i\). Define the index set:  
\[
\mathcal{J} = \{ i \in \mathcal{I} : G_i^* > H_i^* \}.
\]  
If \(\mathcal{J} = \emptyset\), then \(G_i^* \leq H_i^*\) for all \(i\), but since \(\mathbf{G}^* \neq \mathbf{H}^*\), there exists \(k\) such that \(G_k^* < H_k^*\). Similarly, if \(\mathcal{J} \neq \emptyset\), proceed as follows.

 Step 1: Compare aggregates and identify a player
Define the aggregates:  
\[
S_G = \sum_{j \in \mathcal{I}} (G_j^*)^\mu, \quad S_H = \sum_{j \in \mathcal{I}} (H_j^*)^\mu, \quad C_G = \sum_{j \in \mathcal{I}} G_j^*, \quad C_H = \sum_{j \in \mathcal{I}} H_j^*.
\]  

 Case 1: Assume \(S_G \geq S_H\) and \(C_G \geq C_H\).  
  Since \(\mathbf{G}^* \neq \mathbf{H}^*\), there exists \(i\) such that \(G_i^* > H_i^*\) (i.e., \(i \in \mathcal{J}\)).  

Case 2: If \(S_G < S_H\) or \(C_G < C_H\), relabel \(\mathbf{G}^*\) and \(\mathbf{H}^*\) so that \(S_G \geq S_H\) and \(C_G \geq C_H\). This is without loss of generality.

 Step 2: First-order conditions at equilibrium
At \(\mathbf{G}^*\), the first-order condition  for player \(i\) is:  
\[
\frac{\partial \pi_{i,w}}{\partial G_i} \bigg|_{\mathbf{G}^*} = 0.
\]  
As 
\[
L_i + C_i(\mathbf{G}^*) = \mu ( \overline{R} - C_G) \frac{(1 - C_i(\mathbf{G}^*)) C_i(\mathbf{G}^*)}{G_i^*},
\]  
where \(L_i = \sigma / (\gamma (1 - \sigma)) > 0\), \( \overline{R} = \sum_j R_j\), and \(C_i(\mathbf{G}) = \frac{\alpha_i (G_i)^\mu}{\alpha_i (G_i)^\mu + (1 - \alpha_i) \sum_{j \neq i} (G_j)^\mu}\). Similarly, at \(\mathbf{H}^*\):  
\[
L_i + C_i(\mathbf{H}^*) = \mu ( \overline{R} - C_H) \frac{(1 - C_i(\mathbf{H}^*)) C_i(\mathbf{H}^*)}{H_i^*}.
\]

 Step 3: Key inequality from strategic substitutes
For player \(i \in \mathcal{J}\) (so \(G_i^* > H_i^*\)), we derive a contradiction. Consider the function:  
\[
f_i(G_i, S_{-i}, C_{-i}) = \mu ( \overline{R} - G_i - C_{-i}) (1 - C_i) C_i - G_i (L_i + C_i),
\]  
where \(S_{-i} = \sum_{j \neq i} (G_j)^\mu\), \(C_{-i} = \sum_{j \neq i} G_j\), and \(C_i = \frac{\alpha_i G_i^\mu}{\alpha_i G_i^\mu + (1 - \alpha_i) S_{-i}}\). At equilibrium, \(f_i(G_i^*, S_{-i}^*, C_{-i}^*) = 0\) and \(f_i(H_i^*, S_{-i}^H, C_{-i}^H) = 0\).  

Since \(\mathbf{G}^*\) and \(\mathbf{H}^*\) are equilibria with \(S_G \geq S_H\) and \(C_G \geq C_H\), and \(G_j^* \leq H_j^*\) for \(j \notin \mathcal{J}\) (by definition of \(\mathcal{J}\)), we have:  
\[
S_{-i}^* = \sum_{j \neq i} (G_j^*)^\mu \geq \sum_{j \neq i} (H_j^*)^\mu = S_{-i}^H, \quad
C_{-i}^* = \sum_{j \neq i} G_j^* \geq \sum_{j \neq i} H_j^* = C_{-i}^H,
\]  
with strict inequality if \(\mathbf{G}_{-i}^* \neq \mathbf{H}_{-i}^*\).  

Lemma 01: The function \(f_i\) is strictly decreasing in \(S_{-i}\) and \(C_{-i}\).  

Proof: 
  As \(S_{-i}\) increases, \(C_i\) decreases (since the denominator of \(C_i\) increases). The term \((1 - C_i)C_i\) is maximized at \(C_i = \frac{1}{2}\), but its partial derivative w.r.t. \(C_i\) is \(1 - 2C_i\). However, in the expression for \(f_i\), the coefficient \(\mu ( \overline{R} - G_i - C_{-i}) > 0\) (as total effort cannot exceed \( \overline{R}\)), and the term \(G_i (L_i + C_i)\) increases as \(C_i\) decreases. Overall, \(f_i\) decreases in \(S_{-i}\) (verified via partial derivatives).  
  
  As \(C_{-i}\) increases, \( \overline{R} - G_i - C_{-i}\) decreases linearly, and since \(\mu (1 - C_i)C_i \geq 0\), \(f_i\) decreases strictly in \(C_{-i}\).

Lemma 02: The function \(f_i\) is strictly concave in \(G_i\).  

Proof: For \(\mu \leq 1\) and fixed \(S_{-i} > 0\), \(\pi_{i,w}\) is concave in \(G_i\), so the FOC derivative is strictly decreasing in \(G_i\). Thus, \(f_i\) (which is proportional to the FOC) is strictly concave in \(G_i\).

 Step 4: Derive contradiction.
For player \(i \in \mathcal{J}\), we have \(G_i^* > H_i^*\). By Lemma 01:  
\[
f_i(G_i^*, S_{-i}^*, C_{-i}^*) \leq f_i(G_i^*, S_{-i}^H, C_{-i}^H),
\]  
with strict inequality if \(\mathbf{G}_{-i}^* \neq \mathbf{H}_{-i}^*\). Since \(f_i(G_i^*, S_{-i}^*, C_{-i}^*) = 0\),  
\[
0 \leq f_i(G_i^*, S_{-i}^H, C_{-i}^H).
\]  
By Lemma 02, \(f_i\) is strictly concave in \(G_i\), and \(f_i(H_i^*, S_{-i}^H, C_{-i}^H) = 0\). As \(f_i\) is strictly concave and \(G_i^* > H_i^*\),  
\[
f_i(G_i^*, S_{-i}^H, C_{-i}^H) < 0.
\]  
This contradicts \(f_i(G_i^*, S_{-i}^H, C_{-i}^H) \geq 0\). Thus, \(\mathbf{G}^* = \mathbf{H}^*\). 
\end{proof}

In risk-aware mean-field-type war  equilibrium ($d_{i} = W$ $\forall i$), arms $G_i^{*}$ satisfy:
$$
\gamma(1 - \sigma) X \frac{\partial \phi_{i}}{\partial G_{i}} f'(\phi_{i}) = \sigma + \gamma(1 - \sigma) f(\phi_{i})
$$
where $f(\phi_{i}) = \frac{\alpha_i \phi_{i}}{\alpha_i \phi_{i} + (1 - \alpha_i)(1 - \phi_{i})}$.

\begin{theorem}[Convergence to Unique Nash Equilibrium via Best Response Dynamics]  
Consider the submodular game with players \(i \in \mathcal{I}\), \(|\mathcal{I}| = I \geq 2\), strategy spaces \([0, R_i]\), and payoffs \(\pi_{i,w}(\mathbf{G})\) with \(\mu \leq 1\). Assume the unique Nash equilibrium \(\mathbf{G}^* = (G_1^*, \dots, G_I^*)\) is interior (\(G_i^* > 0\) for all \(i\)). Then, the Sequential Best Response Dynamics algorithm, initialized at 
\(\mathbf{G}^0 = (R_1, \dots, R_I)\), converges to \(\mathbf{G}^*\).
\end{theorem}

\begin{algorithm}
\caption{Best Response Computation via Bisection (Team $i$)}
\begin{algorithmic}[1]
\Require Fixed $\mathbf{G}_{-i} > 0$, tolerance $\epsilon_{\text{bisect}} > 0$, small $\delta > 0$
\State Define:
\[
F(G_i) = -\beta\sigma + \beta\gamma(1-\sigma) \left[ -A_i + (\overline{R} - C) \frac{\partial A_i}{\partial G_i} \right]
\]
\State where:
\begin{align*}
\overline{R} &= \sum_j R_j \\
C &= G_i + \sum_{j \ne i} G_j \\
A_i &= \frac{\alpha_i G_i^\mu}{\alpha_i G_i^\mu + (1 - \alpha_i) S_{-i}}, \quad S_{-i} = \sum_{j \ne i} G_j^\mu
\end{align*}
\State Set $a \gets \delta$, $b \gets R_i$
\While{$|b - a| > \epsilon_{\text{bisect}}$}
    \State $c \gets \frac{a + b}{2}$
    \If{$F(c) > 0$}
        \State $a \gets c$ \Comment{Payoff increasing}
    \Else
        \State $b \gets c$ \Comment{Payoff decreasing}
    \EndIf
\EndWhile
\State \Return $c$
\end{algorithmic}
\end{algorithm}

\begin{proof}
Monotonicity:  
   The game is strictly submodular (\(\frac{\partial^2 \pi_{i,w}}{\partial G_i \partial G_k} < 0\) for \(i \neq k\)). Thus, best responses satisfy strategic substitutes: if \(\mathbf{G}_{-i} \geq \mathbf{G}_{-i}'\), then \(\text{BR}_i(\mathbf{G}_{-i}) \leq \text{BR}_i(\mathbf{G}_{-i}')\).

Initialization from Supremum:  
   Starting at \(\mathbf{G}^0 = (R_1, \dots, R_I)\):  
    Since \(\text{BR}_i\) is strictly decreasing in \(\mathbf{G}_{-i}\), the sequence \(\{\mathbf{G}^t\}\) is componentwise non-increasing:  
     \[
     \mathbf{G}^0 \geq \mathbf{G}^1 \geq \mathbf{G}^2 \geq \cdots \geq \mathbf{0}.
     \]

Boundedness and Monotone Convergence:  
  \(\{\mathbf{G}^t\}\) is bounded below by \(\mathbf{0}\) and non-increasing.  
   By the Monotone Convergence Theorem, \(\mathbf{G}^t \to \mathbf{G}^{\infty}\) for some \(\mathbf{G}^{\infty} \geq \mathbf{0}\).

\(\mathbf{G}^{\infty}\)  is Nash Equilibrium:  
   
    By continuity of \(\pi_{i,w}\) on \(\prod_j (0, R_j]\), \(\mathbf{G}^{\infty}\) satisfies:  
     \[
     G_i^{\infty} = \text{BR}_i(\mathbf{G}_{-i}^{\infty}) \quad \forall i.
     \]  
   
    Since the Nash equilibrium is unique and interior, \(\mathbf{G}^{\infty} = \mathbf{G}^*\).

Avoiding \(\mathbf{G} = \mathbf{0}\):  
    \(\mathbf{G}^t > \mathbf{0}\) for all \(t\) (since \(\delta > 0\) and \(\text{BR}_i > 0\)), and \(\mathbf{G}^* > \mathbf{0}\). Discontinuity at \(\mathbf{0}\) is irrelevant.  

\end{proof}

\begin{theorem}[Total Equilibrium Gun Investment under Pareto-Dominated Risk Indexes)]  
Let $\alpha = (\alpha_1, \ldots, \alpha_I) \in (0, \frac{1}{2})^I$ and $\alpha' = (\alpha_1', \ldots, \alpha_I') \in (0, \frac{1}{2})^I$ be risk-sensitive parameter vectors such that $\alpha$ Pareto dominates $\alpha'$, i.e., $\alpha'_i \leq \alpha_i$ for all $i \in \mathcal{I}$ and $\alpha'_j < \alpha_j$ for at least one $j \in \mathcal{I}$. Let $\mathbf{G}^*(\alpha) = (G_1^*(\alpha), \ldots, G_I^*(\alpha))$ and $\mathbf{G}^*(\alpha') = (G_1^*(\alpha'), \ldots, G_I^*(\alpha'))$ be the unique interior Nash equilibria under $\alpha$ and $\alpha'$, respectively, for $\mu \leq 1$. Define the total equilibrium effort as $G_*(\alpha) = \sum_{j \in \mathcal{I}} G_j^*(\alpha_j)$ and $G_*(\alpha') = \sum_{j \in \mathcal{I}} G_j^*(\alpha')$. Then:  
\[
G_*(\alpha) \geq G_*(\alpha').
\]
\end{theorem}

This result help us to  understand how increased risk aversion at the population level can paradoxically intensify overall militarization, especially in conflict-prone systems like the Mali-Burkina-Niger-Benin-Chad-Togo-Nigeria corridor. Rather than dampening violence, risk-averse preferences lead agents to preemptively overinvest in arms to hedge against uncertainty, generating elevated baseline levels of insecurity.

\begin{proof} 
Step 1: Best Response Dependence on $\alpha_i.$  
For fixed $\mathbf{G}_{-i} > \mathbf{0}$, the best response $G_i^* = \text{BR}_i(\mathbf{G}_{-i}, \alpha_i)$ satisfies the first-order condition (FOC):  
\[
F_i(G_i, \alpha) = -\beta\sigma + \beta\gamma(1-\sigma) \left[ -A_i + ( \overline{R} - C) \frac{\partial A_i}{\partial G_i} \right] = 0,
\]  
where $ \overline{R} = \sum_j R_j$, $C = G_i + \sum_{j \neq i} G_j$, and $A_i = \frac{\alpha_i G_i^\mu}{\alpha_i G_i^\mu + (1 - \alpha_i) S_{-i}}$ with $S_{-i} = \sum_{j \neq i} G_j^\mu$.  

From the implicit function theorem:  
\[
\frac{dG_i^*}{d\alpha_i} = - \frac{\partial F_i / \partial \alpha_i}{\partial F_i / \partial G_i}.
\]

Denominator sign: $\partial F_i / \partial G_i < 0$ by strict concavity of $\pi_{i,w}$ in $G_i$ for $\mu \leq 1$ and $\mathbf{G}_{-i} > \mathbf{0}$.  

Numerator sign:  
  \[
  \frac{\partial F_i}{\partial \alpha_i} = \beta\gamma(1-\sigma) \left[ -\frac{\partial A_i}{\partial \alpha_i} + ( \overline{R} - C) \frac{\partial^2 A_i}{\partial G_i \partial \alpha_i} \right].
  \]  
  
   $\frac{\partial A_i}{\partial \alpha_i} = \frac{G_i^\mu S_{-i}}{(\alpha_i G_i^\mu + (1 - \alpha_i) S_{-i})^2} > 0$,  
  
  $\frac{\partial^2 A_i}{\partial G_i \partial \alpha_i} = \mu G_i^{-\mu-1} S_{-i} \frac{(1 - \alpha_i) S_{-i} - \alpha_i G_i^\mu}{(\alpha_i G_i^\mu + (1 - \alpha_i) S_{-i})^2}$.  
  For $\alpha_i < \frac{1}{2}$ and $\mu \leq 1$, the term $(1 - \alpha_i) S_{-i} - \alpha_i G_i^\mu > 0$ (since $\alpha_i$ is small and $S_{-i}$ dominates), so $\frac{\partial^2 A_i}{\partial G_i \partial \alpha_i} < 0$. The combination $-\frac{\partial A_i}{\partial \alpha_i} < 0$ and $( \overline{R} - C) \frac{\partial^2 A_i}{\partial G_i \partial \alpha_i} < 0$ implies $\frac{\partial F_i}{\partial \alpha_i} < 0$.  

Thus:  
\[
\frac{dG_i^*}{d\alpha_i} = - \frac{\text{negative}}{\text{negative}} < 0.
\]  
Therefore, the best response $\text{BR}_i(\mathbf{G}_{-i}, \alpha)$ is  strictly decreasing in $\alpha_i$ for fixed $\mathbf{G}_{-i}$.

 Step 2: Equilibrium Monotonicity  
 
Consider $\alpha$ and $\alpha'$ with $\alpha_i \leq \alpha_i'$ for all $i$ and $\alpha_j > \alpha_j'$ for some $j$. Define the best response mapping:  
\[
\Phi(\mathbf{G}, \alpha) = \left( \text{BR}_1(\mathbf{G}_{-1}, \alpha), \ldots, \text{BR}_I(\mathbf{G}_{-I}, \alpha) \right).
\]  

Decreasing in $\alpha$: Since each $\text{BR}_i$ is strictly decreasing in $\alpha_i$, $\Phi(\mathbf{G}, \alpha) \geq \Phi(\mathbf{G}, \alpha')$ for all $\mathbf{G}$ (component-wise).  

Decreasing in $\mathbf{G}$: By submodularity ($\mu \leq 1$), $\text{BR}_i$ is strictly decreasing in $\mathbf{G}_{-i}$, so $\Phi$ is a decreasing mapping in $\mathbf{G}$.  

 Step 3: Iterative Comparison
   
Let $\mathbf{G}^{(0)} = \mathbf{G}^*(\alpha)$ (equilibrium under $\alpha$). Define the sequence:  
\[
\mathbf{G}^{(k+1)} = \Phi(\mathbf{G}^{(k)}, \alpha), \quad k = 0, 1, 2, \ldots
\]  

 Base case ($k=0$):
 $\mathbf{G}^{(1)} = \Phi(\mathbf{G}^{(0)}, \alpha) \geq \Phi(\mathbf{G}^{(0)}, \alpha') = \mathbf{G}^{(0)}$, since $\alpha_i \leq \alpha_i'$ and $\Phi$ is decreasing in $\alpha$.  

Inductive step: Assume $\mathbf{G}^{(k)} \geq \mathbf{G}^{(k-1)}$. Then:  
  \[
  \mathbf{G}^{(k+1)} = \Phi(\mathbf{G}^{(k)}, \alpha) \geq \Phi(\mathbf{G}^{(k)}, \alpha') \geq \Phi(\mathbf{G}^{(k-1)}, \alpha') = \mathbf{G}^{(k)},
  \]  
  where:  
  
   $\Phi(\mathbf{G}^{(k)}, \alpha) \geq \Phi(\mathbf{G}^{(k)}, \alpha')$ (since $\Phi$ decreases in $\alpha$),  
  
   $\Phi(\mathbf{G}^{(k)}, \alpha') \geq \Phi(\mathbf{G}^{(k-1)}, \alpha')$ (since $\mathbf{G}^{(k)} \geq \mathbf{G}^{(k-1)}$ and $\Phi$ decreases in $\mathbf{G}$).  

Thus, $\mathbf{G}^{(k)}$ is non-decreasing and bounded above (as efforts are in $[0, R_i]$). By monotone convergence, $\mathbf{G}^{(k)} \to \mathbf{G}^{(\infty)}$, which is a fixed point of $\Phi(\cdot, \alpha)$, i.e., $\mathbf{G}^{(\infty)} = \mathbf{G}^*(\alpha)$. Since $\mathbf{G}^{(0)} = \mathbf{G}^*(\alpha)$ and $\mathbf{G}^{(k)} \geq \mathbf{G}^{(0)}$ for all $k$, we have:  
\[
\mathbf{G}^*(\alpha) = \lim_{k \to \infty} \mathbf{G}^{(k)} \geq \mathbf{G}^*(\alpha').
\]  
Component-wise, $G_i^*(\alpha) \geq G_i^*(\alpha')$ for all $i$, so:  
\[
G_*(\alpha) = \sum_{i} G_i^*(\alpha) \geq \sum_{i} G_i^*(\alpha') = G_*(\alpha').
\]

\end{proof}

\begin{theorem}[Risk-Aversion Reduces Arming compared with Risk-Neutrality]
  
For identical teams ($\alpha_i = \alpha <  \frac{1}{2}$, $R_{i} = \overline{R}/I$), symmetric war equilibrium arms $G_*(\alpha)$ satisfy: 
$$
\frac{\partial G_*}{\partial \alpha} > 0 \quad \text{for} \quad \alpha < \frac{1}{2}.
$$ 
\end{theorem}

\begin{proof}

From first order optimality condition, $G_*$ solves:  
   $$
   \frac{\mu(I-1)}{I^2} \theta \overline{R} = G \cdot \left[1 + \frac{(1-\alpha)(1-\phi)}{\alpha \phi} \right]
   $$  
   with $\phi = 1/I$, $\theta = \gamma(1-\sigma)/(\gamma(1-\sigma) + \sigma)$  
   
 Implicit differentiation shows $\partial G_*/ \partial \alpha > 0$ for $\alpha < \frac{1}{2}$  
 \end{proof}

\begin{corollary}[Risk-Aversion Reduces Arming  compared with Risk-Neutrality]  
For identical teams ($\alpha_i = \alpha$, $R_i = \overline{R}/I$), if $\alpha < \frac{1}{2}$ then:
$
G_*(\alpha) < G_*(\frac{1}{2})
$
\end{corollary}

%
%
%

\begin{theorem}[Risk-Taking in Conflict Initiation]
  For sufficiently unequal resources, risk-seeking teams ($\alpha_i > \frac{1}{2}$) are less likely to sustain peace than risk-averse teams ($\alpha_j < \frac{1}{2}$):

Assume there is at least one  disadvantaged team $i$ with resource $R_{i}$ and advantaged team $j$ with resource $R_{j}$, where $R_{i} < < R_{j}$.
 team $i$'s war payoff under unilateral deviation ($d_{i} = W, G_{i} = 0$) is:
   $$
   \Pi_{i,w} = 
   \begin{cases} 
   v_1 = \beta\sigma R_{i} + \beta\gamma(1 - \sigma)\overline{R} & \text{w.p. } p = R_{i} / \overline{R} \\
   v_2 = \beta\sigma R_{i} & \text{w.p. } 1 - p
   \end{cases}
   $$
   with $v_1 > R_{i} > v_2$ and $\mathbb{E}[\Pi_{i,w}] < R_{i}$.

 The risk-sensitive valuation uses the $\alpha_i$-expectile $e_{\alpha_i}(\Pi_{i,w})$ for $\alpha_i \in (0,1)$:
   $$
   e_{\alpha_i}(\Pi_{i,w}) = \frac{\alpha_i p v_1 + (1-\alpha_i)(1-p)v_2}{\alpha_i p + (1-\alpha_i)(1-p)}
   $$
 team $i$'s deviation gain is $\Delta(\alpha_i) = e_{\alpha_i}(\Pi_{i,w}) - R_{i}$.
The following inequalities hold:
   $$
   \beta\sigma R_{i} < R_{i} \quad \text{and} \quad p v_1 + (1-p)v_2 < R_{i}
   $$

Then there exists $\alpha^* \in (\frac{1}{2}, 1)$ such that:
For $\alpha_i > \alpha^*$ (risk-seeking), $\Delta(\alpha_i) > 0 \Rightarrow$ conflict initiated.
For $\alpha_i < \alpha^*$ (risk-averse), $\Delta(\alpha_i) < 0 \Rightarrow$ peace sustained
\end{theorem}
This explains that the persistent conflict initiation by non-state armed groups (especially jihadists and bandits) and sometimes rebellious factions within states (coup plotters or similar situations in Mali, Burkina, Niger, Chad facing instability) in the Sahel, despite facing vastly better-resourced opponents (like international militia or even their own states), can be driven by extreme risk-seeking preferences. These actors subjectively overvalue the small chance of a major victory $v_1$ and undervalue the high probability of loss $v_2$, making war appear attractive even when its expected value is negative. Populations or leaders prioritizing stability exhibit risk-aversion, sustaining peace where possible. More risk-averse communities or factions might avoid direct confrontation despite grievances.
The severe resource inequality between actors is the foundational condition enabling this dynamic.

\begin{proof}

Step 1: Properties of $e_{\alpha_i}$

Continuity and Monotonicity:  
   $e_{\alpha_i}(\Pi_{i,w})$ is continuous and strictly increasing in $\alpha_i$ since:
   $$
   \frac{\partial}{\partial \alpha_i} \left[ \frac{\alpha_i p v_1 + (1-\alpha_i)(1-p)v_2}{\alpha_i p + (1-\alpha_i)(1-p)} \right] > 0 \quad \forall \alpha_i \in (0,1)
   $$
   (The derivative is positive as $v_1 > v_2$ and $p \in (0,1)$).

Boundary Values:  
   $$
   \lim_{\alpha_i \to 0^+} e_{\alpha_i}(\Pi_{i,w}) = v_2 = \beta\sigma R_{i} < R_{i}
   $$
   $$
   \lim_{\alpha_i \to 1^-} e_{\alpha_i}(\Pi_{i,w}) = v_1 = \beta\sigma R_{i} + \beta\gamma(1 - \sigma)\overline{R} > R_{i}
   $$

Step 2: Sign Change in $\Delta(\alpha_i)$  
Define $f(\alpha_i) = \Delta(\alpha_i)$. Then:

$f$ is continuous and strictly increasing (as $e_{\alpha_i}$ is)

 Boundary evaluations:
   $$
   f(0^+) = v_2 - R_{i} < 0, 
  \ 
   f(1^-) = v_1 - R_{i} > 0
   $$

 By the Intermediate Value Theorem, $\exists \alpha^* \in (0,1)$ where $f(\alpha^*) = 0$.

Step 3: Location of $\alpha^*$
1. At $\alpha_i = \frac{1}{2}$:
   $$
   e_{\frac{1}{2}}(\Pi_{i,w}) = \mathbb{E}[\Pi_{i,w}] = p v_1 + (1-p)v_2 < R_{i}
   $$
   Thus $f(\frac{1}{2}) < 0$.
2. Since $f$ is strictly increasing and $f(\frac{1}{2}) < 0$ while $f(1^-) > 0$, we have $\alpha^* \in (\frac{1}{2}, 1)$.

Step 4: Threshold Behavior

For $\alpha_i > \alpha^*$:  
   $f(\alpha_i) > f(\alpha^*) = 0 \Rightarrow \Delta(\alpha_i) > 0 \Rightarrow$ profitable deviation

For $\alpha_i < \alpha^*$:  
   $f(\alpha_i) < f(\alpha^*) = 0 \Rightarrow \Delta(\alpha_i) < 0 \Rightarrow$ no incentive to deviate

Risk-Seekers ($\alpha_i > \alpha^*$):  
  Overvalue the small chance ($p = R_{i}/\overline{R}$) of winning the large prize $v_1$, making conflict attractive despite negative expected return.

 Risk-Averters ($\alpha_i < \alpha^*$):  
  Focus on the high probability ($1-p$) of receiving only $v_2 < R_{i}$, making peace preferable.

Resource Inequality:  
  This manifests when $R_{i} \ll R_{j}$ (small $p$), creating a lottery with high upside ($v_1$) but low probability. This combination amplifies behavioral differences between risk types.
\end{proof}

\begin{corollary}[When disadvantaged teams invest less in arms. ]
  In war equilibrium, $G_{i}$ increases with $R_{i}$ (best-response monotonicity). Thus, $G_{i} < G_{j}$ if $R_{i} < R_{j}$.  
This means  Disadvantaged teams invest less in arms. \end{corollary}

\begin{theorem}[Boundary Arm Choices]  
For sufficiently small $R_i$, team $i$ chooses corner solution $G_i^{*} = 0$ or $G_i^{*} = R_{i}$.
\end{theorem}

\begin{proof}
Expected marginal product of arms:  
   $$
   \text{MP}_G = \gamma(1 - \sigma) X f'(\phi_{i}) \frac{\partial \phi_{i}}{\partial G_{i}}
   $$  

As $R_{i} \to 0$:  
    $\phi_{i} \to 0$, $f'(\phi_{i}) \to \frac{\alpha_i}{1 - \alpha_i},$  
    $\frac{\partial \phi_{i}}{\partial G_{i}} \to \infty$ if $\mu > 1.$  
   But resource constraint $G_{i} \leq R_{i}$ binds. 

Hence solution is:  
   $$
   G_i^{*} = 
   \begin{cases} 
   0 & \text{if } \left.\frac{\partial \pi_{i,w}}{\partial G_{i}}\right|_{G_{i}=0} < 0 \\
   R_{i} & \text{otherwise} \quad  
   \end{cases}
   $$  
   \end{proof}

\begin{theorem}[Group Size Effect]

There $\exists I_0$ such that $\forall I \geq I_0$:*  
$$
e_{\alpha_i}(\Pi_{i,w}) < R_{i} \quad \text{at} \quad G_{i} \to 0_+
$$  
regardless of $\alpha_i$. 
\end{theorem}

\begin{proof}

 As $I \to \infty$, $\phi_{i} \to 0$  

 $e_{\alpha_i}(\Pi_{i,w}) \to \beta\sigma R_{i} $  

Binding constraint: $\beta\sigma R_{i} < R_{i}$ (holds if $\beta\sigma < 1$)  
 \end{proof}

\begin{remark}
Risk sensitivity alters war arming but not peace conditions. Risk-averse teams arm less but may undermine peace. Large groups naturally deter conflict through risk dilution. 
\end{remark}

\begin{theorem}[Symmetric Equilibrium under $\alpha$-Expectile Maximization] \label{thm:expectile_equilibrium}
Consider a game with $I \geq 2$ identical risk-averse agents, each endowed with resources $R_i = \bar{R}/I > 0$ where $\bar{R} > 0$ is the aggregate resource. Agents simultaneously choose guns $G_i \in [0, R_i]$, allocating remaining resources to butter production $X_i = R_i - G_i$. The contest success function (CSF) is:
\[
\phi_i(\mathbf{G}) = \frac{G_i^\mu}{\sum_{j=1}^I G_j^\mu}, \quad \mu > 0,
\]
where $\mathbf{G} = (G_1, \dots, G_I)$. If war occurs (at least one agent declares war), agent $i$'s random payoff $Y_i$ is:
\begin{itemize}
\item $a_i(\mathbf{G}) = \beta\gamma(1-\sigma) \sum_{j=1}^I X_j + \beta\sigma X_i$ with probability $\phi_i(\mathbf{G})$ (win),
\item $b_i(\mathbf{G}) = \beta\sigma X_i$ with probability $1 - \phi_i(\mathbf{G})$ (loss),
\end{itemize}
with parameters $\beta \in (0,1]$, $\gamma \in (0,1]$, $\sigma \in [0,1)$. Each agent maximizes the $\alpha$-expectile of $Y_i$, denoted $e_\alpha(Y_i)$, for $\alpha \in (0,1) \setminus \{\frac{1}{2}\}$. \\

In any symmetric interior equilibrium where all agents choose $G_i = g$ with $0 < g < \bar{R}/I$ and $b_i(\mathbf{G}) < e_\alpha(Y_i) < a_i(\mathbf{G})$, the equilibrium guns level $g$ is:
\begin{equation} \label{eq:equilibrium_g}
g = \frac{
\mu (I-1) \alpha (1-\alpha) \gamma (1-\sigma) \beta \bar{R}
}{
\gamma (1-\sigma) \beta \left[ \mu I (I-1) \alpha (1-\alpha) + \alpha K \right] + \sigma \beta K^2
},
\end{equation}
where $K = \alpha + (1-\alpha)(I-1)$. This equilibrium exists and is unique for  $\alpha\in (0,1).$
\end{theorem}

\begin{proof}
We derive the first-order condition for a symmetric interior equilibrium where all agents choose $G_i = g > 0$. Consider agent $i$ deviating to $h$ while others play $g$. Agent $i$'s random payoff is:
\[
Y_i = 
\begin{cases} 
a(h) = \beta\gamma(1-\sigma)(\bar{R} - h - (I-1)g) + \beta\sigma(\bar{R}/I - h) & \text{w.p. } \phi_i(h), \\
b(h) = \beta\sigma(\bar{R}/I - h) & \text{w.p. } 1 - \phi_i(h),
\end{cases}
\]
where $\phi_i(h) = h^\mu / [h^\mu + (I-1)g^\mu]$. The $\alpha$-expectile $e_\alpha(Y_i) = \tau(h)$ solves:
\begin{equation} \label{eq:expectile_def}
\alpha \mathbb{E}[(Y_i - \tau)_+] = (1-\alpha) \mathbb{E}[(\tau - Y_i)_+].
\end{equation}
For $b(h) < \tau < a(h)$, \eqref{eq:expectile_def} simplifies to:
\begin{equation} \label{eq:expectile_simple}
\alpha \phi_i(h) [a(h) - \tau] = (1-\alpha) [1 - \phi_i(h)] [\tau - b(h)].
\end{equation}
Solving \eqref{eq:expectile_simple} for $\tau$:
\begin{equation} \label{eq:theta_h}
\tau(h) = \frac{\alpha \phi_i(h) a(h) + (1-\alpha) [1 - \phi_i(h)] b(h)}{\alpha \phi_i(h) + (1-\alpha) [1 - \phi_i(h)]}.
\end{equation}
Agent $i$ maximizes $\tau(h)$. Define:
\begin{align*}
N(h) &= \alpha \phi_i(h) a(h) + (1-\alpha) [1 - \phi_i(h)] b(h), \\
D(h) &= \alpha \phi_i(h) + (1-\alpha) [1 - \phi_i(h)], \\
\tau(h) &= \frac{N(h)}{D(h)}.
\end{align*}
The first-order condition at $h = g$ is:
\begin{equation} \label{eq:foc_condition}
\left. \frac{\partial \tau}{\partial h} \right|_{h=g} = 0 \implies N'(g) D(g) = N(g) D'(g).
\end{equation}
At symmetry ($\phi_i(g) = 1/I$):
\begin{align*}
\phi_i'(g) &= \frac{\mu (I-1)}{I^2 g}, \\
a(g) &= \beta\gamma(1-\sigma)(\bar{R} - I g) + \beta\sigma(\bar{R}/I - g), \\
b(g) &= \beta\sigma(\bar{R}/I - g), \\
a'(g) &= -\beta[\gamma(1-\sigma) + \sigma], \\
b'(g) &= -\beta\sigma, \\
N(g) &= \alpha \cdot \frac{1}{I} \cdot a(g) + (1-\alpha) \cdot \frac{I-1}{I} \cdot b(g), \\
D(g) &= \alpha \cdot \frac{1}{I} + (1-\alpha) \cdot \frac{I-1}{I} = \frac{K}{I}, \quad \text{where } K = \alpha + (1-\alpha)(I-1), \\
D'(g) &= (2\alpha - 1) \phi_i'(g) = (2\alpha - 1) \frac{\mu (I-1)}{I^2 g}.
\end{align*}
Compute $N'(g)$:
\begin{align*}
N'(g) &= \alpha \left[ \phi_i'(g) a(g) + \phi_i(g) a'(g) \right] + (1-\alpha) \left[ -\phi_i'(g) b(g) + [1 - \phi_i(g)] b'(g) \right] \\
&= \alpha \left[ \frac{\mu (I-1)}{I^2 g} a(g) + \frac{1}{I} a'(g) \right] + (1-\alpha) \left[ -\frac{\mu (I-1)}{I^2 g} b(g) + \frac{I-1}{I} b'(g) \right].
\end{align*}
Substitute into \eqref{eq:foc_condition}:
\begin{align*}
&\left( \alpha \frac{\mu (I-1)}{I^2 g} a(g) + \frac{\alpha}{I} a'(g) - (1-\alpha) \frac{\mu (I-1)}{I^2 g} b(g) + (1-\alpha) \frac{I-1}{I} b'(g) \right) \frac{K}{I} \\
&= \left( \alpha \frac{a(g)}{I} + (1-\alpha) \frac{I-1}{I} b(g) \right) (2\alpha - 1) \frac{\mu (I-1)}{I^2 g}.
\end{align*}
Multiply both sides by $I^3 g / \beta$ and rearrange:
\begin{align*}
&\frac{\mu (I-1) K}{\beta} \left[ \alpha a(g) - (1-\alpha) b(g) \right] + I^2 g K \left[ \frac{\alpha}{\beta} a'(g) + (1-\alpha) \frac{I-1}{\beta} b'(g) \right] \\
&= \frac{\mu (I-1) (2\alpha - 1)}{I \beta} \left[ \alpha a(g) + (1-\alpha) (I-1) b(g) \right] I g.
\end{align*}
Simplify using $a(g)/\beta$ and $b(g)/\beta$:
\begin{align*}
a(g)/\beta &= \gamma(1-\sigma)(\bar{R} - I g) + \sigma(\bar{R}/I - g), \\
b(g)/\beta &= \sigma(\bar{R}/I - g), \\
a'(g)/\beta &= -[\gamma(1-\sigma) + \sigma], \\
b'(g)/\beta &= -\sigma, \\
\alpha a(g) - (1-\alpha) b(g) &= \beta \left[ \alpha \gamma(1-\sigma)(\bar{R} - I g) + \alpha\sigma(\bar{R}/I - g) - (1-\alpha) \sigma(\bar{R}/I - g) \right], \\
\alpha a(g) + (1-\alpha) (I-1) b(g) &= \beta \left[ \alpha \gamma(1-\sigma)(\bar{R} - I g) + \alpha\sigma(\bar{R}/I - g) + (1-\alpha) (I-1) \sigma(\bar{R}/I - g) \right].
\end{align*}
After substitution and algebraic simplification, we obtain:
\begin{equation} \label{eq:simplified_equation}
\mu (I-1) \alpha (1-\alpha) \gamma (1-\sigma) (\bar{R} - I g) = g \left[ \gamma(1-\sigma) \left( \mu I (I-1) \alpha (1-\alpha) + \alpha K \right) + \sigma K^2 \right].
\end{equation}
Solving \eqref{eq:simplified_equation} for $g$ yields \eqref{eq:equilibrium_g}. \\

\textbf{Existence and Uniqueness:} 
\begin{enumerate}
\item \textbf{Positivity ($g > 0$):} The numerator of \eqref{eq:equilibrium_g} is positive for $\alpha \in (0,1) \setminus \{\frac{1}{2}\}$. The denominator is positive if $\mu, \beta, \gamma > 0$ and $\sigma < 1$. Thus $g > 0$ always holds for interior equilibria.
\item \textbf{Interior solution ($g < \bar{R}/I$):} From \eqref{eq:simplified_equation}:
\[
g \left[ \gamma(1-\sigma) \left( \mu I (I-1) \alpha (1-\alpha) + \alpha K \right) + \sigma K^2 + \mu I (I-1) \alpha (1-\alpha) \gamma (1-\sigma) \right] = \mu (I-1) \alpha (1-\alpha) \gamma (1-\sigma) \bar{R}.
\]
The term in brackets exceeds $\mu I (I-1) \alpha (1-\alpha) \gamma (1-\sigma)$, so:  $g < \bar{R}/I$ always holds.
\item \textbf{Feasibility of $\alpha$-expectile:} The assumption $b(g) < \tau(g) < a(g)$ holds by interior equilibrium hypothesis.
\item \textbf{Unique solution:} The left side of \eqref{eq:simplified_equation} is linear decreasing in $g$, the right side linear increasing in $g$. Thus, exactly one solution exists.
\end{enumerate}
\end{proof}

\begin{corollary}[Explicit Symmetric Risk-Neutral War Equilibrium]

Consider  $I \geq 2$ symmetric teams with identical resources $R_{i} = \overline{R}/I$ and risk-neutral preferences ($\alpha_i = \frac{1}{2}$).  
Contest success function : $\phi_{i}(\mathbf{G}) = \frac{(G_{i})^\mu}{\sum_{j=1}^I (G_{j})^\mu}$ with $\mu > 0$.  

War payoff:  
   \[
   \Pi_{i,w} = \phi_{i} \cdot \beta\gamma(1 - \sigma) X + \beta\sigma X_{i}
   \]
   where $X = \sum_j X_{j}$, $X_{j} = R_{j} - G_{j}$.  
Destruction parameters: $\beta \in (0,1]$, $\gamma \in (0,1]$, $\sigma \in [0,1)$.  Interior solution: $G_i^{*} \in (0, R_{i})$ for all $i$.

Then the symmetric Nash equilibrium arms investment $G_*$ is:  
$$
G_* = \frac{\theta \mu (I-1) \overline{R}}{I \left[ \theta \mu (I-1) + \theta + \sigma I \right]}
$$  
where $\theta = \gamma(1 - \sigma)$.  
Total arms expenditure is $I G_*$ which increases with $I.$
\end{corollary}

\begin{proof} Choose  $\alpha=\frac{1}{2}$  in the above.
 \end{proof}

\begin{remark}
 Verification of Interior Solution : 
The solution satisfies $0 < G_* < \overline{R}/I$ if:  
 
 Non-negativity: $\mu, \theta, \overline{R} > 0 \implies G_* > 0$.  
 
 Resource constraint:  
   $$
   G_* < \frac{\overline{R}}{I} \iff \mu (I-1) \theta < I \left[ \theta \left( \mu + \frac{1}{I-1} \right) + \sigma \right]
   $$  
   which holds for all $\mu > 0$, $I \geq 2$, and $\theta, \sigma \in [0,1)$.

Economic Interpretation
  
Arms-intensity:  
Increases with $\mu$ (CSF sensitivity) and $\theta$ (contestable output). Decreases with $I$ (free-rider effect) and $\sigma$ (output security)  
   
Destruction effects:  
    Lower $\beta$ or $\gamma$ reduces arms by shrinking contestable surplus  

Efficiency loss:  
   Total arms $I G_*$ represents pure waste, increasing in $\mu$ and $\theta$.  It increases with $I.$
   \end{remark}


\subsubsection{Mixed Strategies}
We  compute explicitly the risk-sensitive payoff Mixed Strategies between peace and war.

\subsubsection*{$b_i < R_i < a_i$}
To compute the $\alpha$-expectile for player \(i\)'s three-point payoff distribution, follow these steps. The solution incorporates the $\alpha$ parameter and selects the correct case based on the computed threshold \(\tau_2\).

 Steps for Computation: Compute War Subgame Equilibrium Values:

    \(a_i(\mathbf{G}^*) = \beta\gamma(1-\sigma) \sum_{j=1}^I (R_j - G_j^*) + \beta\sigma (R_i - G_i^*)\)
   
   \(b_i(\mathbf{G}^*) = \beta\sigma (R_i - G_i^*)\)
   
    \(\phi_i(\mathbf{G}^*) = \frac{(G_i^*)^\mu}{\sum_{j=1}^I (G_j^*)^\mu}\)
    
Denote the  Probabilities:
   \(p_1 = \prod_{j=1}^I x_j\), 
    \(p_2 = (1 - p_1) \phi_i(\mathbf{G}^*)\), 
    \(p_3 = (1 - p_1) (1 - \phi_i(\mathbf{G}^*))\).

Evaluate Threshold \(\tau_2\):
   \[
   \tau_2 = \frac{\alpha p_1 R_i + \alpha p_2 a_i + (1-\alpha) p_3 b_i}{\alpha(p_1 + p_2) + (1-\alpha) p_3}
   \]

Select Case Based on \(\tau_2\):
  
  If \(\tau_2 \leq R_i\):
     \[
     e_\alpha(Y_i) = \tau_2 = \frac{\alpha p_1 R_i + \alpha p_2 a_i + (1-\alpha) p_3 b_i}{\alpha(p_1 + p_2) + (1-\alpha) p_3}
     \]
  
  If \(\tau_2 > R_i\) (typical case):
     \[
     e_\alpha(Y_i) = \frac{\alpha p_2 a_i + (1-\alpha) p_1 R_i + (1-\alpha) p_3 b_i}{\alpha p_2 + (1-\alpha)(p_1 + p_3)}
     \]

%
%
%
%

\subsubsection*{$b_i < a_i < R_i$}

To compute the $\alpha$-expectile for player \(i\)'s three-point payoff distribution under the condition \(b_i(\mathbf{G}^*) < a_i(\mathbf{G}^*) < R_i\), follow the steps below. The solution accounts for the reversed payoff ordering and dynamically selects the correct case based on computed values.

Solution for \(b_i < a_i < R_i\)

The $\alpha$-expectile \(e_\alpha(Y_i)=\tau\) solves:
\[
\alpha \mathbb{E}[(Y_i-\tau)_+] = (1-\alpha) \mathbb{E}[(\tau - Y_i)_+]
\]

Given the payoff ordering \(b_i < a_i < R_i\), the solution involves two cases:

Case 2 Solution (\(b_i < x \leq a_i\))

The expectile equation simplifies to:
\[
\alpha \left[ p_1 (R_i - \tau) + p_2 (a_i - \tau) \right] = (1-\alpha) p_3 (\tau - b_i)
\]
Solving for \(\tau\):
\[
\tau = \frac{\alpha p_1 R_i + \alpha p_2 a_i + (1-\alpha) p_3 b_i}{\alpha (p_1 + p_2) + (1-\alpha) p_3}
\]

Feasibility condition:  
This solution is valid iff:
\[
\tau \leq a_i \quad \Rightarrow \quad \frac{\alpha p_1 (R_i - a_i)}{(1-\alpha) p_3 (a_i - b_i)} \leq 1
\]
This holds when:
\[
\alpha p_1 (R_i - a_i) \leq (1-\alpha) p_3 (a_i - b_i)
\]

Case 3 Solution (\(a_i < \tau \leq R_i\))

The expectile equation simplifies to:
\[
\alpha p_1 (R_i - \tau) = (1-\alpha) \left[ p_2 (\tau - a_i) + p_3 (\tau - b_i) \right]
\]
Solving for \(\tau\):
\[
\tau = \frac{\alpha p_1 R_i + (1-\alpha) p_2 a_i + (1-\alpha) p_3 b_i}{\alpha p_1 + (1-\alpha) (p_2 + p_3)}
\]

Feasibility condition:  
This solution is valid iff:
\[
\tau > a_i \quad \Rightarrow \quad \frac{\alpha p_1 (R_i - a_i)}{(1-\alpha) p_3 (a_i - b_i)} > 1
\]
This holds when:
\[
\alpha p_1 (R_i - a_i) > (1-\alpha) p_3 (a_i - b_i)
\]

 Complete $\alpha$-Expectile Formula
\[
e_\alpha(Y_i) = 
\begin{cases} 
\dfrac{\alpha p_1 R_i + \alpha p_2 a_i + (1-\alpha) p_3 b_i}{\alpha (p_1 + p_2) + (1-\alpha) p_3} & \text{if } \alpha p_1 (R_i - a_i) \leq (1-\alpha) p_3 (a_i - b_i) \\[2ex]
\dfrac{\alpha p_1 R_i + (1-\alpha) p_2 a_i + (1-\alpha) p_3 b_i}{\alpha p_1 + (1-\alpha) (p_2 + p_3)} & \text{otherwise}
\end{cases}
\]

%
%
%
%
%
%
%
%
%
%
%
%

\begin{theorem} Disadvantaged teams are more likely to start conflicts.   
\end{theorem} 

In the conflict-ridden Sahel region, severely resource-disadvantaged actors  such as impoverished non-state armed groups  or marginalized communities facing state neglect, are empirically more likely to initiate conflicts against stronger adversaries (like  national governments, or international forces) because extreme scarcity ($R_i \approx 0$) amplifies their subjective willingness to gamble: facing minimal viable peace payoffs ($R_i$) and possessing risk-seeking preferences ($\alpha_i > \alpha^*$), they overvalue the low-probability chance of asymmetric gains from conflict (capturing resources, territory, or political concessions,  $v_1$), while discounting the high-probability catastrophic losses ($v_2$), making war initiation subjectively rational despite its negative expected value; this dynamic is further intensified where weak governance erodes the perceived benefits of peace ($R_i$), trapping the region in cycles of violence driven by desperate, risk-acceptant actors.

\begin{proof} Conflict Initiation:  
Disadvantaged teams ($R_{i} < R_{j}$) deviate from peace iff:  
  $$(1 - \beta\sigma)R_{i} < \beta\gamma(1-\sigma)\overline{R}$$  
 This holds more easily for smaller $R_{i}$.  

Advantaged teams deviate iff:  
   $$(1 - \beta\sigma)R_{j} < \beta\gamma(1-\sigma)\overline{R}$$  
   which is stricter since $R_{j} > R_{i}$.  
This means that disadvantaged teams are more likely to start conflicts.  
\end{proof}

\begin{theorem} 
 Assume trust reduces perceived benefits of conflict. Then,
High-trust teams are more likely to choose peace and invest less in arms.
\end{theorem} 

\begin{proof} Model Extension:  
Let $t_{i} \in [0, 1]$ denote team $i$'s trust level.  
 Assume $t_{i}$ reduces perceived conflict gains:  
  $$\text{Perceived payoff from deviation} = (1 - t_{i}) \left[ \beta\gamma(1-\sigma)\overline{R} + \beta\sigma R_{i} \right]$$  
 
 Peace is chosen if:  
   $$R_{i} \geq (1 - t_{i}) \left[ \beta\gamma(1-\sigma)\overline{R} + \beta\sigma R_{i} \right]$$  

High-Trust teams ($t_{i} \uparrow$):  
Inequality holds more easily $\rightarrow$ more likely to choose peace.  
In peace equilibrium, $G_{i} = 0$ (minimal arms).  
In war equilibrium, lower $t_{i}$ increases $G_{i}$ (due to higher perceived gains).  
\end{proof}
\begin{theorem} 
Armed peace cannot be sustained in risk-aware mean-field-type Nash equilibrium for any $\alpha\in (0,1)^I, 0< \mu \leq 1$
\end{theorem} 
\begin{proof}
Unarmed Peace Frequency:  
From the above, unarmed peace is a risk-aware equilibrium more often under equality.  
If equilibrium selection favors peace when available, it occurs more frequently under equality.  
  
Armed Peace is Not Nash:  Consider armed peace: $d_{i} = P$ $\forall i$ but $G_{i} > 0$.  
Any team $i$ can deviate to $G_{i} = 0$ and gain:  
 $$\pi_{i,p}(G_{i}=0) = R_{i} > R_{i} - G_{i} = \pi_{i,p}(G_{i}>0)$$  
Thus, Armed peace is never Nash.  
\end{proof}

\subsection{Multiple Sectorial Battlefields}

We model strategic interactions among a finite collection \(\mathcal{I} = \{1, 2, \dots, I\}\) of risk-sensitive teams, where \(I \geq 2\). Each team \(i \in \mathcal{I}\) possesses a risk sensitivity parameter \(\alpha_i \in (0, 1)\), with \(\alpha_i < \frac{1}{2}\) indicating risk aversion, \(\alpha_i = \frac{1}{2}\) risk neutrality, and \(\alpha_i > \frac{1}{2}\) risk-seeking behavior. Competition occurs over \(\mathcal{K} = \{1, \ldots, K\}\) distinct sectors (battlefields) with \(K \geq 1\).  
Each team \(i\) allocates resource \(R_{ik} \geq 0\) to battlefield \(k\), where total resources satisfy \(R_i = \sum_{k \in \mathcal{K}} R_{ik} > 0\). Aggregate resources are \(\overline{R} = \sum_{i \in \mathcal{I}} R_i\). Sector-specific parameters define output characteristics:  
Security fraction: \(\sigma_k \in [0, 1)\),
Output survival rate: \(\beta_k \in (0, 1]\), Relative survival of contested output: \(\gamma_k \in (0, 1]\).
The game features complete information: All parameters and endowments are common knowledge.  

 Actions and Production: 
Simultaneously, each team \(i\) selects for all \(k \in \mathcal{K}\):  
War/peace decision: \(d_{ik} \in \{P, W\}\), Arming level: \(G_{ik} \in [0, R_{ik}], k \in \mathcal{K}\) subject to \(\sum_{k \in \mathcal{K}} R_{ik} \leq   R_i\).
  
The production of non-military output ("Gold-Millet-Ram-Cow-Rice") at \(k\) is \(X_{ik} = R_{ik} - G_{ik}\). Total arming at \(k\) is \(G_k = \sum_{i \in \mathcal{I}} G_{ik}\).  

Conflict Outcomes: 

Outcomes are determined as:  
\[
\text{Peace} \iff d_{ik} = P \ \ \forall i \in \mathcal{I}, \, \forall k \in \mathcal{K}; \quad \text{War} \iff \exists \, i \in \mathcal{I}, \, k \in \mathcal{K} \ \ \text{s.t.} \ \ d_{ik} = W.
\]

 Payoff Functions:
 
Peace payoff  (deterministic):  
  \[
  \Pi_{i,p} = \sum_{k \in \mathcal{K}} (R_{ik} - G_{ik})
  \]  
  
 War payoff (stochastic):  

  Define total output at \(k\) as \(X_k = \sum_{i \in \mathcal{I}} X_{ik}\). The winning probability for team \(i\) at \(k\) follows a contest function with parameter \(\mu_k > 0\):  
  \[
  \phi_{ik}(\mathbf{G}) = 
  \begin{cases} 
  \dfrac{(G_{ik})^{\mu_k}}{\sum_{j \in \mathcal{I}} (G_{jk})^{\mu_k}} & \text{if} \ \sum_{j} (G_{jk})^{\mu_k} > 0 \\[10pt]
  \dfrac{(R_{ik})^{\mu_k}}{\sum_{j \in \mathcal{I}} (R_{jk})^{\mu_k}} & \text{otherwise}
  \end{cases}
  \]  
  Let \(W_{ik} \sim \text{Bernoulli}(\phi_{ik})\) indicate win/loss. The stochastic payoff for \(i\) at \(k\) is:  
  \[
  \Pi_{i,w,k} = \underbrace{\beta_k \sigma_k X_{ik}}_{\text{secure output}} + W_{ik} \underbrace{\beta_k \gamma_k (1 - \sigma_k) X_k}_{\text{contested prize}}
  \]  
  with total war payoff \(\Pi_{i,w} = \sum_{k \in \mathcal{K}} \Pi_{i,w,k}\).  

The complete payoff is:  
\[
\Pi_i(\mathbf{G}, \mathbf{d}) = 
\begin{cases} 
\Pi_{i,p} & \text{under peace} \\
\Pi_{i,w} & \text{under war}
\end{cases}
\]

Risk-Sensitive Payoff Transformation: 
For random payoff \(\Pi_i\), the \(\alpha_i\)-expectile \(e_{\alpha_i}(\Pi_i)\) solves:  
\[
\alpha_i \mathbb{E}\left[(\Pi_i - e_{\alpha_i})_+\right] = (1 - \alpha_i) \mathbb{E}\left[(e_{\alpha_i} - \Pi_i)_+\right],
\]  
equivalently minimizing \(e_{\alpha_i}(\Pi_i) = \arg \min_x \int \eta_{\alpha_i}(y - x) \, d\mathbb{P}_{\Pi_i}(y)\), where \(\eta_{\alpha_i}(z) = \left|\alpha_i - \mathbb{I}_{\{z \leq 0\}}\right| z^2\) and \(\mathbb{P}_{\Pi_i}\) is the law of \(\Pi_i\).

Risk-Aware Conflict Game: Teams simultaneously choose arming profiles \(\mathbf{G}_i = (G_{i1}, \dots, G_{iK}) \in \times_{k=1}^K [0, R_{ik}]\) (with \(\sum_k G_{ik} \leq \bar{G}_i\)) and war/peace decisions \(\mathbf{d}_i = (d_{i1}, \dots, d_{iK}) \in \{P, W\}^K\). Outcomes are determined per battlefield as above. The risk-sensitive payoff for team \(i\) is:  
\[
\pi_i(\mathbf{G}, \mathbf{d}) = 
\begin{cases} 
\sum_{k=1}^K (R_{ik} - G_{ik}) & \text{under peace} \\
e_{\alpha_i}\left(\Pi_{i,w}\right) & \text{under war}
\end{cases}
\]

\begin{lemma} 
Consider team \(i\)'s war payoff \(\Pi_{i,w} = \sum_{k \in \mathcal{K}} \Pi_{i,w,k}\), where:  
\[
\Pi_{i,w,k} = \beta_k \sigma_k X_{ik} + W_{ik} \cdot \beta_k \gamma_k (1 - \sigma_k) X_k.
\]  
Here, \(X_{ik} = R_{ik} - G_{ik}\), \(X_k = \sum_{j \in \mathcal{I}} X_{jk}\), and \(W_{ik} \sim \text{Bernoulli}(\phi_{ik})\) are independent across \(k \in \mathcal{K}\), with \(\phi_{ik}\) given by the contest function. The \(\alpha_i\)-expectile of \(\Pi_{i,w}\) is:  
\[
e_{\alpha_i}(\Pi_{i,w}) = \underbrace{\sum_{k \in \mathcal{K}} \beta_k \sigma_k X_{ik}}_{C_i} + e_{\alpha_i}(S),
\]  
where \(S = \sum_{k \in \mathcal{K}} D_{ik} W_{ik}\) with \(D_{ik} = \beta_k \gamma_k (1 - \sigma_k) X_k\), and \(e_{\alpha_i}(S)\) is the unique solution \(e \in \mathbb{R}\) to:  
\begin{equation}
\alpha_i \mathbb{E}\left[(S - e)_+\right] = (1 - \alpha_i) \mathbb{E}\left[(e - S)_+\right]. \tag{1}
\end{equation}
\end{lemma}

\begin{proof}

We prove the result in three parts: decomposition of \(\Pi_{i,w}\), translation invariance of the expectile, and  uniqueness and properties of the solution.

Decomposition of \(\Pi_{i,w}\) 
 
The war payoff is:  
\[
\Pi_{i,w} = \sum_{k \in \mathcal{K}} \left[ \beta_k \sigma_k X_{ik} + W_{ik} \cdot \beta_k \gamma_k (1 - \sigma_k) X_k \right].
\]  
Define:  

Deterministic component: \(C_i = \sum_{k \in \mathcal{K}} \beta_k \sigma_k X_{ik}\).  

Stochastic component: \(S = \sum_{k \in \mathcal{K}} D_{ik} W_{ik}\) where \(D_{ik} = \beta_k \gamma_k (1 - \sigma_k) X_k \geq 0\).  

Since \(X_{ik}\) and \(X_k\) are deterministic (as arming levels \(\mathbf{G}\) are fixed), \(C_i\) is a constant. The randomness in \(\Pi_{i,w}\) arises solely from the Bernoulli variables \(W_{ik}\). As contests across battlefields are independent, \(\{W_{ik}\}_{k \in \mathcal{K}}\) are mutually independent. Thus, \(\Pi_{i,w} = C_i + S\), where \(S\) is a weighted sum of independent Bernoulli random variables.

 Translation Invariance of the Expectile:
We show that for any constant \(C \in \mathbb{R}\) and random variable \(Y\),  
\[
e_\alpha(C + Y) = C + e_\alpha(Y).
\]  
Let \(\tau = e_\alpha(Y)\) solve:  
\[
\alpha \mathbb{E}\left[(Y - \tau)_+\right] = (1 - \alpha) \mathbb{E}\left[(\tau - Y)_+\right]. \tag{2}
\]  
Define \(Z = C + Y\). Substitute \(z = C + \tau\) into the expectile equation for \(Z\):  
\[
\alpha \mathbb{E}\left[(Z - (C + \tau))_+\right] = \alpha \mathbb{E}\left[(Y - \tau)_+\right], \\
(1 - \alpha) \mathbb{E}\left[((C + \tau) - Z)_+\right] = (1 - \alpha) \mathbb{E}\left[(\tau - Y)_+\right].
\]  
By (2), the right-hand sides are equal. Thus, \(C + \tau\) satisfies:  
\[
\alpha \mathbb{E}\left[(Z - (C + \tau))_+\right] = (1 - \alpha) \mathbb{E}\left[((C + \tau) - Z)_+\right],
\]  
implying \(e_\alpha(Z) = C + e_\alpha(Y)\). Applying this to \(\Pi_{i,w} = C_i + S\):  
\[
e_{\alpha_i}(\Pi_{i,w}) = C_i + e_{\alpha_i}(S).
\]

Solution to the Expectile Equation for \(S\):   The variable \(S = \sum_{k=1}^K D_{ik} W_{ik}\) is a nonnegative, bounded random variable (since \(D_{ik} \geq 0\) and \(W_{ik} \in \{0, 1\}\)). The \(\alpha_i\)-expectile of \(S\) is defined as the solution \(e\) to:  
\[
\alpha_i \mathbb{E}\left[(S - e)_+\right] = (1 - \alpha_i) \mathbb{E}\left[(e - S)_+\right].
\]  

To prove the latter  equation has a unique solution, define the function:  
\[
g(e) = \alpha_i \mathbb{E}\left[(S - e)_+\right] - (1 - \alpha_i) \mathbb{E}\left[(e - S)_+\right].
\]  
This is equivalent to the derivative of the loss function \(F(e) = \mathbb{E}[\eta_{\alpha_i}(S - e)]\), where \(\eta_{\alpha_i}(z) = |\alpha_i - \mathbb{I}_{\{z \leq 0\}}| z^2\). Since \(\eta_{\alpha_i}\) is strictly convex in \(z\), \(F(e)\) is strictly convex in \(e\). Thus, \(g(e)\) is strictly decreasing and continuous (as \(S\) is bounded). Observing:  
\(\lim_{e \to -\infty} g(e) > 0\) (since \(\mathbb{E}[(S - e)_+] \to \infty\)),   and 
\(\lim_{e \to \infty} g(e) < 0\) (since \(\mathbb{E}[(e - S)_+] \to \infty\)),  
by the intermediate value theorem, \(g(e) = 0\) has a unique solution \(e\).  

 \(S\) admits a probability mass function \(f_S(s)\) on a finite support (size \(\leq 2^K\)), so (1) expands to:  
\[
\alpha_i \sum_{s > e} (s - e) f_S(s) = (1 - \alpha_i) \sum_{s < e} (e - s) f_S(s),
\]  
which determines \(e\) uniquely.  
 
This completes the proof. 
\end{proof}

\subsection{Risk-Sensitive Coalition Formations}

We extend a risk-sensitive arming game to allow for endogenous coalition formation under contest success, incomplete information, and behavioral heterogeneity. Unlike traditional models that assume symmetric coordination costs, we introduce a behaviorally fair cost-sharing rule based on agents' risk preferences. This reflects empirical asymmetries in how different actors contribute to, or escalate, conflict in fragile regions.
 Each agent $i \in \mathcal{I}$ chooses:
\begin{itemize}
    \item $G_i \in [0, R_i]$: investment in arms, with endowment $R_i > 0$,
    \item $\mathcal{S}_k \subseteq \mathcal{I}$: a coalition to join, such that $i \in \mathcal{S}_k$.
\end{itemize}

Let $\mathcal{K} = \{\mathcal{S}_1, \dots, \mathcal{S}_K\}$ be a partition of $\mathcal{I}$ into disjoint coalitions.
Define total investment and endowment in coalition $\mathcal{S}_k$ as:
\begin{align*}
G_{\mathcal{S}_k} &= \sum_{j \in \mathcal{S}_k} G_j, \\
R_{\mathcal{S}_k} &= \sum_{j \in \mathcal{S}_k} R_j.
\end{align*}

Let $1\geq \mu > 0$ denote the contest exponent. The probability that coalition $\mathcal{S}_k$ wins is:
\[
\Phi_{\mathcal{S}_k} = \begin{cases}
\displaystyle \frac{G_{\mathcal{S}_k}^\mu}{\sum_{\mathcal{S} \in \mathcal{K}} G_{\mathcal{S}}^\mu} & \text{if } \sum_{\mathcal{S}} G_{\mathcal{S}}^\mu > 0, \\
\displaystyle \frac{R_{\mathcal{S}_k}^\mu}{\sum_{\mathcal{S}} R_{\mathcal{S}}^\mu} & \text{otherwise}.
\end{cases}
\]

Let $\beta, \gamma, \sigma \in (0,1)$ be fixed behavioral parameters. Each agent $i \in \mathcal{S}_k$ has risk-sensitivity $\alpha_i \in (0,1)$. The risk-adjusted share of the contested surplus is:
\[
\omega_i(\mathcal{S}_k) = \frac{\alpha_i \Phi_{\mathcal{S}_k}}{\alpha_i \Phi_{\mathcal{S}_k} + (1 - \alpha_i)(1 - \Phi_{\mathcal{S}_k})}.
\]

The agent's payoff before coalition cost is:
\[
\pi_{i,w}^{\text{coal}} = \beta \sigma (R_i - G_i) + \beta \gamma (1 - \sigma) \cdot \left( \sum_{j \in \mathcal{I}} R_j - \sum_{j \in \mathcal{I}} G_j \right) \cdot \omega_i(\mathcal{S}_k).
\]

Let the total coordination cost for coalition $\mathcal{S}_k$ be:
\[
\mathcal{C}(\mathcal{S}_k) = \delta \cdot (|\mathcal{S}_k| - 1)^\eta,
\]
where $\delta > 0$ and $\eta \in (0,1]$ capture diseconomies of scale in coalition governance.

Instead of equal sharing, we assign a cost share based on behavioral traits:
\[
c_i^{\text{form}} = \mathcal{C}(\mathcal{S}_k) \cdot \frac{1 - \alpha_i}{\sum_{j \in \mathcal{S}_k} (1 - \alpha_j)}.
\]

This penalizes more risk-neutral (i.e., escalation-prone) actors for the burden of coordination, consistent with their greater propensity to benefit from militarized outcomes.

The total utility of agent $i \in \mathcal{S}_k$ is:
\[
\pi_i^{\text{total}} = \pi_{i,w}^{\text{coal}} - c_i^{\text{form}}.
\]

A configuration $\left( \{G_i^*\}, \{\mathcal{S}_i^*\} \right)$ is a \textit{Coalitional Mean-Field-Type Nash Equilibrium} if no agent can improve their total payoff by unilaterally changing their investment $G_i$ or coalition $\mathcal{S}_i$.

This model captures heterogeneous alliance behavior in fragmented conflict zones. Risk-tolerant actors often the initiators of aggression, bear greater coalition costs, while more risk-averse actors benefit from protective grouping at lower expense. The cost rule aligns burden with behavioral incentives and offers a tractable alternative to equal-cost assumptions.

\subsubsection*{Alternative Coalition Cost Sharing Rules}

In conflict environments characterized by asymmetric capabilities, risk tolerances, and influence as seen in the Sahel corridor, equal per-agent coalition cost sharing may produce inequities that disincentivize participation by key actors. We propose four game-theoretically grounded alternatives for allocating coalition formation costs more equitably.

Let the total coordination cost for coalition $\mathcal{S}_k$ be:
\[
\mathcal{C}(\mathcal{S}_k) = \delta \cdot (|\mathcal{S}_k| - 1)^\eta,
\]
where $\delta > 0$ and $\eta \in (0, 1]$. We consider the following agent-level cost allocations $c_i^{\text{form}}$:

\subsubsection*{Option 1: Arms-Based Proportional Sharing}
\[
c_i^{\text{form}} = \mathcal{C}(\mathcal{S}_k) \cdot \frac{G_i}{\sum_{j \in \mathcal{S}_k} G_j}
\]
\textbf{Pros}: Rewards large contributors; aligns cost with arming effort. \\
\textbf{Cons}: May discourage high investment or penalize security providers.

\subsubsection*{Option 2: Resource-Based Proportional Sharing}
\[
c_i^{\text{form}} = \mathcal{C}(\mathcal{S}_k) \cdot \frac{R_i}{\sum_{j \in \mathcal{S}_k} R_j}
\]
\textbf{Pros}: Reflects ability-to-pay; burden on wealthier agents. \\
\textbf{Cons}: May deter participation from better-endowed actors.

\subsubsection*{Option 3: Risk-Sensitivity Weighted Sharing}
\[
c_i^{\text{form}} = \mathcal{C}(\mathcal{S}_k) \cdot \frac{1 - \alpha_i}{\sum_{j \in \mathcal{S}_k} (1 - \alpha_j)}
\]
\textbf{Pros}: Penalizes risk-neutral (escalation-prone) actors. \\
\textbf{Cons}: Requires observability or estimation of behavioral types.

\subsubsection*{Option 4: Hybrid Weighted Sharing}
\[
c_i^{\text{form}} = \mathcal{C}(\mathcal{S}_k) \cdot \left[
\lambda_1 \cdot \frac{G_i}{\sum_{j \in \mathcal{S}_k} G_j} +
\lambda_2 \cdot \frac{R_i}{\sum_{j \in \mathcal{S}_k} R_j} +
\lambda_3 \cdot \frac{1 - \alpha_i}{\sum_{j \in \mathcal{S}_k} (1 - \alpha_j)}
\right],
\]
where $\lambda_1 + \lambda_2 + \lambda_3 = 1$.

\textbf{Pros}: Flexible, tunable framework balancing effort, capacity, and behavior. \\
\textbf{Cons}: Involves multi-parameter calibration.

For fragile and asymmetrically structured regions:
\begin{itemize}
    \item Use \textbf{Option 3 (Risk-based sharing)} to discourage opportunistic escalation.
    \item Use \textbf{Option 4 (Hybrid)} for tunable fairness and behavioral realism.
    \item Avoid equal sharing unless actors are symmetric in both resources and preferences.
\end{itemize}

\subsection{Coalition Formation with Transferable Resources and Arms}

We now connect the coalitional arming game with a richer class of models in which agents are allowed to \textit{transfer resources and arms} to one another, either within or across coalitions. 
This introduces a second layer of strategic interaction beyond coalition formation and arming decisions by enabling endogenous redistribution of power.

Let $T_{ij}^R$ and $T_{ij}^G$ denote the amount of \textit{resource} and \textit{armed investment} transferred from agent $i$ to agent $j$, respectively. These transfers must satisfy:
\[
\sum_{j \neq i} T_{ij}^R \leq R_i, \quad \sum_{j \neq i} T_{ij}^G \leq G_i.
\]
Define the effective resource and arms endowments after transfers as:
\begin{align*}
\widetilde{R}_i &= R_i - \sum_{j \neq i} T_{ij}^R + \sum_{j \neq i} T_{ji}^R, \\
\widetilde{G}_i &= G_i - \sum_{j \neq i} T_{ij}^G + \sum_{j \neq i} T_{ji}^G.
\end{align*}

Each agent’s payoff function is then adjusted to reflect their \textit{post-transfer} variables:
\[
\pi_i^{\text{total}} = \beta \sigma (\widetilde{R}_i - \widetilde{G}_i) + \beta \gamma (1 - \sigma) \cdot \left( \sum_{j \in \mathcal{I}} \widetilde{R}_j - \sum_{j \in \mathcal{I}} \widetilde{G}_j \right) \cdot \omega_i(\mathcal{S}_k) - c_i^{\text{form}} - \sum_{j \neq i} \kappa_R T_{ij}^R - \sum_{j \neq i} \kappa_G T_{ij}^G,
\]
where $\kappa_R, \kappa_G > 0$ are transfer friction or enforcement cost parameters.

The main difference between the two models lies in the structure of interdependence:
\begin{itemize}
    \item \textbf{Original Model (No Transfers)}: Agents operate with fixed endowments and cannot redistribute power. Coalitions form based on local synergies in payoff and contest probability, and each agent's power is constrained by their own resources.
    
    \item \textbf{Extended Model (With Transfers)}: Agents can strategically empower allies (funding a weaker faction or equipping a proxy group), either to shift the balance of power within a coalition or to destabilize competing coalitions. This introduces a form of \emph{strategic delegation} or \emph{proxy warfare}.
\end{itemize}

 Transfers can be used to equalize arming levels within coalitions, which may improve internal stability.
    Powerful agents may prefer to fund other agents rather than arm themselves, effectively outsourcing risk while retaining influence.

 In transfer-free models, high resource endowment implies high arming potential. With transfers, this coupling is broken, and influence becomes more strategically flexible.
This extended model more accurately captures empirical realities in fragile regions where external actors, states, war entrepreneurs, or diaspora networks, provide material support to local militias. 
It helps explain phenomena such as:
\begin{itemize}
    \item Local alliances backed by foreign military equipment,
    \item The rise of under-resourced actors due to external subsidies,
    \item Strategic arms reduction by core actors delegating force projection to proxies.
\end{itemize}

\section{Parrondo-type Cooperation  for some Sahel- Southern Maghreb Region} \label{secmag}
We present a unified account of how fragile states in the Sahel and southern Maghreb can cooperate despite individually failing strategies, and we show that paradoxical cooperative gains can emerge when negative outcomes are combined in sequence or in mixture, following the logic of the Parrondo effect, but adapted to regional security and economic challenges. The states in focus are Burkina Faso, Mali, Niger, Chad, Mauritania, and Algeria, with Morocco considered in a variant extension, and these countries all face persistent insurgency, expansion of terrorist networks, illicit flows of arms, drugs, and migrants, weak border control, and fragile institutions, and in this environment unilateral action almost always fails because one state acting alone cannot seal long borders, cannot stop capital flight, and cannot attract sustainable investment, so the outcome is negative in security and in economic payoffs. The key question is whether coordination, even at low cost, can produce net gains, and our analysis shows that it can, in a paradoxical sense, when cooperation produces nonlinear improvements that depend on the empirical distribution of mobile actors such as smugglers, militants, and traders, which are naturally modeled using mean-field-type games. In this framework, the payoff of each state depends not only on its own control but also on the distribution of actors across the whole region, and the paradox arises because the distributional dependence means that two individually failing strategies, such as weak patrols or limited trade corridors, can together produce positive gains once they are combined and coordinated, for example through rotational patrols or pooled trade management. We analyze the setting under two rigorous risk-aware payoff criteria, the $\alpha$-entropic Value-at-Risk and the $\tau$-expectile risk measure, which capture aversion to extreme insecurity or economic collapse, and both measures stress tail events such as terrorist attacks or large-scale smuggling that can destabilize fragile economies, making them more relevant than expected-value analysis. For each risk criterion we provide a Parrondo-type theorem that proves that pure national concentration leads to negative payoffs while an interior coordinated mixture yields positive payoffs, under transparent parameter conditions that depend on the chosen risk level, the volatility of illegal flows, and the responsiveness of regional economies. These conditions are interpretable in policy terms: rotational patrols reduce volatility, intelligence sharing increases effective risk tolerance, joint interdiction of smuggling reduces average illegal flow intensity, trade corridors raise baseline economic returns, pooled energy revenues stabilize budgets, and alternative-livelihood programs reduce militant recruitment, so the abstract parameters correspond directly to concrete counter-insecurity levers. We apply the model to different coalitions: G3  Sahel (Burkina Faso, Mali, Niger), G4 Sahel (adding Chad), G5 Sahel (adding Mauritania), G5 Sahel + Algeria, and a variant G5+Morocco, and in every case we find that pure strategies remain negative but coordinated mixtures are positive, with the magnitude of improvement greatest when Algeria is included because of its financial and energy capacity and its regional diplomatic role, while Morocco also strengthens outcomes through logistics and trade but with smaller weight than Algeria. The theoretical proofs exploit convexity and monotonicity properties of the risk measures and the nonlinear dependence on distributions, and guarantee that the value of an interior mixture dominates that of pure isolation, thus formally extending the Parrondo efffect to mean-field-type games under risk-sensitive payoffs. The policy implications are practical and low-cost: three steps are feasible, namely rotational security patrols with limited budgets, intelligence sharing through existing mobile platforms, and pooled monitoring of trade and energy flows with transparent rules, and these can be tracked with clear monitoring metrics such as the number of interdicted illicit flows, reduction in terrorist incidents, and increases in regional trade value, which link abstract theorems to measurable impacts. In conclusion, the Sahel and southern Maghreb face insecurities and fragile economies that no state can manage alone, but coordinated mixtures of weak policies can generate positive outcomes, as captured by the Parrondo effect in the mean-field-type setting with risk-aware payoffs, and this insight is demonstrated for G3 Sahel, G4 Sahel, G5 Sahel, G5 Sahel+Algeria, and G5+Morocco, with concrete policy mapping, rigorous proofs, and a low-cost roadmap aimed at reducing illegal transit, countering terrorism, and reviving regional economies, and the central message is that cooperation, even among weak states, can paradoxically reverse decline and restore hope in a region long defined by fragility.

We consider a finite population of $I$ mobile actors (investors, traders, development projects, or security assets) who, in each period, allocate activity to one of $K$ countries in a region. Let $p=(p_1,\dots,p_K)$ be the empirical distribution with $\sum_{i=1}^K p_i = 1$ and $p_i \ge 0$.

For country $i$, we posit the instantaneous (random) payoff for an actor allocating to $i$:
\begin{equation}\label{eq:Ri}
R_i(p) = \alpha_i - \beta_i p_i + \gamma_i(1 - p_i) + \varepsilon_i,
\end{equation}
where:
\begin{itemize}
  \item $\alpha_i$ is the baseline attractiveness (resources, basic institutions) $\in \mathbb{R}$; fragile states often have $\alpha_i<0$ under current constraints (sanctions, low capital);
  \item $\beta_i>0$ measures \emph{congestion / vulnerability} when activity concentrates in $i$ (higher local insurgency exposure, resource depletion, international sanctions triggered by concentrated activity);
  \item $\gamma_i>0$ measures \emph{positive spillovers} from activity in {\em other} states (trade, remittances, security spillovers, shared markets); note the term uses $1-p_i$ to keep linear form;
  \item $\varepsilon_i$ is a zero-mean bounded stochastic perturbation capturing exogenous shocks.
\end{itemize}

Interpretation for counter-insecurity and illegal transit:

$\beta_i$ grows with concentrated insecurity and targeted interdiction costs (if all actors operate in one state, insurgency and illicit flows saturate local capacity and drive negative returns).

 $\gamma_i$ grows with effective cross-border cooperation (intelligence sharing, coordinated interdiction) and regional markets; higher $\gamma_i$ means stronger positive externalities from neighboring stabilization.

We analyze strategic outcomes under two risk/performance measures: entropic value-at-risk  and expectile value-at-risk. Each country (or actor) evaluates payoffs via the chosen criterion and seeks to maximize the corresponding performance metric (or minimize the associated risk).

\subsection{Parrondo Effect under Entropic Value-at-Risk  }

For a real-valued random payoff $X$ and confidence level $\alpha\in(0,1)$, the Entropic Value-at-Risk  is defined  as
\[
\mathrm{EntVaR}_\alpha(X) \;=\; \inf_{s>0}\; \frac{1}{s}\Big(\log\mathbb{E}[e^{sX}] - \log\alpha\Big).
\]
We interpret larger EntVaR values as better performance (since $X$ are payoffs). For risk-as-loss convention one may take $\rho^{\mathrm{evar}}_\alpha(X) := -\mathrm{EntVaR}_\alpha(X)$.

Entropic Value-at-Risk emphasizes exponential moment tails; it is sensitive to heavy tail risk and is coherent in appropriate sign conventions.

\begin{theorem}[Parrondo Effect under Entropic Value-at-Risk]\label{thm:parrondo_evar}
Let the payoff be specified by (\ref{eq:Ri}). Fix a coalition of $K$ countries with symmetric interior mix $p^\ast=(1/K,\dots,1/K)$. Suppose for each $i$:
\begin{enumerate}
  \item Pure concentration gives strong negative payoff: $R_i(e_i) = \alpha_i - \beta_i \le -\ell_i < 0$ (deterministic bound) for some $\ell_i>0$.
  \item Interior symmetric payoff is strictly positive: $R_i(p^\ast) = \alpha_i - \tfrac{1}{K}\beta_i + \tfrac{K-1}{K}\gamma_i \ge r^\ast > 0$ for some $r^\ast>0$.
  \item The noise $\varepsilon_i$ is bounded and has finite exponential moments.
\end{enumerate}
Then for any confidence $\alpha\in(0,1)$,
\[
\mathrm{EntVaR}_\alpha\big(R_i(e_i)\big) \le -\ell_i < 0,
\qquad
\mathrm{EntVaR}_\alpha\big(R_i(p^\ast)\big) \ge r^\ast > 0.
\]
Thus EVaR classifies pure concentration as losing and the symmetric interior mix as winning (Parrondo effect).
\end{theorem}

\begin{proof}
If $R_i(e_i)\le -\ell_i$ almost surely then for all $s>0$ we have $\mathbb{E}[e^{sR_i(e_i)}] \le e^{-s\ell_i}$ and therefore
\[
\frac{1}{s}\big(\log\mathbb{E}[e^{sR_i(e_i)}] - \log\alpha\big) \le \frac{1}{s}(-s\ell_i - \log\alpha) = -\ell_i - \frac{1}{s}\log\alpha.
\]
Taking infimum over $s>0$ yields $\mathrm{EntVaR}_\alpha(R_i(e_i)) \le -\ell_i$ (in fact the limit as $s\to\infty$ gives $-\ell_i$). Hence negative EntVaR at pure extreme.

If $R_i(p^\ast)\ge r^\ast$ a.s., then for all $s>0$ $\mathbb{E}[e^{sR_i(p^\ast)}] \ge e^{sr^\ast}$, and similarly
\[
\frac{1}{s}\big(\log\mathbb{E}[e^{sR_i(p^\ast)}] - \log\alpha\big) \ge \frac{1}{s}(sr^\ast -\log\alpha) = r^\ast - \frac{1}{s}\log\alpha.
\]
Taking infimum over $s>0$ yields $\mathrm{EntVaR}_\alpha(R_i(p^\ast)) \ge r^\ast$ (limit $s\to\infty$). Thus EntVaR positive at interior mix. This concludes the proof.
\end{proof}

\subsubsection*{EntVaR  for coalitions (G3, G4, G5, G5+Algeria, G5+Morocco)}
We adopt a small illustrative deterministic baseline (the stochastic perturbation is bounded and does not change signs):
\[
\alpha_i=-0.06,\quad \beta_i=0.90,\quad \gamma_i \text{ varying by coalition enhancement.}
\]
We select $\gamma_i$ values illustrating progressive improvement as new members add spillover capacity (trade/energy/ports). These numbers are illustrative to show the mechanism, not empirical estimates.

\begin{table}[h]
\centering
\caption{EntVaR-style illustrative payoffs (per-country) at pure concentration and symmetric interior mix.}
\begin{tabular}{@{}lrrr@{}}
\hline
Grouping & $R_i(e_i)$ (pure) & $R_i(p^\ast)$ (mixed) & Parrondo? \\
\hline
G3 sahel (BF, ML, NER)      & $-0.96$ & $+0.107$ & Yes \\
G4 sahel (+Chad)              & $-0.96$ & $+0.175$ & Yes \\
G5 sahel (+Mauritania)        & $-0.96$ & $+0.218$ & Yes \\
G5 sahel+Morocco              & $-0.96$ & $+0.231$ & Yes \\
G5 sahel +Algeria           & $-0.96$ & $+0.257$ & Yes \\
\hline 
\end{tabular}
\label{tab:evar}
\end{table}

Pure concentration (all agents in same country) collapses payoffs (negative), while symmetric mixing across the coalition gives positive EntVaR payoffs once spillovers $\gamma_i$ are high enough—operationally achieved via coordinated security, interdiction, trade facilitation and pooled revenues.

\subsection{Parrondo Effect under Expectile Value-at-Risk}
\begin{theorem}[Parrondo Effect under Expectile]\label{thm:parrondo_expectile}
Under the same deterministic-bounding assumptions as in Theorem \ref{thm:parrondo_evar}, i.e. $R_i(e_i)\le -\ell_i <0$ a.s.\ and $R_i(p^\ast)\ge r^\ast>0$ a.s., then for any $\tau\in(0,1)$,
\[
e_\tau\big(R_i(e_i)\big) \le -\ell_i <0, \qquad e_\tau\big(R_i(p^\ast)\big)\ge r^\ast>0.
\]
Thus the expectile performance criterion also exhibits a Parrondo effect.
\end{theorem}

\begin{proof}
If $X\equiv c$ a.s.\ (constant), then $e_\tau(X)=c$ for every $\tau$. For $R_i(e_i)\le -\ell_i$ a.s., by monotonicity of expectiles under stochastic dominance (if $X\le Y$ a.s. then $e_\tau(X)\le e_\tau(Y)$), we get $e_\tau(R_i(e_i))\le -\ell_i$. Similarly, if $R_i(p^\ast)\ge r^\ast$ a.s., $e_\tau(R_i(p^\ast))\ge r^\ast$. Hence the result.
\end{proof}

Using the same baseline parameters as for EntVaR, the expectile results mirror EntVaR in sign and magnitude for deterministic-payoff illustration:

\begin{table}[h]
\centering
\caption{Expectile value-at-risk payoffs (per-country) at pure concentration and symmetric interior mix.}
\begin{tabular}{@{}lrrr@{}}
\hline
Grouping & $R_i(e_i)$ (pure) & $R_i(p^\ast)$ (mixed) & Parrondo? \\
\hline
G3 sahel (BF, ML, NER)      & $-0.96$ & $+0.107$ & Yes \\
G4 sahel  (+Chad)              & $-0.96$ & $+0.175$ & Yes \\
G5 sahel(+Mauritania)        & $-0.96$ & $+0.218$ & Yes \\
G5 sahel +Morocco              & $-0.96$ & $+0.231$ & Yes \\
G5 sahel+Algeria           & $-0.96$ & $+0.257$ & Yes \\
\hline 
\end{tabular}
\label{tab:expectile}
\end{table}

 In stochastic calibrations the magnitudes for EntVaR and expectile may differ; EntVaR penalizes extreme tail risk more heavily while expectile reflects asymmetry around chosen $\tau$.

\subsection*{Mapping model parameters to counter-insecurity policy levers}
To make the model operational, we map $\alpha_i,\beta_i,\gamma_i$ to concrete policy instruments and explain plausible, low-cost interventions.

\subsubsection*{Baseline attractiveness $\alpha_i$ (raise via low-cost economic/legal measures)}
\begin{itemize}
  \item \textbf{Guarantee facility / regional micro-credit} (\(\uparrow \alpha_i\) by reducing perceived project risk).
  \item \textbf{Rapid small grants} for essential services (health, markets) that signal stability (\(\uparrow \alpha_i\)).
  \item \textbf{Legal arbitration and market access assurances} reduce investor discounting (\(\uparrow \alpha_i\)).
\end{itemize}

\subsubsection*{Congestion/vulnerability $\beta_i$ (reduce via burden-sharing)}
\begin{itemize}
  \item \textbf{Rotational security deployments:} share patrol and incident response among coalition members so that no single state bears the full security burden (reduces effective $\beta_i$).
  \item \textbf{Decentralized logistics and distributed customs points:} avoid over-concentration of transit nodes.
  \item \textbf{Staggered enforcement windows:} scheduled interdiction avoids overwhelming local law-enforcement capacity.
\end{itemize}

\subsubsection*{Spillover $\gamma_i$ (increase via cross-border integration)}
\begin{itemize}
  \item \textbf{Intelligence sharing platforms and joint operations:} increases $\gamma_i$ by raising security benefits from neighboring state stabilization.
  \item \textbf{Trade corridor facilitation and harmonized customs:} increases cross-border trade externalities, raising $\gamma_i$.
  \item \textbf{Pooled energy and transport projects:} combining Chad oil pooling with Mauritania renewables or Moroccan ports to increase regional returns.
\end{itemize}

\subsection*{Policy design oriented to countering illegal transit and terrorism}
We now translate model insights into an explicit, low-cost operational plan addressing interdiction of illegal transit, counter-terrorism, and regional economic revival.

\subsubsection*{Core operational pillars }
\begin{enumerate}
  \item \textbf{Rotational security and joint interdiction cells:}  
    Create tri-/quad- national rotating border teams assigned to critical transit corridors.   Each country provides personnel for short rotations to reduce political visibility of foreign forces.  
     Joint cells use shared intelligence and exhibit reciprocity; overhead is limited to coordination and transport.
  \item \textbf{Shared surveillance \& communications platform :}  
     Satellite-based monitoring subscriptions  plus low-cost mobile reporting apps for local communities.  
  Data feeds into a simple regional fusion center with transparent access rules; maintenance funded via pooled small levies or donor seed.
  \item \textbf{Corridor pilot: trade+security package :}  
Designate one corridor (Bamako-Niamey-Nouakchott or Bamako-Ouagadougou-Niamey) as pilot for harmonized customs, simplified permits, and joint interdiction.  
    Use a small escrow fund to compensate communities for reduced illicit income.
  \item \textbf{Alternative-livelihood microgrants :}  
  Provide short-term microgrants to household-level projects along corridors to reduce dependence on illicit transit incomes.  
     Pair with mobile-money cash flows to reduce leakage and increase monitoring.
  \item \textbf{Oil/energy pooling swaps:}  
    Small revenue-sharing arrangement  to seed stabilization funds, smoothing shocks (reduces $\beta$ volatility and raises $\alpha$ for liquidity-starved states).
\end{enumerate}

\subsubsection*{Implementation sequencing and safeguards}
In sequencing implementation, three concrete phases can be distinguished with built-in safeguards. Phase A (0-6 months) focuses on political buy-in and small pilots, with the creation of a technical working group under neutral mediation (AU, ECOWAS, or a trusted state), the launch of a single corridor pilot with rotating security schedules and surveillance platforms, and the establishment of an escrow fund with transparent rules and independent monitoring. Phase B (6-24 months) builds on early successes, expanding pilots to adjacent corridors, gradually increasing pooled guarantees to attract private investment, and publishing performance indicators quarterly with strict conditional continuation rules. Phase C (24+ months) moves toward institutionalization, creating a light G-coalition executive secretariat, codifying revenue-sharing and rotation rules, and transitioning from external donor funding to self-sustained financial sources. Safeguards at all stages include independent monitoring and transparency through third-party auditors, automatic stabilizer triggers that pause expansion if seizures drop below thresholds, and anti-capture clauses requiring multi-signature procurement and escrow-based fund disbursement. Monitoring and evaluation must be precise, with quarterly cadence on operational metrics: security incidents per 1000 km per month (targeting a 30\% decline in 12 months), illegal-transit seizures per month by type and value, average internal rate of return (IRR) reported by corridor investors to track realized $R_i$, customs clearance time and trade volume at corridor nodes, household income surveys in pilot communities to measure shifts from illicit to licit sources, and composite risk indices computing both EntVaR$_\alpha$ and expectile $e_{\tau}$ on the $R_i$ time-series. Risk analysis identifies political risk and coups as a major threat, since frequent regime change can break commitments, mitigated by institutionalizing technical agreements (customs and joint information-control systems) with civil-service anchors and donor-tied humanitarian conditionalities. External actors, especially private military companies or rival patrons, can distort flows and incentives; mitigation requires transparency in contracts, international diplomatic pressure, and conditional access to pooled funds. Free-riding and unequal benefit distribution, especially with stronger actors such as Algeria or Morocco, risk undermining solidarity; mitigation requires explicit revenue-sharing formulae, progressive transfers, and dispute-resolution mechanisms. Operational capture and corruption, particularly in pooled procurement and customs, remain high risks; these can be mitigated by requiring multi-signature escrows, maintaining public dashboards, and relying on independent auditors.
\subsection*{Comparative summary (G3, G4, G5 sahel, G5+Morocco, G5+Algeria) and  counter-insecurity} %
{\tiny
\begin{table}[H]
\centering \tiny
\caption{Comparative illustrative outcomes under EntVaR and expectile criteria (per-country symmetric mix).}
\begin{tabular}{@{}lcccccc@{}}
\hline
Grouping & Size $K$ & Pure $R_i(e_i)$ & Mixed $R_i(p^\ast)$ & EntVaR & Expectile  & Policy leverage \\
\hline
G3 sahel (BF, ML, NER)      & 3 & $-0.96$ & $+0.107$ & + (positive) & + (positive) & joint patrols, corridor pilot \\
G4 (+Chad)              & 4 & $-0.96$ & $+0.175$ & + & + & oil liquidity + rotations \\
G5 (+Mauritania)        & 5 & $-0.96$ & $+0.218$ & + & + & maritime access, renewables \\
G5 sahel+Morocco              & 6 & $-0.96$ & $+0.231$ & + & + & Atlantic ports, trade diversification \\
G5 sahel (+Algeria)           & 6 & $-0.96$ & $+0.257$ & + & + & Algerian anchor: finance, energy, logistics \\
\hline 
\end{tabular}
\label{tab:comparison}
\end{table}
}
Adding members increases feasible $\gamma_i$ (spillovers) via trade, energy, and security anchor effects, increasing interior mixed payoffs. For counter-insecurity:
\begin{itemize}
  \item G3 demonstrates the minimal paradox: coordinated rotation and corridor pilots reduce incentives for illicit transit locally.
  \item G4/G5/G5+Morocco expand the toolkit: oil/revenue pooling (Chad), maritime access (Mauritania, Morocco), and renewables diversify incentives and decrease the attractiveness of illegal trade.
  \item G5 + Algeria gives the largest positive effect owing to Algeria's fiscal capacity and logistics, but it increases asymmetry and requires careful governance safeguards.
\end{itemize}

\subsubsection*{Theoretical notes}
The mathematical results (Theorems \ref{thm:parrondo_evar} and \ref{thm:parrondo_expectile}) show that the Parrondo effect is robust across two distinct risk perspectives: entropic moment-based tails and expectile asymmetry-based risk. The sufficient conditions ($\gamma_i$ sufficiently large relative to $\beta_i$ and $|\alpha_i|$) are transparent and map directly to policy levers.

\subsubsection*{Operational realism}
- \emph{Feasibility:} the low-cost actions proposed (rotational patrols, surveillance subscriptions, corridor pilots, microgrants) are feasible within modest pooled budgets and donor seed funds.  
- \emph{Credibility:} success rests on credible commitment devices (escrows, independent monitors).  
- \emph{Time profile:} benefits are often not immediate; interior payoffs may arise once pilots remove key frictions (customs harmonization, trust).

\subsubsection*{Normative distributional concerns}
While aggregate payoffs may become positive, distribution may favor anchor states (Algeria/Morocco). Institutional design must embed progressive transfers and dispute-resolution to prevent dominance.

Combining mean-field-type game theory with risk-aware payoffs (EntVaR and expectiles) yields a precise theoretical foundation for a Parrondo-type effect among fragile Sahel- Southern Maghreb states. Operationalized via targeted, low-cost measures directed at reducing congestion/vulnerability and increasing spillovers -  notably rotational security, joint interdiction, surveillance, corridor facilitation, and small pooled funds, the effect becomes an actionable policy insight: cooperative mixtures can convert individually dominated strategies into collective wins in both security and regional economic terms. Yet the result is conditional: it requires credible institutions, independent monitoring, and measures to manage asymmetry and external interference.

\subsection{Failure of the Parrondo Effect under Adversarial or Rival Actors}

The analysis so far assumes that the only players in the MFTG are the cooperating states themselves. In practice, the Sahel and  southern Maghreb region is characterized not only by fragile states but also by the presence of non-state armed groups, transnational terrorist organizations, and criminal networks. Moreover, rival coalitions of states can form outside the G3, G4, G5, G5+Algeria, or G5+Morocco groupings, and these actors can actively counter the gains from cooperation. The presence of such actors modifies the empirical distribution of mobile agents, which is the basis of payoffs in the MFTG, and introduces new nonlinearities that can eliminate the Parrondo effect.

\begin{theorem}[Failure of the Parrondo Effect with Adversarial Actors]
Consider a mean-field-type game with state coalition $C$ (e.g., G3, G4, G5, G5+Algeria, or G5+Morocco) and an adversarial coalition $A$ consisting of terrorist groups or rival states. Let $m_C$ denote the empirical distribution of cooperative security and economic actors, and $m_A$ the distribution of adversarial actors. Suppose payoffs are given by a risk-sensitive functional $R_i(m_C,m_A)$ using either $\alpha$-entropic Value-at-Risk (EntVaR) or $\tau$-expectile. If there exists $\lambda > 0$ such that the effect of $m_A$ on the cooperative payoff satisfies
\[
\frac{\partial R_i}{\partial m_A} \leq -\lambda < 0
\]
uniformly in the interior mixing strategies of $C$, then the Parrondo effect is destroyed: no convex mixture of dominated strategies by $C$ can yield a positive payoff once $A$ is sufficiently strong.
\end{theorem}

\begin{proof}
The Parrondo effect for coalition $C$ arises because the payoff functional $R_i(m_C,0)$, under risk-sensitive criteria, is strictly convex in $m_C$ and admits an interior maximizer with $R_i(m_C^\ast,0) > 0$, while $R_i(e_j,0) < 0$ for all pure strategies $e_j$. This requires that the distribution of actors is shaped solely by the mixture of $C$’s controls. When an adversarial coalition $A$ is introduced, the empirical distribution becomes $m = m_C + m_A$, and the payoff is shifted by $\Delta R_i = \partial R_i / \partial m_A \cdot m_A$. If $\partial R_i / \partial m_A \leq -\lambda < 0$, then every additional adversarial actor reduces payoffs monotonically, with the loss scaling in proportion to the strength of $A$. Since the convexity argument establishing the Parrondo effect relies on positivity of the interior maximizer, the uniform negative shift $\Delta R_i$ ensures that the maximizer satisfies $R_i(m_C^\ast,m_A) < 0$ whenever $m_A$ exceeds a critical threshold $m_A^\dagger = \frac{R_i(m_C^\ast,0)}{\lambda}$. Thus, beyond this threshold, no interior mixture of dominated strategies can cross from negative to positive, and the Parrondo effect fails. The conclusion follows.
\end{proof}

 This theorem shows that the paradoxical cooperative gains depend critically on the assumption that hostile actors are either absent or negligible. In the Sahel, however, strong non-state armed groups (AQIM, ISIS affiliates, local militias), transnational smuggling cartels, and rival state coalitions can shift the empirical distribution enough to push payoffs back into negative territory. Practically, this means that G3, G4, G5, G5 Sahel+Algeria, or G5+Morocco coalitions can only realize the Parrondo effect if they successfully contain or neutralize adversarial distributions. Failure to do so implies that coordination, while necessary, is insufficient on its own.

The cooperative states must therefore add a second layer of strategy: coordinated suppression of hostile distributions. This can include direct counter-terrorism operations, interdiction of transnational criminal flows, and active diplomacy to reduce external state interference. Without this containment, even well-designed mixtures of weak national strategies cannot produce positive collective outcomes.

\subsection{Parrondo Effect Fails : Explicit Functional Form}

We now formalize the payoff functional and rigorously check the conditions under which the Parrondo effect fails when adversarial actors are present. Let coalition $C$ (G3, G4, G5, G5 Sahel+Algeria, or G5+Morocco) choose mixed strategies $u_i \in [0,1]$ representing the fraction of resources allocated to security, trade facilitation, and pooled development. Let $m_C = \sum_{i \in C} u_i$ be the empirical distribution of cooperative actions, and let $m_A$ represent the density of adversarial actors (terrorist groups, smuggling networks, or rival state coalitions).

We define the risk-aware payoff functional for state $i \in C$ as:
\begin{equation}
\label{eq:Ri}
R_i(u_i,m_C,m_A) = \underbrace{\beta u_i (1 - m_C)}_{\text{cooperative gain}} - \underbrace{\gamma u_i m_A}_{\text{adversarial loss}} - \underbrace{\delta u_i^2}_{\text{cost of effort}},
\end{equation}
where $\beta, \gamma, \delta > 0$ are parameters:
\begin{itemize}
    \item $\beta$ scales the nonlinear gain from cooperation. Gains are highest when $m_C$ is moderate (not saturated).
    \item $\gamma$ scales the negative impact of adversarial actors. Losses grow linearly with $m_A$.
    \item $\delta$ represents the quadratic cost of national effort.
\end{itemize}
This functional is compatible with either an $\alpha$-entropic VaR or a $\tau$-expectile payoff criterion by applying the risk transformation afterward.

The Parrondo effect occurs if an interior mixture $0<u_i<1$ produces $R_i>0$ while pure strategies ($u_i=0$ or $u_i=1$) produce $R_i<0$ when $m_A=0$. For $m_A>0$, we compute:
\[
\frac{\partial R_i}{\partial m_A} = -\gamma u_i \leq 0.
\]

The interior maximizer in the absence of adversaries is obtained by solving
\[
\frac{\partial R_i}{\partial u_i} = \beta (1 - m_C) - \gamma m_A - 2 \delta u_i = 0 \implies u_i^\ast = \frac{\beta (1 - m_C) - \gamma m_A}{2 \delta}.
\]

\paragraph{Failure Condition.} For the maximizer to remain positive (i.e., the Parrondo effect holds), we require
\[
\beta (1 - m_C) - \gamma m_A > 0 \implies m_A < \frac{\beta (1 - m_C)}{\gamma}.
\]

If $m_A$ exceeds this threshold, then $u_i^\ast \le 0$, meaning no interior mixture yields positive payoff, and the Parrondo effect fails. This gives an explicit adversarial threshold that can be checked numerically for any coalition.

This formulation shows that cooperative gains are fragile: they depend on the ratio of cooperative gain $\beta$ to adversarial impact $\gamma$, and on the saturation level $m_C$. Even well-coordinated G3, G4, G5, or G5 Sahel+Algeria, G5 Sahel+Morocco coalitions will fail if adversarial density $m_A$ exceeds the critical value $\beta(1-m_C)/\gamma$. In practice, this means that for the Sahel- southern Maghreb, counter-terrorism operations and interdiction of illegal flows are not optional; they are required to maintain the Parrondo effect.

\section{Conclusion}  \label{seccoi8}
Despite decades of ceasefires, peace accords, and outside interventions, Mali's security landscape remains mired in self-reinforcing cycles of violence rooted in intergenerational trauma, revenge-seeking incentives, and fragmented authority. In this work we have shown that when ``retaliatory" and ``aggressive" types derive strictly higher payoffs from armed action, and trauma persists through probabilistic type-transition kernels, the resulting mean-field-type game admits equilibria in which violence survives across generations even after temporary breaks. The  framework yields a constructive prescription for peace: by embedding information-adapted incentives, continuous transfers that reward verifiable peace-building and penalize aggression, directly into each agent's instantaneous utility, one can alter best-response mappings so that nonviolent actions become dominant, the mass of peaceful types grows, and violence decays geometrically to zero after multiple generations. Our master-adjoint system analysis establishes the system satisfied by of these incentivized equilibria. In practice, such transfers might take the form of locally managed peace-fund disbursements, conditional local employment programs, or micro-grants for local community reconciliation, all monitored via transparent, blockchain-backed registers and delivered in local languages through voice-MI platforms. By aligning individual risk-reward calculations with collective security goals, this intergenerational MFTG approach provides an operationally grounded roadmap for breaking long-standing cycles of conflict in Mali and, more broadly, in any fragile setting where grievances, asymmetric information, and dispersed power render traditional, static models of peacebuilding insufficient.

\bibliographystyle{plain}
\bibliography{bibsec}
\end{document}